\documentclass[opre,nonblindrev]{informs3}

\OneAndAHalfSpacedXI %

\usepackage{xspace}

\usepackage{soul}
\usepackage{hyperref}
\usepackage{ccaption}
\usepackage{graphicx}
\usepackage{caption}
\usepackage{endnotes}
\usepackage[font={normalsize}]{subcaption}
\usepackage{mathtools}

\usepackage{bbm}

\usepackage{cleveref}

\newtheorem{observation}{Observation}

\usepackage{algorithm}
\newcolumntype{x}[1]{>{\centering\arraybackslash\hspace{0pt}}p{#1}}
\usepackage[noend]{algpseudocode}
\algnewcommand\algorithmicforeach{\textbf{for each}}
\algdef{S}[FOR]{ForEach}[1]{\algorithmicforeach\ #1\ \algorithmicdo}

\newcount\Comments  %
\newcommand{\kibitz}[2]{\ifnum\Comments=1{\color{#1}{#2}}\fi}

\usepackage{multirow}

\usepackage{soul}

\usepackage{tikz}
\usepackage{pgfpages}
\usepackage{pgfplots}
\usepackage{pgfplotstable}
\usetikzlibrary{arrows,decorations,patterns,trees,matrix,calc,shapes}

\usetikzlibrary{patterns}
\pgfplotsset{compat=newest}
\usetikzlibrary{shapes.geometric}

\pgfdeclarepatternformonly{light north east lines}{\pgfqpoint{-1pt}{-1pt}}{\pgfqpoint{10pt}{10pt}}{\pgfqpoint{9pt}{9pt}}%
{
  \pgfsetlinewidth{0.4pt}
  \pgfpathmoveto{\pgfqpoint{0pt}{0pt}}
  \pgfpathlineto{\pgfqpoint{9.1pt}{9.1pt}}
  \pgfusepath{stroke}
}

\pgfdeclarepatternformonly{light dots}{\pgfqpoint{-1pt}{-1pt}}{\pgfqpoint{4pt}{4pt}}{\pgfqpoint{3pt}{3pt}}%
{
  \pgfpathcircle{\pgfqpoint{0pt}{0pt}}{0.4pt}
  \pgfusepath{fill}
}

\pgfdeclarepatternformonly{light vertical lines}{\pgfqpoint{-1pt}{-1pt}}{\pgfqpoint{10pt}{10pt}}{\pgfqpoint{9pt}{9pt}}%
{
  \pgfsetlinewidth{0.4pt}
  \pgfpathmoveto{\pgfqpoint{0pt}{0pt}}
  \pgfpathlineto{\pgfqpoint{0pt}{9pt}}
  \pgfusepath{stroke}
}

\pgfdeclarepatternformonly{light horizontal lines}{\pgfqpoint{-1pt}{-1pt}}{\pgfqpoint{8pt}{8pt}}{\pgfqpoint{7pt}{7pt}}{
  \pgfsetlinewidth{0.5pt}
  \pgfpathmoveto{\pgfqpoint{0pt}{0pt}}
  \pgfpathlineto{\pgfqpoint{7pt}{0pt}}
  \pgfusepath{stroke}
}

\pgfdeclarepatternformonly{light south west lines}{\pgfqpoint{-1pt}{-1pt}}{\pgfqpoint{10pt}{10pt}}{\pgfqpoint{9pt}{9pt}}{
  \pgfsetlinewidth{0.4pt}
  \pgfpathmoveto{\pgfqpoint{9pt}{0pt}}
  \pgfpathlineto{\pgfqpoint{0pt}{9pt}}
  \pgfusepath{stroke}
}

\usetikzlibrary{decorations.markings}

\newcommand{\lcp}{\textrm{LCP}(q,A)}

\newcommand{\I}{\mathcal{I}}
\newcommand{\J}{\mathcal{J}}
\newcommand{\N}{\mathcal{N}} 
\newcommand{\HH}{\mathcal{H}}

\newcommand{\barI}{\bar{\mathcal{I}}}

\newcommand{\on}{\texttt{on}\xspace}
\newcommand{\off}{\texttt{off}\xspace}

\newcommand{\tp}{\textrm{Th}}
\newcommand{\rLL}{(\mathsf{L}\mathsf{L})}
\newcommand{\rHL}{(\mathsf{H}\mathsf{L})}
\newcommand{\rLH}{(\mathsf{L}\mathsf{H})}
\newcommand{\rHH}{(\mathsf{H}\mathsf{H})}

\newcommand{\e}{T}
\newcommand{\pe}{\pi_{\mathsf{ex}}(0)}

\newcommand{\R}{\mathsf{R}}

\newcommand{\LL}{\mathcal{L}}

\newcommand{\floor}[1]{\lfloor{#1}\rfloor}

\let\oldmarginpar\marginpar
\renewcommand{\marginpar}[1]{\oldmarginpar{\raggedright #1}}
\usepackage[showdeletions]{color-edits}
\addauthor[Zhe]{zl}{red}
\addauthor[Yanwei]{ys}{blue}
\addauthor[Chiwei]{cy}{cyan}

\usepackage{tikz}
\usepackage{pgfpages}
\usepackage{pgfplots}
\usepackage{pgfplotstable}
\usetikzlibrary{arrows,decorations,patterns,trees,matrix,calc,shapes}

\usepackage{amsmath}
\usepackage{xcolor}
\usepackage{booktabs}
\usepackage[shortlabels,inline]{enumitem}
\usepackage{hyperref}
\hypersetup{
	colorlinks = true,
	citecolor = blue,
        urlcolor = blue
}
\usepackage{cleveref}
\usepackage{tikz}

\makeatletter
\def\Hy@Warning#1{}
\makeatother

\usepackage{natbib}
\usepackage{comment}
\bibpunct[, ]{(}{)}{,}{a}{}{,}%
\def\bibfont{\small}%

\usepackage{chngcntr}

\usepackage{dsfont}
\usepackage{graphicx}
\usepackage{subcaption}

\counterwithin{subsection}{section}

\newtheorem{theorem}{Theorem}
\newtheorem{lemma}{Lemma}
\newtheorem{proposition}{Proposition}
\newtheorem{corollary}{Corollary}

\newtheorem{remark}{Remark}
\newtheorem{example}{Example}

\newtheorem{definition}{Definition}

\makeatletter
\g@addto@macro\normalsize{%
	\setlength\abovedisplayskip{6pt}
	\setlength\belowdisplayskip{6pt}
	\setlength\abovedisplayshortskip{4pt}
	\setlength\belowdisplayshortskip{5pt} 
}
\makeatother

\TheoremsNumberedThrough     %
\ECRepeatTheorems

\EquationsNumberedThrough    %

\MANUSCRIPTNO{MS-0001-1922.65}

\usepackage{xr}
\makeatletter

\newcommand*{\addFileDependency}[1]{%
\typeout{(#1)}%
\@addtofilelist{#1}
\IfFileExists{#1}{}{\typeout{No file #1.}}
}\makeatother

\begin{document}

\RUNAUTHOR{ Sun, Liu and Yan}

\RUNTITLE{On-Off Systems with Strategic Customers}

\TITLE{On-Off Systems with Strategic Customers}

\ARTICLEAUTHORS{

\AUTHOR{Yanwei Sun}

\AFF{Imperial College Business School, Imperial College London, \EMAIL{yanwei@imperial.ac.uk}}

\AUTHOR{Zhe Liu}

\AFF{Imperial College Business School, Imperial College London, \EMAIL{zhe.liu@imperial.ac.uk}}

\AUTHOR{Chiwei Yan}

\AFF{Department of Industrial Engineering and Operations Research,
University of California, Berkeley, \EMAIL{chiwei@berkeley.edu}}
} %

\ABSTRACT{%

Motivated by applications such as urban traffic control and make-to-order systems, we study a fluid model of a single-server, on-off system that can accommodate multiple queues. The server visits each queue in order: when a queue is served, it is “on”, and when the server is serving another queue or transitioning between queues, it is “off”. Customers arrive over time, observe the state of the system, and decide whether to join. We consider two regimes for the formation of the on and off durations. In the \emph{exogenous} setting, each queue’s on and off durations are predetermined. We explicitly characterize the equilibrium outcome in closed form and give a compact linear program to compute the optimal on-off durations that maximizes total reward collected from serving customers. In the \emph{endogenous} setting, the durations depend on customers’ joining decisions under an exhaustive service policy where the server never leaves a non-empty queue. We show that an optimal policy in this case extends service beyond the first clearance for \emph{at most} one queue. Using this property, we introduce a closed-form procedure that computes an optimal policy in no more than $2n$ steps for a system with $n$ queues.

}%


\maketitle

\section{Introduction}
\label{sec:intro}

In many service systems, the server operates in an on-off manner, alternating between active service (\on) and periods of inactivity (\off). 
For example, a ride-hailing electric vehicle (EV) appears \on when actively transporting customers and \off while charging. %
This on-off model also applies when multiple
queues share a single server: 
one queue's \off status can happen when the server is \on with another queue or transitioning between queues.
For example, at a traffic intersection, one direction of traffic moves (\on) while the others wait (\off). Under traditional traffic lights, these \on and \off intervals are predetermined, whereas at roundabouts or under adaptive controls, one direction may be served exhaustively before switching. %
Likewise, in a make-to-order manufacturing system with a single machine, only one product line can be processed at a time (\on), forcing the others to wait (\off) until the machine becomes available.

Customers, referring to arriving agents in the above examples, are often strategic: they join the system only if their expected utility, which decreases with waiting time, outweighs outside options. In addition, one customer's decision to join affects the system state and influences the utilities of subsequent arrivals. %
At a traffic intersection, drivers may reroute when facing potentially heavy congestion. In a make-to-order manufacturing system, if production lead times appear excessive, customers may switch to another manufacturer.

To capture this strategic interaction, we study a fluid model with multiple queues,  where both arrivals and services are modeled as continuous flows.
We consider two distinct sources of the \on-\off durations of each queue: \emph{exogenous} \on-\off durations that do not depend on customers’ joining strategies, and \emph{endogenous} \on-\off durations influenced by customers' joining strategies under an \textit{exhaustive service policy} (where the server never leaves a non-empty queue). 
To be specific,
\begin{enumerate}
    \item in the exogenous scenario, both the \on and \off durations of each queue are predetermined (by the planner) and remain fixed, regardless of the customers' joining strategies;
    \item in the endogenous scenario, the planner determines \textit{post-clearance durations}---additional time that the server serves a queue after it initially becomes empty. The resulting \on and \off durations of each queue then depend on customers' joining strategies.
\end{enumerate}
We mainly focus on maximizing the long-run average throughput, that is, the number of customers served per unit time. However, all our results extend naturally to maximizing the long-run average reward, where each queue can yield a different reward per served customer.

In the exogenous setting, the only externality at play is negative: when earlier arrivals at the focal queue choose to join, they increase the waiting time for later arrivals. 
For the endogenous setting, both negative and positive externalities arise: while early joiners still impose additional delays on those who come later, their presence ensures the queue remains non-empty, which in turn keeps the server engaged with the queue. 

Interestingly, in both settings, we find that at equilibrium, customers do \emph{not} follow a queue-length-thresholds joining strategy. 
This is in sharp contrast to the literature (see, e.g., \citealt{economou_2008_ORL_markovian,guo_2011_strategic_behaviors_markovian_vacation,manou_2024_MSOM_strategic_transportation}). 
Instead, customers' equilibrium joining behavior follows a pattern of \emph{herding cycles}, where a period in which all arriving customers join is followed by a period in which none do, and this sequence repeats indefinitely. The durations of joining and not-joining periods together precisely make up one full cycle of the system.

We completely characterize the unique equilibrium joining strategy as well as the \textit{equilibrium outcome}, i.e., the queueing dynamics at equilibrium, given any exogenous \on-\off durations. We find that, when varying the \on-\off durations, there are potentially nine distinct equilibrium outcomes. Roughly speaking, when the waiting patience is relatively low, the equilibrium outcome is exhaustive---that is, the server never leaves a non-empty queue at the equilibrium outcome. In contrast, when the waiting patience is relatively high, the equilibrium outcome is non-exhaustive because customers are willing to wait through multiple cycles. 
Building on the equilibrium results, we give a compact linear programming formulation to find optimal exogenous \on-\off durations.

In the \textit{endogenous} regime, we simplify the problem of characterizing the equilibrium outcome into a linear complementarity problem (LCP). The constraint coefficients of this LCP exhibit a $Z$-matrix structure (i.e., non-positive off-diagonal elements) and additional box constraints. We present a simple algorithm with completely closed-form updates that finds the unique equilibrium outcome in at most $n$ steps for a system with $n$ queues. 
Moreover, we show that 
at one of the optimal exhaustive service policies, with the possible exception of \textit{at most one queue}, the server immediately departs from a queue as soon as it becomes empty.

Exploiting this observation, we provide an algorithm again with all closed-form updates that computes an optimal exhaustive service policy in at most $2n$ iterations for a system with $n$ queues. %
Finally, when all queues have the same service rate, we show that an optimal outcome from the exogenous regimes can also be implemented by the optimal exhaustive service policy that requires significantly fewer parameters.

\smallskip

\noindent\textbf{Organization of the Paper.} 
\Cref{sec:relatedwork} reviews the related literature. In \Cref{sec:model}, we present the model setup. We first analyze the exogenous regime in \Cref{sec:exogenous_regime}, with the equilibrium analysis in \Cref{sec:equilibrium_analysis_fixed_duration} and the planner's optimal policy in \Cref{sec:opt_fixed_duration}. We then turn to the endogenous regime in \Cref{sec:exhaustive}, following a similar structure with the equilibrium analysis in \Cref{subsec:equilibrium_analysis} and the optimal policy in \Cref{sec:opt_ex}. The connection between exogenous and endogenous regimes is discussed in \Cref{subsec:connection}.
\Cref{sec:closing} concludes the paper. Auxiliary results are provided in Appendix \ref{app_sec:auxiliary_result}, and all proofs are collected in Appendix \ref{app:pf}. For clarity of exposition, the main text focuses on throughput maximization, with the extension to general reward maximization presented in Appendix \ref{app_sec:reward_maximization}.

\section{Related Work}
\label{sec:relatedwork}

The study of customers' decentralized (strategic) behavior in service systems was pioneered by \cite{naor_1969_regulation}, who analyzed a single-server system with an observable queue. Since then, this topic has drawn increasing attention.  %
We refer readers to the survey books by \cite{hassin_2003_book_queuegame_toqueue} and \cite{hassin_2016_book_rational_queueing} for comprehensive reviews. In particular, our work is most related to vacation systems and dynamic congestion games.

\smallskip

\noindent\textit{Vacation Systems.}  
On-off systems are also known as \emph{vacation queues} or \emph{queues with removable servers} \citep{tian_2006_book_vacation_systems}. When there are multiple queues with setup/switchover times/costs, the system is typically referred to as a \textit{polling system}, where the \off duration of a focal queue is equal to the \on durations of other queues plus their switchover times.
For the review on polling systems \textit{without} considering strategic customers, see \cite{boon_2011_polling_survey} and \cite{borst_2018_polling_survey}.
Our work encompasses both single and multiple queues.
To the best of our knowledge, no existing research has studied multi-queue vacation systems with strategic customers. A few papers have considered strategic behaviors in single-queue vacation systems, as discussed below.

Regarding exogenous durations, \citet{economou_2008_ORL_markovian} analyze a Markovian system where both \on and \off durations are exponentially distributed. \citet{economou_2022_EJOR_markovian_reneging} extend this framework to incorporate customer reneging behaviors under a Markovian setting. \citet{logothetis_2023_fluid} study the same problem with fluid arrivals and services but still assume exponentially distributed \on-\off durations.
Regarding the endogenous setting,
\citet{burnetas_2007_QS_vacation_setuptimes} consider a system in which the server adopts an exhaustive service policy, never taking a vacation until the queue is empty. When a new customer arrives, the server returns to work after an exponentially distributed setup time. \citet{guo_2011_strategic_behaviors_markovian_vacation} extend the above setting to allow a general threshold, i.e., the server returns once the queue length reaches a certain threshold. 
They also discuss how to set the threshold to maximize social welfare.
\citet{manou_2014_OR_transportation_station} and \citet{manou_2024_MSOM_strategic_transportation} study customers' joining decisions in the (aforementioned) public transportation contexts. They consider exhaustive scenarios, either due to customers abandoning if not boarding the immediately arriving bus or because buses have sufficient capacity. 
Our work uses a fluid model that is \textit{not} memoryless, necessitating a different approach of analysis. 
We also study both exogenous and endogenous \on-\off regimes.
Finally, unlike prior studies that primarily focus on social welfare maximization, our objective is to maximize throughput (reward).

\smallskip

\noindent\textit{Dynamic Congestion Games.}  
Our work is also closely related to dynamic congestion games or \emph{Nash flows over time}. It is a time-dependent extension of the static congestion game to
the continuous-time evolution of travelers (or particles) from sources to sinks in a capacitated network. This field has drawn increasing attention in recent years \citep{koch2011nash, olver_2022_FOCS_long_term_nashflows, cominetti_2022_OR_long_term_dynamic_congestion_game,olver_2023_FOCS_convergence_nashflows}. Our work resembles previous studies in how delays are captured using a similar fluid model (see the Vickrey bottleneck model \citep{vickrey1969congestion} in dynamic congestion game literature). However, those studies typically assume that intersections can process traffic from all directions simultaneously. 
In contrast, our study is the first to relax this assumption by examining a single intersection that only admits one direction of flow at a time, naturally giving rise to an \on-\off\ pattern on a focal path. Such a relaxation is particularly important in practice, as traffic congestion often originates from bottlenecks at intersections.

\section{Model and Preliminaries}
\label{sec:model}

We consider a fluid system consisting of \( n \geq 1 \) queues, denoted by the set \( \N := \{1, \dots, n\} \). 
Each queue follows a first-come-first-served discipline. There is one server and it can only serve one queue at a time. We assume the server operates under a \emph{non-idling} policy, meaning it never remains idle at a queue when there are customers waiting to be served. %
For each queue $i \in \N$, the server's status alternates between active (``\on'') and suspension (``\off''). 
An \off period of queue $i$ refers to any time when the server is not actively serving queue $i$. This includes periods when the server is attending to other queues or any switchover times required for transitioning between queues. When $n\geq 2$, we assume that the service order is fixed and the server visits each queue exactly once in one cycle. Then, without loss of generality,  we assume that the switchover time depends only on the departing queue, not the destination.
Let $\tau_i \geq 0$ denote the switchover time from queue $i$ to any other queue, and define $\tau = (\tau_i)_{i \in \N}$. We assume that $1^\top \tau > 0$.\footnote{The case where \( 1^\top \tau = 0 \), meaning there are no switchover times, is not particularly interesting because the optimal policy would involve the server constantly switching between queues, spending only an infinitesimally short time at each queue.}

For each queue $i \in \N$, infinitesimal customers arrive at rate $\lambda_i>0$ and are served at rate $\mu_i>0$. To ease the exposition, we assume $\mu_i>\lambda_i$ for all $i \in \N$ in the main body and defer the results of the other case to Appendix~\ref{app_sec:lambda_i_greater_mu_i}.
Let $\rho_i:=\lambda_i/\mu_i$ be the utilization rate of queue $i$ and $\rho=(\rho_i)_{i\in \N}$. Without loss of generality, we normalize customer utility for not joining the queue to be zero. 
Customers are sensitive to delay. Specifically, if a customer in queue $i$ has a waiting time $w$, her utility $u_i(w)$ is strictly decreasing in $w$, and $\lim_{w\to+\infty}u_i(w) <0$ and $u_i(0)>0$. 
As a consequence, a customer of queue $i$ joins only if her waiting time is smaller than $\theta_i := u_i^{-1}(0)$. We thus call $\theta_i$ \emph{waiting patience} of customers in queue $i$.

Upon arrival, a customer observes the \emph{state} of the system, denoted as \( s(t) = (s_i(t))_{i=1}^n \),\footnote{%
All our results remain valid if customers can only observe the state of their own queues. 
This is because the system is deterministic and observing $s_i$ allows one to infer the states of other queues given a strategy profile.} 
where \( s_i(t) \) represents the state of queue \( i \) at time \( t \). The state is defined as \( s_i(t) := (q_i(t), \iota_i(t), e_i(t)) \), where \( q_i(t) \geq 0 \) is the queue length of queue \( i \), \( \iota_i(t) \in \{0,1\} \) indicates whether the server is \off (\(0\)) or \on (\(1\)) for queue \( i \), and \( e_i(t) \) represents the time elapsed in the current \on-\off state \( \iota_i(t) \) at queue \( i \). Let $\mathbb{S}$ be the set of all possible states. Customers of queue $i$ adopt a symmetric joining strategy $f_i: \mathbb{S} \to [0,1]$, which maps the observed state to a joining probability.

We consider two regimes, which differ in whether customers' strategies impact \on-\off durations of the queues.

\begin{enumerate}
    \item \textit{Exogenous Regime}: 
    Both the \on and \off durations of each queue are predetermined by the planner %
    irrespective of customers’ joining strategies.
    Let $L_i$ and $\bar{L}_i$ be the \on and \off duration of queue $i$, respectively. 
    \item \textit{Endogenous Regime}: 
    The planner implements an \emph{exhaustive service policy} characterized by the \emph{post-clearance durations} $T = (T_i)_{i=1}^n$, where $T_i$ represents the additional time the server continues serving queue $i$ after it initially becomes empty. If $T_i=0$, the server leaves immediately upon clearing queue $i$. In this regime, the \on-\off durations of a queue depend on the customers’ joining strategies, as they determine how long the server actually stays at each queue. We restrict to exhaustive service policies since it is commonly assumed and studied in the literature (see, e.g., \citealt{burnetas_2007_QS_vacation_setuptimes,guo_2011_strategic_behaviors_markovian_vacation,manou_2014_OR_transportation_station,manou_2024_MSOM_strategic_transportation}).
\end{enumerate}

We now proceed to define customers' equilibrium strategies under both regimes. Let $f:=(f_i)_{i=1}^n$ be a strategy profile of customers of all queues. Denote by $W_i^{f}(s)$ the waiting time of a customer of queue $i$ arriving at state $s$ under strategy profile $f$.

\begin{definition}[Equilibrium]
\label{def:eq}
A joining strategy profile $f=(f_i)_{i=1}^n$ forms an equilibrium if and only if for any queue $i\in \N$ and system state $s \in \mathbb{S}$, 
\begin{align}
\nonumber
    f_i(s)=
    \begin{cases}
        1, & \text{if } W_i^{f}(s) < \theta_i, \\
        0, & \text{if } W_i^{f}(s) > \theta_i,
    \end{cases}
\end{align}
and $f_i(s)$ can take any value within $[0,1]$ if $W_i^{f}(s)=\theta_i$.%
\end{definition}

Our first main objective is to understand the customers’ equilibrium joining strategy and the corresponding
\textit{equilibrium outcome}, namely the queueing dynamics under the customers’ equilibrium strategy. Our second main objective is to characterize the planner’s optimal decisions in each regime in terms of maximizing the long-run average throughput, i.e., the number of customers served per unit time. All our results can be extended, with minor modifications, to settings with reward-maximization objectives, where the reward per customer differs by queues. To keep the exposition concise, we focus on throughput maximization throughout the main body.
The results and discussions about reward maximization are provided in Appendix \ref{app_sec:reward_maximization}.

\section{Exogenous Regime}
\label{sec:exogenous_regime}

In this section, we focus on the exogenous regime. First, we analyze the customers' equilibrium joining strategies in \Cref{sec:equilibrium_analysis_fixed_duration}. Based on this analysis, we then derive the optimal \on-\off durations in \Cref{sec:opt_fixed_duration}.

\subsection{Equilibrium Analysis}
\label{sec:equilibrium_analysis_fixed_duration}

We study a focal queue~$i \in \N$ and analyze the customers' joining behavior under the exogenous \on-\off durations $(L_i, \bar{L}_i)$.
With \on-\off durations fixed, the state of queue $i$, $s_i(t)$, is sufficient for queue-$i$ customers to make their joining decisions.
It also turns out to be more convenient to track the system state using the \emph{residual time} $r_i(t)$, defined as the time left in the current status $\iota_i(t)$, instead of the elapsed time $e_i(t)$. 
In this section, we thus represent the state of queue $i$ at time $t$ as $s_i(t) = (q_i(t), \iota_i(t), r_i(t))$.

For notational brevity, we omit the superscript $f$ in the waiting time function $W_i^{f}(\cdot)$ whenever it causes no ambiguity. 
Besides, we simply write $W_i(s_i(t))$ as $W_i(t)$ when the context is clear.
Additionally, we use $h(t^-):=\lim_{x\rightarrow t^-} h(x)$
and $h(t^+):=\lim_{x\rightarrow t^+} h(x)$ to denote the left- and right-limit of some real-valued function $h(\cdot)$ at $t$.

Our next result characterizes the waiting time at a given state.

\begin{lemma}[Waiting Time]
\label{lem:waiting_time}
For given \on-\off durations $(L_i,\bar{L}_i)$, a customer arriving at time $t$ with state $s_i(t)=(q_i(t),\iota_i(t),r_i(t))$ faces a waiting time
\begin{align}
    W_i(t) 
    &= \underbrace{\frac{q_i(t)}{\mu_i}}_{\text{\normalfont{\on} waiting time}} 
    \;+\; 
    \underbrace{\bigl(1-\iota_i(t)\bigr)\,r_i(t) \;+\; z_i(t)\,\bar{L}_i}_{\text{\normalfont{\off} waiting time}},
    \label{eq:waiting_time_fixed_duration}
\end{align}
where 
 \begin{equation}\label{eq:z_i(t)}
        z_i(t) = \left\lfloor \frac{q_i(t)/\mu_i + \iota_i(t) (L_i - r_i(t))}{L_i} \right\rfloor.
    \end{equation}
\end{lemma}

The waiting time consists of two components: the \on waiting time $q_i(t)/\mu_i$ measures the total waiting time while the server is serving the customers ahead in queue $i$, and the \off waiting time $(1-\iota_i(t))\,r_i(t) + z_i(t)\,\bar{L}_i$ captures the remaining \off time $r_i(t)$ if the current state is \off ($\iota_i(t)=0$) plus additional $z_i(t)$ number of \off durations $\bar{L}_i$.
The number $z_i(t)$, given by \eqref{eq:z_i(t)}, depends on the system state when the customer arrives. Specifically,
\begin{enumerate}
    \item If the current state is \off ($\iota_i(t)=0$), the customer will experience $q_i(t) / (\mu_iL_i)$ (possibly fractional) \on durations before being served. By the alternation of \on-\off periods, there will be $\floor{q_i(t) / (\mu_iL_i)}$ additional \off durations in between.
    \item 
    If the current state is \on ($\iota_i(t)=1$), the \on waiting time remains $q_i(t)/\mu_i$. To compute the \off waiting time, we construct a \textit{hypothetical} system that matches
    case~(i). Specifically, we let the hypothetical \on waiting time be extended to $q_i(t)/\mu_i + L_i - r_i(t)$, the hypothetical status be \off, and the hypothetical residual time be zero. From the perspective of a customer arriving at time $t$, there is no difference between the actual system and this hypothetical system. By case~(i), we thus obtain the \off waiting time $\floor{(q_i(t)/\mu_i + L_i - r_i(t))/L_i}$.

    As an illustration, consider the case where \( q_i(t)/\mu_i = 3.5 \), \( L_i = 1 \), and \( r_i(t) = 0.2 \) with $\iota_i(t)=1$. The remaining \on time can serve \( 0.2\mu_i \) customers, leaving \( q_i(t) - r_i(t)\mu_i = 3.3\mu_i \) customers still to be served. Completing service for these \( 3.3\mu_i \) customers requires 3.3 additional \on periods, which in turn necessitate 4 full \off periods. Our formula \eqref{eq:z_i(t)} confirms this, yielding \( z_i(t) = \floor{3.5 + 1 - 0.2} = 4 \), consistent with the hypothetical system calculations above. 
    \end{enumerate}

\begin{remark}[Unique Equilibrium Joining Strategy]
\label{remark:unique_equilibrium_fixed_duration}
Notice that the waiting time $W_i(t)$ in \eqref{eq:waiting_time_fixed_duration} does not depend on the joining strategies of other customers (though the evolution of the state 
depends on these decisions). Thus, substituting \eqref{eq:waiting_time_fixed_duration} into the equilibrium \Cref{def:eq}
immediately yields a unique equilibrium joining strategy.
\end{remark}

We now proceed to characterize the resulting \emph{equilibrium outcome}, i.e., the queueing dynamics under the equilibrium joining strategy. Such a characterization is crucial for deriving the long-run average performance, as shown later. 
For technical convenience, we restrict to \emph{periodic} outcomes for each queue $i \in \mathcal{N}$, where the system state repeats in cycles, i.e., $s_i(t) = s_i(t+L_i+\bar{L}_i)$ for all (large enough) $t$. We will show later that among all periodic outcomes, there exists an (essentially) unique equilibrium outcome.

Finding a system's equilibrium outcome is challenging, as it requires analyzing the evolution of a three-dimensional state in continuous time. Nevertheless, we will explicitly characterize it using several key waiting time properties, as established later. While multiple equilibrium outcomes may exist, they all result in the same throughput---referred to as an \textit{essentially unique} equilibrium outcome. Furthermore, we will show that, except for one edge case, the equilibrium outcomes are actually unique, not just essentially unique (notice that the equilibrium strategy is unique by \Cref{remark:unique_equilibrium_fixed_duration} while the equilibrium outcome is essentially unique).

In \textit{one cycle} of the equilibrium outcome under the \on-\off durations $(L_i,\bar{L}_i)$, let $J_i(L_i, \bar{L}_i)\geq 0$
and $\bar{J}_i(L_i, \bar{L}_i)\geq 0$ denote the customers' \emph{joining duration} and \emph{not-joining duration}, respectively, 
and let $T_i(L_i, \bar{L}_i)\geq 0$ 
denote the \emph{post-clearance duration}, i.e., the additional time that the server continues serving queue $i$ since it becomes empty for the first time.
For simplicity, we write them as $J_i, \bar{J}_i$ and $T_i$ when the underlying \on-\off durations are clear. 
The post-clearance period must be part of the joining period since during the post-clearance period, arriving customers have zero waiting time and will always join.
Besides, we have $J_i + \bar{J}_i = L_i + \bar{L}_i$ by definition.

Note that we do \textit{not} assume the joining and not-joining periods to be contiguous within a single cycle. In principle, within a cycle of length \( (L_i+\bar{L}_i) \), these periods could alternate multiple times while still satisfying the condition \( J_i + \bar{J}_i = L_i + \bar{L}_i \). However, the following result establishes that these periods are indeed contiguous (\Cref{lem:fixed_duration_waitingtime}~(i)) and further illustrates key properties of the waiting time.

\begin{lemma}
\label{lem:fixed_duration_waitingtime}
For given \on-\off durations $(L_i,\bar{L}_i)$, suppose $\bar{J}_i>0$ (so customers do not always join at equilibrium).
Let $t = 0$ be the beginning of a not-joining period.
Under the equilibrium outcome:
\begin{enumerate}
    \item For $t\in (0,\bar{J}_i)$, no customer joins the system; for $t\in (\bar{J}_i, L_i+\bar{L}_i)$, all customers join the system, i.e., the joining and not-joining periods at equilibrium are contiguous.
    \item The waiting time $W_i(t)$
    has an upward jump of amount $\bar{L}_i$ at $t = 0$, i.e., $ W_i(0^+) = W_i(0^-) + \bar{L}_i$.
    \item The waiting time $W_i(t)$ is continuous in $t \in (0,\, J_i + \bar{J}_i)$ and $W_i(\bar{J}_i)=\theta_i$.
    \item The derivative of the waiting time with respect to $t$ is: 
    \begin{align*}
        W_i^\prime(t) =
        \begin{cases}
            -1, &  t \in (0, \bar{J}_i), \\[5pt]
            -1 + \rho_i, & t \in (\bar{J}_i, J_i + \bar{J}_i - T_i), \\[5pt]
            0, &  t \in (J_i + \bar{J}_i - T_i, J_i + \bar{J}_i),
        \end{cases}
    \end{align*}
    where $\rho_i := \lambda_i / \mu_i$ and when $T_i>0$, the last $T_i$ units of time in the joining period must be a post-clearance period. 
\end{enumerate}
\end{lemma}

We do not specify the waiting time at the switching epochs from joining to not joining as they are measure-zero.
The upward jump in the waiting time is mainly
due to $z_i(t)$ defined in \eqref{eq:z_i(t)} whereby customer waiting incurs an additional \off duration $\bar{L}_i$ at certain points in time.
\Cref{lem:waiting_time} indicates that in one cycle of the equilibrium outcome, there is only one jump due to periodicity, and it must occur at the transition from joining to not joining. This explains \Cref{lem:fixed_duration_waitingtime}~(ii).
Besides, since the waiting time is continuous at the transition time from not joining to joining, we must have $W_i(\bar{J}_i)=\theta_i$ by the equilibrium constraints.
This explains \Cref{lem:fixed_duration_waitingtime}~(iii). %

\begin{figure*}[htbp]
  \centering
  \resizebox{.68\textwidth}{!}{\begin{tikzpicture}[font=\large, line width=1pt]

\draw (-0.55,6.8) node[left, rotate=90] {waiting time};

\draw[line width=1.6pt] (0,6) -- (3,3) --  (6,1) -- (7,1);
\draw[line width=1.6pt] (7,6) -- (10,3) -- (13,1) --(14,1);

\draw[thick,dashed] (0,0) -- (0,6.7);
\draw[thick,dashed] (3,0) -- (3,6.7);
\draw[thick,dashed] (7,0) -- (7,6.7);
\draw[thick,dashed] (10,0) -- (10,6.7);
\draw[thick,dashed] (14,0) -- (14,6.7);

\draw[thick,dashed] (6,1) -- (6,1.5);
\draw[thick,dashed] (13,1) -- (13,1.5);

\draw [to-to,blue] (0,0.2) --(3,0.2);
\draw (1.5,0.2) node[right, color=blue, above] {$\bar{J}_i$};

\draw [to-to,blue] (7,0.2) --(10,0.2);
\draw (8.5,0.2) node[right, color=blue, above] {$\bar{J}_i$};

\draw [to-to,blue] (3,0.2) --(7,0.2);
\draw (5,0.2) node[right, color=blue, above] {$J_i$};

\draw [to-to,blue] (10,0.2) --(14,0.2);
\draw (12,0.2) node[right, color=blue, above] {$J_i$};

\draw [to-to,blue] (6,1.4) --(7,1.4);
\draw (6.5,1.4) node[right, color=blue, above] {$T_i$};
\draw [to-to,blue] (13,1.4) --(14,1.4);
\draw (13.5,1.4) node[right, color=blue, above] {$T_i$};

\draw  (2.0,4.5) node {$-1$};
\draw  (9.0,4.5) node {$-1$};

\draw[above]  (5,2) node {$-1+\rho_i$};
\draw[above]  (12,2) node {$-1+\rho_i$};

\draw[thick,dashed] (0,3) -- (14,3);
\draw (0,3) node[left] {$\theta_i$};

\draw (0,6) node {\small{$\clubsuit$}};
\draw (7,6) node {\small{$\clubsuit$}};
\draw (7,1) node {\small{$\clubsuit$}};
\draw (14,1) node {\small{$\clubsuit$}};

\draw (3,3) node {$\diamondsuit$};
\draw (10,3) node {$\diamondsuit$};

\end{tikzpicture}}
  \caption{
  Illustration of the Waiting Time Dynamics under Exogenous \on-\off Durations with Strictly Positive Not-Joining Duration $\bar{J}_i$. 
  \textit{Note}. The amount of the upward jump at \small{$\clubsuit$}
is $\bar{L}_i$.
  The post-clearance duration $T_i$ can be zero.
In the other case where the not-joining duration is zero ($\bar{J}_i=0$), we have $W_i(\diamondsuit)<\theta_i$ while the pattern remains unchanged.
  }
\label{fig:waiting_time_fixed_duration}
\end{figure*}
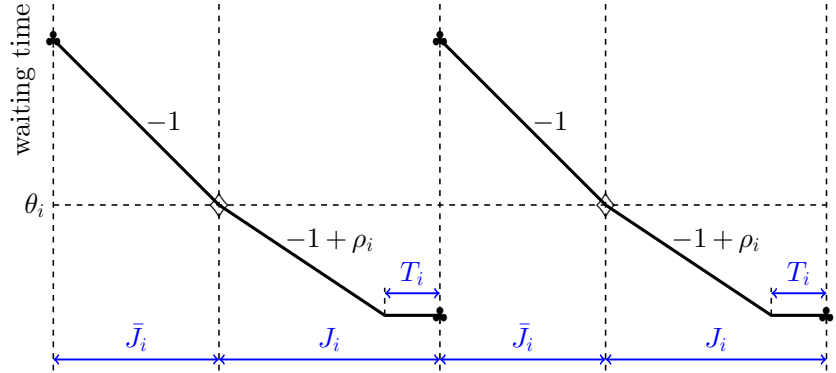

We now explain \Cref{lem:fixed_duration_waitingtime} (iv). For \( t \in (0, \bar{J}_i) \), i.e., during the not-joining period, consider a small \( \delta t > 0 \) such that \( t+\delta t \in (0, \bar{J}_i) \). The waiting time for a customer arriving at \( t + \delta t \) is given by 
$W_i(t + \delta t) = W_i(t) - \delta t$.
Since no one joins in between, this customer waits \( \delta t \) less than the one arriving at \( t \), leading to $W_i^\prime(t) = -1$. For $t \in (\bar{J}_i, L_i + \bar{L}_i - T_i)$, i.e., during the joining period but not the post-clearance period, the waiting time for a customer arriving at $t + \delta t$ is $  W_i(t + \delta t) = W_i(t) - \delta t + (\lambda_i/\mu_i) \delta t,$ where $(\lambda_i/\mu_i) \delta t$ 
reflects the additional waiting time caused by the $\lambda_i \delta t$ customers joining in between and served (at rate $\mu_i$) earlier.
This leads to $W_i^\prime(t) = -1 + \rho_i$.
For $t \in (L_i + \bar{L}_i - T_i,L_i + \bar{L}_i)$, i.e., during the post-clearance duration, since the server immediately serves all joining customers, the waiting time is always zero and so is its derivative.
A formal and detailed derivation can be found in the proof.

Let ``{\small{$\clubsuit$}}'' be the transition
time from joining to not joining and ``{\small{$\diamondsuit$}}'' the transition time from not joining to joining. 
\Cref{fig:waiting_time_fixed_duration} illustrates \Cref{lem:fixed_duration_waitingtime} with strictly positive not-joining duration $\bar{J}_i>0$.
We have $W_i(\diamondsuit)=\theta_i$ and there is a waiting-time jump at {\small{$\clubsuit$}}.

We call an equilibrium outcome \emph{exhaustive} if the server never leaves a non-empty queue. Otherwise, we say the equilibrium outcome is \emph{non-exhaustive}.
In what follows, we first discuss exhaustive equilibrium outcomes in \Cref{subsec:fixed_duration_exhaustive_outcome}. We then discuss non-exhaustive equilibrium outcomes in \Cref{subsec:fixed_duration_non_exhaustive_outcome}.
We summarize the two cases and provide further discussions in \Cref{subec:summary_fixed_duration}.

\subsubsection{Exhaustive Equilibrium Outcomes}
\label{subsec:fixed_duration_exhaustive_outcome}

In general, customers may be willing to join the system even if their service is delayed across multiple cycles, adding complexity to the equilibrium analysis.
In this section, we start from 
simpler cases where either the waiting patience is small enough so that customers do not wait for multiple cycles, or the waiting patience is large enough so that all customers are willing to join.
Under both situations, we will show that there exists a (essentially) unique equilibrium outcome that is exhaustive, i.e., the server never leaves a non-empty queue at the equilibrium outcome.

\begin{figure*}[ht]
\centering
  \subfloat[{\small{$\bar{L}_i\geq \theta_i$ and $L_i\geq \frac{\lambda_i\theta_i}{\mu_i-\lambda_i}$}} \label{subfig:exh_case1}]{  \resizebox{.49\textwidth}{!}{\begin{tikzpicture}[font=\large, line width=1pt] 

\draw (-3.4,6) node[left, rotate=90] {{\small{queue length}}};

\draw[line width=1.6pt](-3,1) -- (-2,1) -- (0,4)  -- (3,1) -- (4,1) -- (5,1) -- (7,4);

\draw[thick,dashed] (-3,0) -- (-3,6);
\draw[thick,dashed] (0,0) -- (0,6);
\draw[thick,dashed] (4,0) -- (4,6);
\draw[thick,dashed] (7,0) -- (7,6);

\draw [to-to,black] (-3,5) --(0,5);
\draw (-1.5, 5) node[right, color=black, above] {$\bar{L}_i$};
\draw [to-to,black] (0,5) --(4,5);
\draw (2,5) node[right, color=black, above] {$L_i$};
\draw [to-to,black] (4,5) --(7,5);
\draw (5.5, 5) node[right, color=black, above] {$\bar{L}_i$};

\draw	(0,4.2) node[anchor=west,red] {$\lambda_i\theta_i$};
\draw	(-3,1.2) node[anchor=east,red] {$0$};

\draw[thick,dotted] (-2,1) -- (-2,0);
\draw [thick,dotted] (5,1) --(5,0);

\draw [to-to,blue] (-2,0.2) --(4,0.2);
\draw (1,0.2) node[right, color=blue, above] {$J_i$};

\draw [to-to,blue] (-3,0.2) --(-2,0.2);
\draw (-2.5,0.2) node[right, color=blue, above] {$\bar{J}_i$};
\draw [to-to,blue] (4,0.2) --(5,0.2);
\draw (4.5,0.2) node[right, color=blue, above] {$\bar{J}_i$};

\draw[thick,dotted] (3,1) -- (3,2.2);
\draw [to-to,red] (3,2) --(4,2);
\draw (3.5,2) node[right, color=red, above] {$T_i$};

\draw [to-to,red] (-2,1) --(0,1);
\draw (-1,1) node[right, color=red, above] {$\theta_i$};

\draw [to-to,red] (5,1) --(7,1);
\draw (6,1) node[right, color=red, above] {$\theta_i$};

\draw [above] (5.8,2.5) node {\small{$\lambda_i$}};
\draw [above] (-1,2.8) node {\small{$\lambda_i$}};
\draw (2.1,2.9) node {\small{$-\mu_i+\lambda_i$}};

 \draw (-3,1) node {\small{$\clubsuit$}};
\draw (4,1) node {\small{$\clubsuit$}};

\draw (-2,1) node {$\diamondsuit$};
\draw (5,1) node {$\diamondsuit$};

\end{tikzpicture}}}
  \hfill
  \subfloat[{\small{$\bar{L}_i\geq \theta_i$ and $\frac{\lambda_i\theta_i}{\mu_i}\leq L_i<\frac{\lambda_i\theta_i}{\mu_i-\lambda_i}$}} \label{subfig:exh_case2}]{  \resizebox{.49\textwidth}{!}{\begin{tikzpicture}[font=\large, line width=1pt]

\draw[line width=1.6pt] (-3,1) -- (-2,1) -- (0,4) -- (3,3)-- (4,1) -- (5,1) -- (7,4);

\draw[thick,dashed] (0,0) -- (0,6);
\draw[thick,dashed] (-3,0) -- (-3,6);
\draw[thick,dashed] (4,0) -- (4,6);
\draw[thick,dashed] (7,0) -- (7,6);

\draw [to-to,black] (4,5) --(7,5);
\draw (5.5, 5) node[right, color=black, above] {$\bar{L}_i$};

\draw [to-to,black] (-3,5) --(0,5);
\draw (-1.5, 5) node[right, color=black, above] {$\bar{L}_i$};
\draw [to-to,black] (4,5) --(0,5);
\draw (2,5) node[right, color=black, above] {$L_i$};

\draw	(0,4.1) node[anchor=east,red] {$\lambda_i\theta_i$};
\draw	(-3,1.2) node[anchor=east,red] {$0$};

\draw[thick,dotted] (-2,1) -- (-2,0);
\draw[thick,dotted] (3,3) -- (3,0);
\draw[thick,dotted] (5,0) -- (5,1);

\draw [to-to,blue] (-2,0.2) --(3,0.2);
\draw (0.5,0.2) node[right, color=blue, above] {$J_i$};

\draw [to-to,blue] (3,0.2) --(5,0.2);
\draw (3.8,0.2) node[right, color=blue, above] {$\bar{J}_i$};

\draw [to-to,red] (5,1) --(7,1);
\draw (6,1) node[right, color=red, above] {$\theta_i$};
\draw [to-to,red] (-2,1) --(0,1);
\draw (-1,1) node[right, color=red, above] {$\theta_i$};

\draw [above] (-0.9,2.8) node {\small{$\lambda_i$}};
\draw [above]	(1.7,3.5) node {\small{$-\mu_i+\lambda_i$}};
\draw [above] (5.9,2.5) node {\small{$\lambda_i$}};
\draw (3.7,2.3) node {\small{$-\mu_i$}};

\draw (3,3) node {\small{$\clubsuit$}};

\draw (-2,1) node {$\diamondsuit$};
\draw (5,1) node {$\diamondsuit$};

\end{tikzpicture}}}

  \subfloat [{\small{$\bar{L}_i\geq \theta_i$ and $L_i<\frac{\lambda_i\theta_i}{\mu_i}$}} \label{subfig:exh_case3}]{  \resizebox{.49\textwidth}{!}{\begin{tikzpicture}[font=\large, line width=1pt] 

\draw (-3.4,6) node[left, rotate=90] {{\small{queue length}}};

\draw[line width=1.6pt](-3,1) -- (-2,1) -- (0,4) --(1,4) -- (3,1) -- (4,1) -- (6,4) --(7,4);

\draw[thick,dashed] (-3,0) -- (-3,6);
\draw[thick,dashed] (1,0) -- (1,6);
\draw[thick,dashed] (3,0) -- (3,6);
\draw[thick,dashed] (7,0) -- (7,6);

\draw [to-to,black] (-3,5) --(1,5);
\draw (-1, 5) node[right, color=black, above] {$\bar{L}_i$};
\draw [to-to,black] (1,5) --(3,5);
\draw (2,5) node[right, color=black, above] {$L_i$};
\draw [to-to,black] (3,5) --(7,5);
\draw (5, 5) node[right, color=black, above] {$\bar{L}_i$};

\draw	(1,4.2) node[anchor=west,red] {$\mu_i L_i$};
\draw	(-3,1.2) node[anchor=east,red] {$0$};

\draw[thick,dotted] (-2,1) -- (-2,0);
\draw[thick,dotted] (0,4) -- (0,0);
\draw[thick,dotted] (4,1) -- (4,0);
\draw [thick,dotted] (6,4) --(6,0);

\draw [to-to,blue] (-2,0.2) --(0,0.2);
\draw (-1,0.2) node[right, color=blue, above] {$J_i$};
\draw [to-to,blue] (4,0.2) --(6,0.2);
\draw (5,0.2) node[right, color=blue, above] {$J_i$};

\draw [to-to,blue] (0,0.2) --(4,0.2);
\draw (2,0.2) node[right, color=blue, above] {$\bar{J}_i$};

\draw [to-to,red] (-2,1) --(1,1);
\draw (-0.5,1) node[right, color=red, above] {$\theta_i$};
\draw [to-to,red] (4,1) --(7,1);
\draw (5.5,1) node[right, color=red, above] {$\theta_i$};

\draw [above] (5.3,3.2) node {\small{$\lambda_i$}};
\draw [above] (-1,2.8) node {\small{$\lambda_i$}};
\draw (2.2,2.9) node {\small{$-\mu_i$}};

\draw (6,4)  node {\small{$\clubsuit$}};
\draw  (0,4) node {\small{$\clubsuit$}};

\draw (-2,1) node {$\diamondsuit$};
\draw (4,1) node {$\diamondsuit$};

\end{tikzpicture}}}
    \hfill
\subfloat[{\small{$\bar{L}_i< \theta_i$ and $L_i\geq \frac{\lambda_i\bar{L}_i}{\mu_i-\lambda_i}$}} \label{subfig:exh_case4} ($\bar{J}_i = 0$)]{  \resizebox{.49\textwidth}{!}{\begin{tikzpicture}[font=\large, line width=1pt]

\draw[line width=1.6pt](-3,1) -- (0,4)  -- (3,1) -- (4,1) -- (7,4);

\draw[thick,dashed] (-3,0) -- (-3,6);
\draw[thick,dashed] (0,0) -- (0,6);
\draw[thick,dashed] (4,0) -- (4,6);
\draw[thick,dashed] (7,0) -- (7,6);

\draw [to-to,black] (-3,5) --(0,5);
\draw (-1.5, 5) node[right, color=black, above] {$\bar{L}_i$};
\draw [to-to,black] (0,5) --(4,5);
\draw (2,5) node[right, color=black, above] {$L_i$};
\draw [to-to,black] (4,5) --(7,5);
\draw (5.5, 5) node[right, color=black, above] {$\bar{L}_i$};

\draw	(0,4.2) node[anchor=west,red] {$\lambda_i\bar{L}_i$};
\draw	(-3,1.2) node[anchor=east,red] {$0$};

\draw [to-to,blue] (-3,0.2) --(4,0.2);
\draw (0.5,0.2) node[right, color=blue, above] {$J_i$};

\draw[thick,dotted] (3,1) -- (3,2.2);
\draw [to-to,red] (3,2) --(4,2);
\draw (3.5,2) node[right, color=red, above] {$T_i$};

\draw [above] (5.6,2.8) node {\small{$\lambda_i$}};
\draw [above] (-1.3,2.8) node {\small{$\lambda_i$}};
\draw (2.1,2.9) node {\small{$-\mu_i+\lambda_i$}};

\end{tikzpicture}}}

  \caption{Queueing Dynamics of Exhaustive Equilibrium Outcomes under \on-\off Durations $(L_i,\bar{L}_i)$.\\ 
  }
\label{fig:fixed_duration_exhaustive_equilibrium}
\end{figure*}

\begin{theorem}[Exhaustive Equilibrium Outcomes]
Given \on-\off durations $(L_i,\bar{L}_i)$: 
\label{thm:equilibrium_fixedduration_exhaustive}
\begin{enumerate}
    \item When $\bar{L}_i\geq \theta_i$, there exists a unique equilibrium outcome, which is exhaustive, as characterized in 
    Figures \ref{subfig:exh_case1}, \ref{subfig:exh_case2} and \ref{subfig:exh_case3} for different values of $L_i$.
    \item When $\bar{L}_i < \theta_i$ and $L_i > (\lambda_i \bar{L}_i)/(\mu_i-\lambda_i)$, there exists a unique equilibrium outcome, which is exhaustive, as characterized in \Cref{subfig:exh_case4}. 
    When $\bar{L}_i < \theta_i$ and $L_i = (\lambda_i \bar{L}_i)/(\mu_i-\lambda_i)$, there exists essentially unique equilibrium outcomes, and one of the equilibrium outcomes, which is exhaustive, is characterized in \Cref{subfig:exh_case4}.  %
\end{enumerate}    
\end{theorem}
The discussion of the remaining case $\bar{L}_i < \theta_i$ and $L_i < \lambda_i \bar{L}_i/(\mu_i-\lambda_i)$ will be deferred to \Cref{subsec:fixed_duration_non_exhaustive_outcome}.

\smallskip

When the \off\ duration \(\bar{L}_i\) exceeds the waiting patience \(\theta_i\), the resulting equilibria vary based on the specific \on-\off\ durations and model parameters, as illustrated in Figures~\ref{subfig:exh_case1}, \ref{subfig:exh_case2}, and \ref{subfig:exh_case3}. We have the following observations.
(1) Note that the equilibrium outcome must be exhaustive. This is because when 
$\bar{L}_i \geq \theta_i$, if there are any remaining customers when the server is leaving queue~$i$, their waiting time must exceed the patience $\theta_i$, which contradicts the equilibrium constraints.
This also implies that customers are deterred from joining the system at the moment the server is leaving queue~$i$. %
(2) The proportion of the joining duration relative to the cycle length depends on \(L_i\). When \(L_i\) is relatively large (\Cref{subfig:exh_case1}), all customers arriving after \(\diamondsuit\) (but before the server departs from queue~\(i\) again) choose to join, as they can be served before the server leaves. In this case, when queue~\(i\) becomes empty, the server continues serving, with the actual discharge rate matching the arrival rate \(\lambda_i\) (given our assumption that \(\lambda_i < \mu_i\) in the main text). For intermediate values of \(L_i\) (\Cref{subfig:exh_case2}), only a limited fraction of customers opt to join, whereas for small \(L_i\) (\Cref{subfig:exh_case3}), very few customers decide to join. (3) The durations of joining, not-joining, and potential post-clearance periods can be determined using the flow-balance equation—ensuring that the number of customers who join equals the number served—and geometric relationships.
As a result, the equilibrium outcome is completely characterized in \Cref{fig:fixed_duration_exhaustive_equilibrium}.
For instance, in \Cref{subfig:exh_case1}, we have \(\bar{J}_i = \bar{L}_i - \theta_i\), leading to \(J_i = L_i + \bar{L}_i - \bar{J}_i = L_i + \theta_i\). Additionally, the waiting time satisfies \(T_i = L_i - \lambda_i \theta_i/(\mu_i - \lambda_i)\). A similar reasoning applies to \Cref{subfig:exh_case2} and \Cref{subfig:exh_case3}. Since the post-clearance duration is zero in both cases, the service (discharge) rate remains \(\mu_i\). This implies \(\lambda_i J_i = \mu_i L_i\), yielding \(J_i = \mu_i L_i/\lambda_i\) and \(\bar{J}_i = L_i + \bar{L}_i - \mu_i L_i/\lambda_i\).

\smallskip

In the scenario where $\bar{L}_i < \theta_i$ and $L_i > \lambda_i \bar{L}_i / (\mu_i - \lambda_i)$ (\Cref{subfig:exh_case4}), the condition $L_i > \lambda_i \bar{L}_i / (\mu_i - \lambda_i)$ ensures that $\mu_i L_i > \lambda_i (L_i + \bar{L}_i)$, meaning the server has sufficient capacity to clear all customers if they choose to join. Additionally, the condition $\bar{L}_i < \theta_i$ suggests that customers exhibit relatively high patience, making them willing to join as long as the queue does not grow excessively. Given these two conditions, all customers will join the system, and queue~$i$ achieves the first-best throughput~$\lambda_i$.
A formal argument is provided in the proof in the appendix.

\smallskip

For the boundary case $L_i = \lambda_i \bar{L}_i/(\mu_i - \lambda_i)$ with $L_i < \theta_i$, it can be shown that all customers still join the system at equilibrium. This condition implies $\mu_i L_i =\lambda_i (L_i + \bar{L}_i)$, where $\lambda_i (L_i + \bar{L}_i)$ is the total number of customers joining the system in one cycle. By the flow-balance equation, we must have exactly $\mu_i L_i$ customers served each cycle. Thus, the post-clearance duration should be zero since otherwise the number of served customers $\mu_i\cdot (L_i - T_i) + \lambda_i T_i$ would be strictly less than $\mu_i L_i$.
The equilibrium outcome can be non-exhaustive; yet since all customers join, every equilibrium achieves the same throughput $\lambda_i$. For instance, if $\theta_i \to \infty$, one can arbitrarily shift the queueing dynamics in \Cref{subfig:exh_case4} (with $T_i=0$) upward without violating any equilibrium constraints. 
This is because, $L_i = \lambda_i \bar{L}_i/(\mu_i - \lambda_i)$ ensures the server can \textit{precisely} clear the queue if all customers join, maintaining system stability and ensuring a finite waiting time.

\subsubsection{Non-Exhaustive Equilibrium Outcomes}
\label{subsec:fixed_duration_non_exhaustive_outcome}

Now, we consider the remaining case where $\bar{L}_i < \theta_i$ and $L_i < \lambda_i \bar{L}_i / (\mu_i - \lambda_i)$, meaning customers have relatively high patience, but the server lacks the capacity to clear all customers if they all join. In this scenario, the equilibrium dynamics become more intricate, as some customers may still choose to enter the system despite facing multiple service cycles of delays. As we will show later, this results in non-exhaustive equilibrium outcomes.

We first show in \Cref{lem:fixed_duration_non_exhaustive_all_join_isnot_outcome} of the appendix that when $\bar{L}_i < \theta_i$ and $L_i < (\lambda_i \bar{L}_i)/(\mu_i - \lambda_i)$, the post-clearance duration is zero under any equilibrium outcomes and that all customers joining the system must not form an equilibrium, i.e., $\bar{J}_i>0$.
Thus, the waiting time dynamics follow the pattern depicted in \Cref{fig:waiting_time_fixed_duration} with $T_i=0$. %
To determine the equilibrium outcome, we must derive (i) the duration of joining \( J_i \) and not-joining \( \bar{J}_i \), and (ii) the precise timing of these periods.

\begin{figure*}[htbp]
\centering
 \subfloat[{\small Case 1}\label{subfig:nonexh_case1}]{  \resizebox{.47\textwidth}{!}{\begin{tikzpicture}[font=\large, line width=1pt] 

\draw (-3.4,6) node[left, rotate=90] {{\small{queue length}}};

\draw[thick,dashed] (0,0) -- (0,6);
\draw[thick,dashed] (-3,0) -- (-3,6);
\draw[thick,dashed] (4,0) -- (4,6);
\draw[thick,dashed] (7,0) -- (7,6);

\draw	(0,4.2) node[anchor=west,red] {$\Bar{q}_i$};
\draw	(-3,1.2) node[anchor=east,red] {$\underline{q}_i$};

\draw[line width=1.6pt] (-3,1) -- (-1,4) --(0,4) -- (1,2) -- (4,1) -- (6,4) --(7,4);

\draw [above]	(3,1.5) node {\small{$-\mu_i+\lambda_i$}};
\draw [above] (5.3,3.2) node {\small{$\lambda_i$}};
\draw [above] (-1.9,2.8) node {\small{$\lambda_i$}};
\draw (0.8,2.9) node {\small{$-\mu_i$}};

\draw [to-to,black] (4,5) --(7,5);
\draw (5.5, 5) node[right, color=black, above] {$\bar{L}_i$};

\draw [to-to,black] (-3,5) --(0,5);
\draw (-1.5, 5) node[right, color=black, above] {$\bar{L}_i$};
\draw [to-to,black] (4,5) --(0,5);
\draw (2,5) node[right, color=black, above] {$L_i$};

\draw[thick,dotted] (-1,4) -- (-1,0);
\draw[thick,dotted] (1,2) -- (1,0);
\draw [thick,dotted] (6,4) --(6,0);

\draw [to-to,red] (1,1.7) --(0,1.7);
\draw (0.5,1.7) node[right, color=red, above] {$\zeta_i$};

\draw [to-to,blue] (-1,0.2) --(1,0.2);
\draw (-0.2,0.2) node[right, color=blue, above] {$\bar{J}_i$};

\draw [to-to,blue] (6,0.2) --(1,0.2);
\draw (3.5,0.2) node[right, color=blue, above] {$J_i$};

\draw (-1,4) node {\small{$\clubsuit$}};
\draw (6,4) node {\small{$\clubsuit$}};

\draw (1,2) node {$\diamondsuit$};

\end{tikzpicture}}}
 \hfill
  \subfloat[{\small Case 2}\label{subfig:nonexh_case2}]{  \resizebox{.47\textwidth}{!}{\begin{tikzpicture}[font=\large, line width=1pt] 

\draw[thick,dashed] (1,0) -- (1,6);
\draw[thick,dashed] (-3,0) -- (-3,6);
\draw[thick,dashed] (3,0) -- (3,6);
\draw[thick,dashed] (7,0) -- (7,6);

\draw	(-3,4.2) node[anchor=east,red] {$\Bar{q}_i$};
\draw	(1.8,1.2) node[anchor=east,red] {$\underline{q}_i$};

\draw[line width=1.6pt] (-3,4) -- (-1,3) -- (0,1.5) -- (1,1) -- (3,4) -- (5,3) -- (6,1.5) -- (7,1);

\draw [above] (4.5,3.4) node {\small{$-\mu_i+\lambda_i$}};
\draw [above] (-1.5,3.4) node {\small{$-\mu_i+\lambda_i$}};
\draw [above] (-0.2,2) node {\small{$-\mu_i$}};
\draw [above] (5.8,2) node {\small{$-\mu_i$}};

\draw (1.6,2.6) node {\small{$\lambda_i$}};

\draw [to-to,black] (3,5) --(7,5);
\draw (5, 5) node[right, color=black, above] {$L_i$};

\draw [to-to,black] (-3,5) --(1,5);
\draw (-1, 5) node[right, color=black, above] {$L_i$};
\draw [to-to,black] (3,5) --(1,5);
\draw (2,5) node[right, color=black, above] {$\bar{L}_i$};

\draw[thick,dotted] (0,1.5) -- (0,0);
\draw[thick,dotted] (-1,3) -- (-1,0);
\draw [to-to,blue] (-1,0.2) --(0,0.2);
\draw (-0.5,0.2) node[right, color=blue, above] {$\bar{J}_i$};

\draw[thick,dotted] (6,1.5) -- (6,0);
\draw [to-to,blue] (6,0.2) --(5,0.2);
\draw (5.5,0.2) node[right, color=blue, above] {$\bar{J}_i$};

\draw[thick,dotted] (5,3) -- (5,0);
\draw [to-to,blue] (5,0.2) --(0,0.2);
\draw (2.5,0.2) node[right, color=blue, above] {$J_i$};

\draw [to-to,red] (-3,2.2) --(-1,2.2);
\draw (-2,2.2) node[right, color=red, above] {$\zeta_i$};
\draw [to-to,red] (3,2.2) --(5,2.2);
\draw (4,2.2) node[right, color=red, above] {$\zeta_i$};

\draw (-1,3) node {\small{$\clubsuit$}};
\draw (5,3) node {\small{$\clubsuit$}};

\draw (0,1.5) node {$\diamondsuit$};
\draw (6,1.5) node {$\diamondsuit$};

\end{tikzpicture}}}
\vspace{0.7em}
   \subfloat[{\small Case 3}\label{subfig:nonexh_case3}]{  \resizebox{.47\textwidth}{!}{\begin{tikzpicture}[font=\large, line width=1pt] 

\draw (-3.4,6) node[left, rotate=90] {{\small{queue length}}};

\draw[line width=1.6pt](-3,1) -- (-2,1) -- (0,4) --(1,4) -- (3,1) -- (4,1) -- (6,4) --(7,4);

\draw[thick,dashed] (-3,0) -- (-3,6);
\draw[thick,dashed] (1,0) -- (1,6);
\draw[thick,dashed] (3,0) -- (3,6);
\draw[thick,dashed] (7,0) -- (7,6);

\draw [to-to,black] (-3,5) --(1,5);
\draw (-1, 5) node[right, color=black, above] {$\bar{L}_i$};
\draw [to-to,black] (1,5) --(3,5);
\draw (2,5) node[right, color=black, above] {$L_i$};
\draw [to-to,black] (3,5) --(7,5);
\draw (5, 5) node[right, color=black, above] {$\bar{L}_i$};

\draw	(1,4.2) node[anchor=west,red] {$\Bar{q}_i$};
\draw	(-3,1.2) node[anchor=east,red] {$\underline{q}_i$};

\draw[thick,dotted] (-2,1) -- (-2,0);
\draw[thick,dotted] (0,4) -- (0,0);
\draw[thick,dotted] (4,2.2) -- (4,0);
\draw [thick,dotted] (6,4) --(6,0);
\draw [thick,dotted] (-2,2.2) --(-2,1);

\draw [to-to,red] (-3,2) --(-2,2);
\draw (-2.5,2) node[right, color=red, above] {$\zeta_i$};
\draw [to-to,red] (3,2) --(4,2);
\draw (3.5,2) node[right, color=red, above] {$\zeta_i$};

\draw [to-to,blue] (-2,0.2) --(0,0.2);
\draw (-1,0.2) node[right, color=blue, above] {$J_i$};
\draw [to-to,blue] (4,0.2) --(6,0.2);
\draw (5,0.2) node[right, color=blue, above] {$J_i$};

\draw [to-to,blue] (0,0.2) --(4,0.2);
\draw (2,0.2) node[right, color=blue, above] {$\bar{J}_i$};

\draw [above] (5.3,3.2) node {\small{$\lambda_i$}};
\draw [above] (-1,2.8) node {\small{$\lambda_i$}};
\draw (2.2,2.9) node {\small{$-\mu_i$}};

\draw (6,4) node {\small{$\clubsuit$}};
\draw (0,4) node {\small{$\clubsuit$}};

\draw (-2,1) node {$\diamondsuit$};
\draw (4,1) node {$\diamondsuit$};

\end{tikzpicture}}}
 \hfill
  \subfloat[{\small Case 4}\label{subfig:nonexh_case4}]{  \resizebox{.47\textwidth}{!}{\begin{tikzpicture}[font=\large, line width=1pt]

\draw[line width=1.6pt] (-3,1) -- (-2,1) -- (0,4) -- (3,3)-- (4,1) -- (5,1) -- (7,4);

\draw[thick,dashed] (0,0) -- (0,6);
\draw[thick,dashed] (-3,0) -- (-3,6);
\draw[thick,dashed] (4,0) -- (4,6);
\draw[thick,dashed] (7,0) -- (7,6);

\draw [to-to,black] (4,5) --(7,5);
\draw (5.5, 5) node[right, color=black, above] {$\bar{L}_i$};

\draw [to-to,black] (-3,5) --(0,5);
\draw (-1.5, 5) node[right, color=black, above] {$\bar{L}_i$};
\draw [to-to,black] (4,5) --(0,5);
\draw (2,5) node[right, color=black, above] {$L_i$};

\draw	(0,4.1) node[anchor=east,red] {$\Bar{q}_i$};
\draw	(-3,1.2) node[anchor=east,red] {$\underline{q}_i$};

\draw[thick,dotted] (-2,1) -- (-2,0);
\draw[thick,dotted] (3,3) -- (3,0);
\draw[thick,dotted] (5,0) -- (5,2.2);

\draw [to-to,blue] (-2,0.2) --(3,0.2);
\draw (0.5,0.2) node[right, color=blue, above] {$J_i$};

\draw [to-to,blue] (3,0.2) --(5,0.2);
\draw (3.8,0.2) node[right, color=blue, above] {$\bar{J}_i$};

\draw [thick,dotted] (-2,2.2) --(-2,1);
\draw [to-to,red] (-3,2) --(-2,2);
\draw (-2.5,2) node[right, color=red, above] {$\zeta_i$};
\draw [to-to,red] (4,2) --(5,2);
\draw (4.5,2) node[right, color=red, above] {$\zeta_i$};

\draw [above] (-0.9,2.8) node {\small{$\lambda_i$}};
\draw [above]	(1.7,3.5) node {\small{$-\mu_i+\lambda_i$}};
\draw [above] (5.9,2.5) node {\small{$\lambda_i$}};
\draw (3.7,2.3) node {\small{$-\mu_i$}};

\draw (3,3) node {\small{$\clubsuit$}};

\draw (5,1) node {$\diamondsuit$};
\draw (-2,1) node {$\diamondsuit$};

\end{tikzpicture}}}
\vspace{0.7em}
  \subfloat[{\small Case 5}\label{subfig:nonexh_case5}]{  \resizebox{.47\textwidth}{!}{\begin{tikzpicture}[font=\large, line width=1pt] 

\draw (-3.4,6) node[left, rotate=90] {{\small{queue length}}};

\draw[line width=1.6pt](-3,1) -- (-1,3) --(0,3) --(1,4) -- (3,1) -- (5,3) -- (6,3)--(7,4);

\draw[thick,dashed] (-3,0) -- (-3,6);
\draw[thick,dashed] (1,0) -- (1,6);
\draw[thick,dashed] (3,0) -- (3,6);
\draw[thick,dashed] (7,0) -- (7,6);

\draw [to-to,black] (-3,5) --(1,5);
\draw (-1, 5) node[right, color=black, above] {$\bar{L}_i$};
\draw [to-to,black] (1,5) --(3,5);
\draw (2,5) node[right, color=black, above] {$L_i$};
\draw [to-to,black] (3,5) --(7,5);
\draw (5, 5) node[right, color=black, above] {$\bar{L}_i$};

\draw	(1,4.2) node[anchor=west,red] {$\Bar{q}_i$};
\draw	(-3,1.2) node[anchor=east,red] {$\underline{q}_i$};

\draw[thick,dotted] (-1,3) -- (-1,0);
\draw[thick,dotted] (0,3) -- (0,0);
\draw[thick,dotted] (5,3) -- (5,0);
\draw [thick,dotted] (6,3) --(6,0);

\draw [to-to,red] (-3,1) --(-1,1);
\draw (-2,1) node[right, color=red, above] {$\zeta_i$};
\draw [to-to,red] (3,1) --(5,1);
\draw (4,1) node[right, color=red, above] {$\zeta_i$};

\draw [to-to,blue] (-1,0.2) --(0,0.2);
\draw (-0.5,0.2) node[right, color=blue, above] {$\bar{J}_i$};
\draw [to-to,blue] (5,0.2) --(6,0.2);
\draw (5.5,0.2) node[right, color=blue, above] {$\bar{J}_i$};
\draw [to-to,blue] (0,0.2) --(5,0.2);
\draw (2.5,0.2) node[right, color=blue, above] {$J_i$};

\draw [above] (-2.3,1.8) node {\small{$\lambda_i$}};
\draw [above] (0.4,3.5) node {\small{$\lambda_i$}};

\draw [above] (6.1,3.2) node {\small{$\lambda_i$}};
\draw [above] (3.9,2) node {\small{$\lambda_i$}};

\draw (2.3,3.3) node {\small{$-\mu_i+\lambda_i$}};

\draw (5,3) node {\small{$\clubsuit$}};
\draw (-1,3) node {\small{$\clubsuit$}};

\draw (0,3) node {$\diamondsuit$};
\draw (6,3) node {$\diamondsuit$};

\end{tikzpicture}}}
  
  \vspace{0.5em}
  \caption{Queueing Dynamics of Non-Exhaustive Equilibrium Outcomes under \on-\off Durations $(L_i,\bar{L}_i)$.
  }
\label{fig:fixed_duration_nonexhaustive_equilibrium}
\end{figure*}
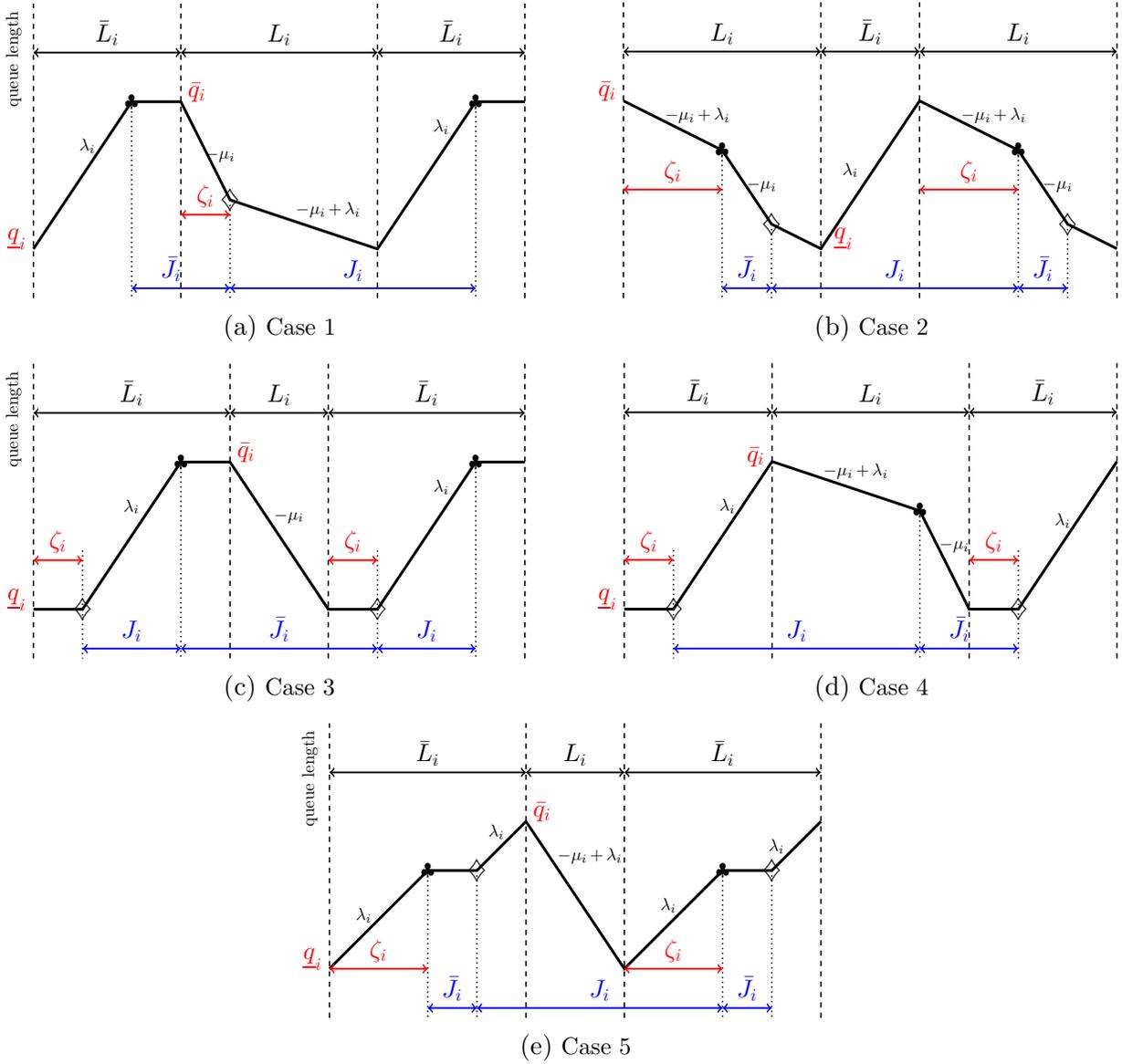
  
Regarding (i), since $T_i = 0$, the flow-balance equation becomes $\lambda_i J_i = \mu_i L_i$. This leads to $J_i = (\mu_i/\lambda_i) L_i$ and thus $\bar{J}_i = L_i + \bar{L}_i - J_i = \bar{L}_i - (\mu_i - \lambda_i) L_i/\lambda_i$.
Regarding (ii), there are five possible cases:
\begin{enumerate}[label=\emph{Case \arabic*}., leftmargin=*]
    \item The joining period starts during an \on period and ends during the subsequent \off period.
    \item The joining period starts during an \on period and ends during the next cycle's \on period. %
    \item The joining period starts during an \off period and ends during the same \off period.
    \item The joining period starts during an \off period and ends during the subsequent \on period.
    \item The joining period starts during an \off period and ends during the next cycle's \off period.
\end{enumerate}
Note that, given that $J_i = \mu_i L_i/\lambda_i  > L_i$, we can exclude the scenario where the joining period starts and ends within the same \on period. The queueing dynamics under these five cases are plotted in \Cref{fig:fixed_duration_nonexhaustive_equilibrium}.

In each case of \Cref{fig:fixed_duration_nonexhaustive_equilibrium}, to determine the equilibrium outcome, it is sufficient to know the following three quantities: the minimum queue length $\underline{q}_i$, the maximum queue length $\bar{q}_i$, and a variable $\zeta_i$ representing the duration between a switching epoch of the server's status to a switching epoch of the customers' joining decisions. %
Using simple geometry and the flow-balance equation, 
we can derive the maximum queue length $\bar{q}_i$ given $\underline{q}_i$ and $\zeta_i$. 
For example, in \Cref{subfig:nonexh_case1}, we have $\bar{q}_i = \underline{q}_i + \lambda_i \cdot ( \bar{L}_i - (\bar{J}_i - \zeta_i) )$. We can thus focus on determining $\underline{q}_i$ and $\zeta_i$.
Building on this insight, our next theorem explicitly characterizes the unique equilibrium outcome, which turns out to be non-exhaustive.

\begin{theorem}[Non-Exhaustive Equilibrium Outcomes]
\label{thm:fixed_duration_nonexhaustive_equilibrium}
Given any \on-\off durations $(L_i, \bar{L}_i)$ with $\bar{L}_i < \theta_i$ and $L_i < \lambda_i \bar{L}_i/(\mu_i - \lambda_i)$, there exists a unique equilibrium outcome, which is non-exhaustive. Let
$k_i := \floor{\theta_i/(L_i + \bar{L}_i)}$.
\begin{enumerate}
    \item If $ L_i \geq \rho_i \bar{L}_i$, the equilibrium outcome depends on which interval $k_i$ falls into:\footnote{The boundary values can belong to either case.}
    \smallskip
    \begin{center}
        \resizebox{0.7\textwidth}{!}{\begin{tikzpicture}
    \draw (1,0) -- (9,0);
  
  \foreach \x/\lab in {1/{$\frac{\theta_i}{L_i+\bar{L}_i}-1$}, 
                       3/{$\frac{\theta_i-\frac{\mu_i}{\lambda_i}L_i}{L_i+\bar{L}_i}$}, 
                       5/{$\frac{\theta_i-\bar{L}_i}{L_i+\bar{L}_i}$}, 
                       7/{$\frac{\theta_i-\frac{\mu_i-\lambda_i}{\lambda_i}L_i}{L_i+\bar{L}_i}$}, 
                       9/{$\frac{\theta_i}{L_i+\bar{L}_i}$}} {
    \draw (\x,0.1) -- (\x,-0.1) node[below] {\lab}; %
  }

  \node at (2,0.3) {\footnotesize{\textnormal{Case 1}}};
  \node at (4,0.3) {\footnotesize{\textnormal{Case 2}}};
  \node at (6,0.3) {\footnotesize{\textnormal{Case 4}}};
  \node at (8,0.3) {\footnotesize{\textnormal{Case 5}}};

\end{tikzpicture}}
    \end{center}
    \item If $ L_i < \rho_i\bar{L}_i$, again, the equilibrium outcome depends on which interval $k_i$ falls into:
    \smallskip
    \begin{center}
        \resizebox{0.7\textwidth}{!}{\begin{tikzpicture}
    \draw (1,0) -- (9,0);
  \foreach \x/\lab in {1/{$\frac{\theta_i}{L_i+\bar{L}_i}-1$}, 
                       3/{$\frac{\theta_i-\bar{L}_i}{L_i+\bar{L}_i}$}, 
                       5/{$\frac{\theta_i-\frac{\mu_i}{\lambda_i}L_i}{L_i+\bar{L}_i}$}, 
                       7/{$\frac{\theta_i-\frac{\mu_i-\lambda_i}{\lambda_i}L_i}{L_i+\bar{L}_i}$}, 
                       9/{$\frac{\theta_i}{L_i+\bar{L}_i}$}} {
    \draw (\x,0.1) -- (\x,-0.1) node[below] {\lab}; %
  }

  \node at (2,0.3) {\footnotesize{\textnormal{Case 1}}};
  \node at (4,0.3) {\footnotesize{\textnormal{Case 3}}};
  \node at (6,0.3) {\footnotesize{\textnormal{Case 4}}};
  \node at (8,0.3) {\footnotesize{\textnormal{Case 5}}};

\end{tikzpicture}}
    \end{center}
\end{enumerate}
The corresponding variables $\zeta_i$ and $\underline{q}_i$ in \Cref{fig:fixed_duration_nonexhaustive_equilibrium} are uniquely determined as follows:
\begin{enumerate}[label=Case \arabic*:, leftmargin=1.8cm]
    \item $\zeta_i = (k_i + 1)(L_i + \bar{L}_i) - \theta_i$ and $\underline{q}_i = k_i \mu_i L_i + \lambda_i (L_i - \zeta_i)$;
    \item $\zeta_i = k_i (L_i + \bar{L}_i) + (\mu_i/\lambda_i) L_i - \theta_i$ and $\underline{q}_i = k_i \mu_i L_i + \lambda_i \left( (\mu_i/\lambda_i) L_i - \bar{L}_i - \zeta_i \right)$;
    \item $\zeta_i = k_i (L_i + \bar{L}_i) + \bar{L}_i - \theta_i$ and $\underline{q}_i = k_i \mu_i L_i$;
    \item $\zeta_i = k_i (L_i + \bar{L}_i) + \bar{L}_i - \theta_i$ and $\underline{q}_i = k_i \mu_i L_i$;
    \item $\zeta_i = k_i (L_i + \bar{L}_i) + (\mu_i - \lambda_i)L_i/\lambda_i - \theta_i$ and $\underline{q}_i = k_i \mu_i L_i - \lambda_i \zeta_i$.
\end{enumerate}
\end{theorem}

The main conditions we use to determine the equilibrium outcome are: (i) $W_i(\diamondsuit) = \theta_i$, meaning the waiting time of customers arriving at the transition epoch $\diamondsuit$ from not joining to joining equals their patience, as established in \Cref{lem:fixed_duration_waitingtime} (iii); (ii) \emph{feasibility}: the duration of each line segment\footnote{The term \emph{line segment} refers to a specific linear portion of the piecewise linear queueing dynamics curve depicted in \Cref{fig:fixed_duration_nonexhaustive_equilibrium}.
The \emph{duration} of a line segment denotes the time interval over which this linear segment extends.} in any equilibrium case of \Cref{fig:fixed_duration_nonexhaustive_equilibrium} must be non-negative; and (iii) \emph{integer \off waiting periods}: the term inside the floor function of $z_i(\diamondsuit)$, as defined in \eqref{eq:z_i(t)}, must be an integer. The reasoning behind the third condition is that the term inside the floor function represents the additional (possibly fractional) number of full \off periods a customer arriving at time $\diamondsuit$ must wait through, excluding the current residual \off period if $\iota_i(\diamondsuit) = 0$. Since such a customer is served exactly at the start of an \on period (as shown in \Cref{lem:fixed_duration_non_exhaustive_joiningduration_waitingtime_integer} in the Appendix), the total number of full \off periods they wait through---apart from any ongoing residual \off period---must be an integer.

Note that the second condition (feasibility) and the third condition (integer property) naturally translate into inequality constraints. A priori, it is not obvious whether these inequalities guarantee the uniqueness of the equilibrium outcome. However, as shown in \Cref{thm:fixed_duration_nonexhaustive_equilibrium}, the interval defined by these constraints has a left boundary of $\theta_i / (L_i + \bar{L}_i) - 1$ and a right boundary of $\theta_i / (L_i + \bar{L}_i)$, with the length of the interval being \textit{exactly} one. As a result, there is exactly one integer within this interval, which uniquely determines the equilibrium.\footnote{In the edge case where $\theta_i / (L_i + \bar{L}_i)$ is an integer, $k_i$ can be either $\theta_i / (L_i + \bar{L}_i) - 1$ or $\theta_i / (L_i + \bar{L}_i)$, corresponding to Case 1 and Case 5. However, these two cases coincide in this scenario---\Cref{subfig:nonexh_case1} and \Cref{subfig:nonexh_case5} degenerate to the same queueing dynamics---so the equilibrium outcome remains unique.}

\subsubsection{Summary and Discussions}
\label{subec:summary_fixed_duration}

The results from \Cref{thm:equilibrium_fixedduration_exhaustive} and \Cref{thm:fixed_duration_nonexhaustive_equilibrium} together explicitly characterize the (essentially) unique equilibrium outcome. A noteworthy commonality across all equilibrium outcomes is the presence of \emph{herding cycles}: customers alternate between a joining period, during which all arriving customers join, and a subsequent not-joining period, during which no one joins. This cyclic pattern repeats indefinitely.

The cyclic herding pattern stands in sharp contrast to earlier studies where residual times are unobservable or memoryless.
In those settings, equilibrium strategies typically reduce to a phase-dependent queue-length threshold strategy
\citep[e.g.,][]{economou_2008_ORL_markovian,manou_2014_OR_transportation_station,manou_2024_MSOM_strategic_transportation}. Customers join if the queue length is below the threshold and do not join otherwise.
In other words,
customers' utility for joining the system with a shorter queue is higher than their utility for joining the system with a longer queue. 
Such behavior is dubbed \textit{avoid the crowd} (ATC) by \cite{hassin_1997_OR_follow_aviod_the_crowd}.
The opposite behavior that customers prefer a longer queue is referred to as \textit{follow the crowd} (FTC).

By contrast, in our setting, customers sometimes join when the queue is longer (FTC-type behavior, as in the \on period shown in \Cref{subfig:nonexh_case4}), and sometimes they prefer a shorter queue (ATC-type behavior, as illustrated in the \on period shown in \Cref{subfig:nonexh_case1}). 
The key driver of these nuanced behaviors arises from the non-Markovian properties of the system. 
This complexity leads to a different joining pattern compared with those
observed in prior work.

Another interesting observation from the equilibrium outcomes is that there may be no customer joining during the \on period, see \Cref{subfig:exh_case3} and \Cref{subfig:nonexh_case3}. Although the server is actively serving, any newly arriving customers would not be served within the current \on period. 
Instead, they must wait for potentially multiple cycles until the server returns to their queue, resulting in waiting times longer than their patience.

\subsection{Optimal Exogenous Durations}
\label{sec:opt_fixed_duration}

With the equilibrium outcomes derived for given \on-\off\ durations, we now turn to finding those that maximize throughput. In this section, we focus on the case \(n \geq 2\). The case \(n = 1\) is relatively straightforward in the absence of additional operational constraints, as the server can always serve the queue. We defer the more interesting case of \(n = 1\) with added operational constraints---motivated by electric vehicle operations---to Appendix~\ref{app_subsec:optimal_exo_n=1}.

Define $L=(L_i)_{i\in \N}$ and $\bar{L}=(\bar{L}_i)_{i\in \N}$. 
Given \on-\off durations $(L,\bar{L})$, in equilibrium, the system throughput is:
\begin{align*}
  \tp(L,\bar{L}) = \frac{\sum_{i\in\N}\left[\mu_i(L_i-T_i(L_i,\bar{L}_i))+\lambda_i T_i(L_i,\bar{L}_i)\right]}{L_1 + \bar{L}_1},
\end{align*}
where $T_i(L_i,\bar{L}_i)$ is the equilibrium post-clearance duration of queue $i$.
Here, $\lambda_i T_i(L_i,\bar{L}_i)$ is the number of customers served when queue $i$ is empty in one cycle. This is because the effective service rate during the post-clearance period is $\lambda_i$ rather than $\mu_i$ as all arriving customers are served immediately.
The term $\mu_i(L_i-T_i(L_i,\bar{L}_i))$ is the number of customers served when queue $i$ is not empty in one cycle.
The denominator is the length of one cycle. 
Recall that the \off duration of queue $i$ includes the \on durations of other queues and the switchover times, i.e., $\bar{L}_i=1^\top L_{-i}+1^\top \tau$, where $L_{-i}=(L_j)_{j\in \N\setminus\{i\}}$. This implies that $L_i + \bar{L}_i = L_j + \bar{L}_j =1^\top L + 1^\top \tau$ for all $i,j \in \N$.

As observed, given \on-\off durations, although there are nine different patterns of possible equilibrium outcomes, as shown by Figures \ref{fig:fixed_duration_exhaustive_equilibrium} and \ref{fig:fixed_duration_nonexhaustive_equilibrium}, to compute the throughput $\tp(L,\bar{L})$, it is sufficient to know the equilibrium post-clearance duration $T_i(L_i,\bar{L}_i),\forall i\in \N$.
Our next result provides a closed-form expression for the equilibrium post-clearance duration that allows us to analyze and optimize the long-run average throughput.

\begin{corollary}
\label{coro:fixed_duration_post_clearance_duration}
Given any \on-\off durations $(L_i,\bar{L}_i)$, under the equilibrium outcome, the post-clearance duration is given by:
\begin{align}
\label{eq:equilibrium_fd_hat_T_i} 
    T_i(L_i,\bar{L}_i) = \max\left\{\underbrace{L_i -\frac{\lambda_i}{\mu_i-\lambda_i}\theta_i}_{\textrm{$T_i$ in \Cref{subfig:exh_case1}}}, ~ \underbrace{ L_i - \frac{\lambda_i}{\mu_i-\lambda_i}\bar{L}_i}_{\textrm{$T_i$ in \Cref{subfig:exh_case4}}} ,~0\right\}  ,
\end{align}
where the first term corresponds to the post-clearance duration under the equilibrium outcome in \Cref{subfig:exh_case1}, and the second term corresponds to \Cref{subfig:exh_case4}.
\end{corollary}

Putting it all together, we have the throughput-maximization problem as follows:

\begin{subequations}
    \begin{align}
    \max_{L,\bar{L}\in \mathbb{R}^n_+}\quad &\frac{\sum_{i\in\N}\left[\mu_i(L_i-T_i(L_i,\bar{L}_i))+\lambda_i T_i(L_i,\bar{L}_i)\right]}{L_1 + \bar{L}_1} \tag{$\mathsf{P}$} \label{prob:fixed_duration}\\
    \textrm{s.t.}\quad &  \textrm{$T_i(L_i,\bar{L}_i)$ satisfies \eqref{eq:equilibrium_fd_hat_T_i}, \quad$\forall i\in\N$}, \label{constraint:hat_T_i_equilibrium_constraints} \\
     & \bar{L}_i = 1^\top L_{-i}+ 1^\top \tau ,\quad\forall i\in \N. \label{constraint:one_cycle}
\end{align}
\end{subequations}
Constraints \eqref{constraint:hat_T_i_equilibrium_constraints} impose equilibrium conditions. Constraints \eqref{constraint:one_cycle} come from the definition of the \off durations in the multi-queue setting.

The above problem formulation is not convex. Fortunately, we can reformulate it into a linear-fractional program, and then a linear program. To do that, first notice that \eqref{prob:fixed_duration} is equivalent to the following linear-fractional program \eqref{prob:lf_fixed_duration}:

\begin{align}
   \max_{\bar{L},L,T\in\mathbb{R}^n_{+}}\quad &  \frac{\mu^\top L - (\mu-\lambda)^\top T}{L_1+\bar{L}_1}  \tag{$\mathsf{LF}$}\label{prob:lf_fixed_duration}\\
\textrm{s.t.}\quad~ &    T_i\geq 0,\quad\forall i\in\N, \nonumber\\[2mm]
  & T_i \geq L_i - \frac{\lambda_i}{\mu_i-\lambda_i}\theta_i ,\quad\forall i\in\N, \nonumber \\
  & T_i \geq L_i - \frac{\lambda_i}{\mu_i-\lambda_i} \bar{L}_i ,\quad\forall i\in\N, \nonumber\\[1mm]
  & \eqref{constraint:one_cycle}. \nonumber
\end{align}
This is because the objective is decreasing in $T_i$ for all $i\in\N$. As a consequence,  
one of the first three constraints must be binding at optimality, equivalent to \eqref{constraint:hat_T_i_equilibrium_constraints}.

To further linearize \eqref{prob:lf_fixed_duration},  define auxiliary variables $g = 1/(L_1+\bar{L}_1)\in\mathbb{R}_+$, $x=L/(L_1 + \bar{L}_1)\in\mathbb{R}_+^n$, $\bar{x} = \bar{L}/(L_1+\bar{L}_1)\in\mathbb{R}_+^n$, and $y = T/(L_1+\bar{L}_1)\in\mathbb{R}_+^n$. Since $L_1+\bar{L}_1=L_1+ 1^\top L_{-1}+1^\top \tau>0$, problem \eqref{prob:lf_fixed_duration} is equivalent to the following linear program \eqref{prob:LP_fixed_duration} (see, e.g., Section 4.3.2 of \citealt{boyd_2004_convex_book}):
\begin{align}
\max_{x,\bar{x},y\in\mathbb{R}^n_{+},g\in\mathbb{R}_{+}}\quad &  
   \mu^\top x -(\mu-\lambda)^\top y\tag{$\mathsf{LP}$}\label{prob:LP_fixed_duration}\\
\textrm{s.t.}\qquad~ & y_i\geq 0,\quad \forall i\in\N, \nonumber\\[1mm]
  & y_i \geq x_i -  \frac{\lambda_i\theta_i}{\mu_i-\lambda_i} g ,\quad \forall i\in\N, \nonumber \\
  & y_i \geq x_i -  \frac{\lambda_i}{\mu_i-\lambda_i}\bar{x}_i,\quad \forall i\in\N, \nonumber \\[1mm]
   & x_i + \bar{x}_i = 1,\quad \forall i \in \N,\nonumber \\
    & \bar{x}_i = 1^\top x_{-i} + 1^\top \tau \cdot g ,\quad \forall i \in \N.\nonumber 
\end{align}
The (first three) equilibrium constraints in \eqref{prob:lf_fixed_duration} are equivalently translated into the first three linear inequalities in \eqref{prob:LP_fixed_duration}. 
The penultimate constraints are auxiliary constraints introduced from linearizing the fractional objective. 
The last constraints are the transformed constraints \eqref{constraint:one_cycle}.

After obtaining the optimal solutions of \eqref{prob:LP_fixed_duration}, $x^\ast,\bar{x}^\ast,y^\ast,g^\ast$, 
the optimal \on-\off durations $(L^\ast,\bar{L}^\ast)$ and the corresponding post-clearance durations $T_i(L_i^\ast,\bar{L}_i^\ast)$ can be recovered as:
\begin{align*}
    L_i^\ast = \frac{x_i^\ast }{ g^\ast}, \quad 
    \bar{L}_i^\ast =   \frac{\bar{x}_i^\ast}{g^\ast}, \quad  
    T_i(L_i^\ast,\bar{L}_i^\ast) = \frac{y_i^\ast}{g^\ast},\quad \forall i \in \N.
\end{align*}  

\begin{proposition}
\label{prop:LP_fixed_duration}
The throughput-maximizing \on-\off durations $(L^\ast,\bar{L}^\ast)$ can be found 
via solving the linear program \eqref{prob:LP_fixed_duration}. 
\end{proposition}

The formulation \eqref{prob:LP_fixed_duration} provides an efficient computational method for identifying the optimal \on-\off durations. To gain further structural insights, we now consider the special case in which the service rates are identical across all queues. In this scenario, we find that an \emph{optimal equilibrium outcome}, referred to as the equilibrium outcome under some optimal \on-\off durations, is exhaustive.%

\begin{theorem}
\label{thm:polling_identical_service_rates}
If all queues have the same service rate, 
an optimal equilibrium outcome of problem \eqref{prob:fixed_duration} is exhaustive. Furthermore, this optimal equilibrium outcome must be of the type shown in \Cref{subfig:exh_case1} or \Cref{subfig:exh_case4}. %
\end{theorem}

Besides, when the service rates are the same, we find that an optimal equilibrium post-clearance duration can sometimes be \textit{strictly} positive (without achieving the first-best throughput). 
Although having a positive post-clearance duration reduces the effective service rate for that queue, it also increases the cycle length, thereby diluting the negative impact of switchover times. As a result, the system can achieve a higher overall throughput.

One might wonder whether it is the case that the optimal equilibrium outcome is always exhaustive, given that clearing the queue should, intuitively, reduce waiting times and attract more customers. However, this intuition does \textit{not} always hold when service rates differ across queues. 
The reason is that if a particular queue has sufficiently patient customers, even a non-exhaustive queue can attract all arriving customers of this queue. 
On the other hand, a shorter \on duration in that queue can reduce its negative externalities on other queues with higher service rates. In turn, this makes other queues more appealing to their arriving customers, ultimately improving the overall throughput.

\section{Endogenous Regime: Exhaustive Service Policies}
\label{sec:exhaustive}

We now turn to the endogenous regime, where the server adopts an exhaustive service policy and the \on-\off durations of each queue depend on customers’ joining strategies. In this section, we assume $n \geq 2$.\footnote{The endogenous \on-\off durations under an exhaustive service policy for $n=1$ are considered in \cite{guo_2011_strategic_behaviors_markovian_vacation} under a Markovian system,
where the server remains off as long as the queue is empty and returns to service once the queue length reaches a certain threshold. Our fluid model can also incorporate this single-queue situation and allow the server to stay longer even when the queue is empty. However, since the resulting insights are similar to theirs, we omit the case of $n=1$.
Specifically, when the server is off the queue, arriving customers tend to follow the crowd; in contrast, when the server is actively serving the queue, customers tend to avoid the crowd.} 
The \off duration of queue $i$ is endogenously determined by the service times of the other queues and the associated switchover times.

As mentioned earlier, an exhaustive service policy is characterized by the post-clearance durations $T = (T_i)_{i\in\N}$, where $T_i \geq 0$ is the \emph{additional} time the server stays at queue $i$ after it \emph{first} becomes empty in a cycle. %
In the special case $T = 0$, we obtain the \emph{pure exhaustive service policy} denoted by $\pe$, 
where the server leaves \emph{immediately} once the queue becomes empty. In contrast to the exogenous regime in \Cref{sec:exogenous_regime}, under exhaustive service policies, the post-clearance durations $T$ are directly controlled by the planner, while the \on and \off durations now
depend on customers’ joining strategies. Let $L_i(T)$ and $\bar{L}_i(T)$ be the \on and \off durations of queue $i$ in one cycle at equilibrium under post-clearance durations $T$, respectively. 
We have $\bar{L}_i(T) = 1^\top \tau + 1^\top L_{-i}(T)$, where $L_{-i}(T)=(L_j(T))_{j\neq i}$ represent the \on durations of other queues and $\tau=(\tau_i)_{i\in \N}$ is the vector of switchover times.
This equation says that the \off duration of queue $i$ equals the sum of all switchover times plus the total \on times of the other queues.
Similarly, let $J_i(T)$ and $\bar{J}_i(T)$ denote the customers’ joining duration and not-joining duration in one cycle at equilibrium under post-clearance durations $T$. When there is no confusion, we may omit the explicit dependence on $T$ for ease of notation.

In what follows, we first characterize customers’ equilibrium joining strategies in Section \ref{subsec:equilibrium_analysis}, and then determine the post-clearance durations $T$ that maximize throughput in Section \ref{sec:opt_ex}. Finally, in \Cref{subsec:connection}, we make the connection with the exogenous regime studied in \Cref{sec:exogenous_regime}.

\smallskip

\noindent\textit{Notation and Glossary.}
When we say $y(x)\in\mathbb{R}^n$ is increasing with $x\in\mathbb{R}^m$, we mean that for any $i\in[n]$, $y_i(x)$ is increasing with $x_j$ for all $j\in [m]$ given $x_{-j}$.
 For $q\in\mathbb{R}^n$ and $M\in\mathbb{R}^{n\times n}$, $\textsf{LCP}$($q,M$) 
 represents the following linear complementarity problem:
 \begin{align}
 \label{prob:lcp}
\text{Find}\quad &z \in \mathbb{R}^{n} \tag{\textsf{LCP}$(q,M)$} \\
\text{s.t.}\quad& z  \geq 0,~ q+Mz\geq0,  \nonumber \\
& z^\top(q+Mz)  = 0\nonumber.
\end{align}
Observe that, since $z\geq 0$ and $q+Mz\geq 0$, the last complementarity constraint is equivalent to $z_i\cdot (q+Mz)_i=0$ for all $i\in [n]$.

\subsection{Equilibrium Analysis}
\label{subsec:equilibrium_analysis}

The following result characterizes the structure of the equilibrium outcome under an exhaustive service policy. \Cref{fig:exhaustiveservice_equilibrium} below is an illustration.

\begin{lemma}[Equilibrium Structure]
\label{lem:exhaustive_equilibrium_structure}
At an equilibrium outcome under an exhaustive service policy, the following hold for any queue $i \in \N$ during each cycle:
\begin{enumerate}
    \item The beginning of the not-joining duration of queue $i$ (marked by $\clubsuit$)
    must coincide with the time epoch when the server leaves queue $i$.
    \item The not-joining period of queue $i$
    must be 
    a contiguous interval contained in the \off period of queue $i$, i.e., it is \emph{not} possible that customers start to join during an \on period.
\end{enumerate}
\end{lemma}

\begin{figure*}[htbp]
  \centering
  \resizebox{.50\textwidth}{!}{\begin{tikzpicture}[font=\large, line width=1pt] 

\draw (-3.4,6) node[left, rotate=90] {{\small{queue length}}};

\draw[line width=1.6pt](-3,1) -- (-2,1) -- (0,4)  -- (3,1) -- (4,1) -- (5,1) -- (7,4);

\draw[thick,dashed] (-3,0) -- (-3,6);
\draw[thick,dashed] (0,0) -- (0,6);
\draw[thick,dashed] (4,0) -- (4,6);
\draw[thick,dashed] (7,0) -- (7,6);

\draw [to-to,red] (-3,5) --(0,5);
\draw (-1.5, 5) node[right, color=red, above] {$\bar{L}_i(T)$};
\draw [to-to,red] (0,5) --(4,5);
\draw (2,5) node[right, color=red, above] {$L_i(T)$};
\draw [to-to,red] (4,5) --(7,5);
\draw (5.5, 5) node[right, color=red, above] {$\bar{L}_i(T)$};

\draw	(0,4.2) node[anchor=west,red] {\small{$\lambda_i\cdot \alpha_i(T) \cdot \theta_i$}};
\draw	(-3,1.2) node[anchor=east,red] {$0$};

\draw[thick,dotted] (-2,1) -- (-2,0);
\draw [thick,dotted] (5,1) --(5,0);

\draw [to-to,blue] (-2,0.2) --(4,0.2);
\draw (1,0.2) node[right, color=blue, above] {$J_i$};

\draw [to-to,blue] (-3,0.2) --(-2,0.2);
\draw (-2.5,0.2) node[right, color=blue, above] {$\bar{J}_i$};
\draw [to-to,blue] (4,0.2) --(5,0.2);
\draw (4.5,0.2) node[right, color=blue, above] {$\bar{J}_i$};

\draw[thick,dotted] (3,1) -- (3,2.2);
\draw [to-to] (3,2) --(4,2);
\draw (3.5,2) node[right, above] {$T_i$};

\draw [to-to,red] (-2,1) --(0,1);
\draw (-1,1) node[right, color=red, above] {\small{$\alpha_i(T)\cdot \theta_i$}};

\draw [to-to,red] (5,1) --(7,1);
\draw (6,1) node[right, color=red, above] {\small{$\alpha_i(T)\cdot \theta_i$}};

\draw[thick,dotted] (-2.2,1) -- (-2.2,2.2);
\draw [to-to] (-3,2) --(-2.2,2);
\draw (-2.6,2) node[right, above] {$\tau_i$};

\draw [above] (5.8,2.5) node {\small{$\lambda_i$}};
\draw [above] (-1,2.8) node {\small{$\lambda_i$}};
\draw (2.1,2.9) node {\small{$-\mu_i+\lambda_i$}};

 \draw (-3,1) node {\small{$\clubsuit$}};
\draw (4,1) node {\small{$\clubsuit$}};

\draw (-2,1) node {$\diamondsuit$};
\draw (5,1) node {$\diamondsuit$};

\end{tikzpicture}}
  \caption{Queueing Dynamics of Queue $i$ under the Exhaustive Service Policy $\e$. %
  \emph{Note}. The switchover time $\tau_i$ can be larger than $\bar{J}_i$. 
  Besides, the not-joining duration $\bar{J}_i$ can be zero, in which case $\diamondsuit$ coincides with $\clubsuit$.
  }
\label{fig:exhaustiveservice_equilibrium}
\end{figure*}

\Cref{lem:exhaustive_equilibrium_structure} implies that the pattern of herding cycles is retained,
as plotted in \Cref{fig:exhaustiveservice_equilibrium}.
It can be shown that the waiting time decreases from ${\clubsuit}$ to the next cycle’s ${\clubsuit}$. 
(The waiting time pattern is the same as the exogenous case, see \Cref{fig:waiting_time_fixed_duration}.)
This implies that when the server is \off, a higher queue length leads to a lower waiting time. In other words, FTC behavior dominates during the \off phase: a longer queue signals earlier server activation, reducing the waiting time for arriving customers and encouraging them to join.
Conversely, when the server is \on, a lower queue length leads to a lower waiting time, indicating 
ATC behavior.

We denote the duration from the epoch 
when the first customer joins during the \off period
(marked by $\diamondsuit$ in \Cref{fig:exhaustiveservice_equilibrium})
to the epoch when the server returns to queue $i$ as $\alpha_i(T)\cdot \theta_i$ for some $\alpha_i(T)\geq 0$.
We claim that knowing $\alpha(T)=(\alpha_i(T))_{i=1}^n$ is sufficient to determine the equilibrium outcome. To see this, %
note that the maximum queue length is $\lambda_i \cdot \alpha_i(T) \cdot \theta_i$, and the \on duration of queue $i$ is
\begin{align}
\label{eq:exhaustiveservice_on_duration_alpha_j_theta_j}
    L_i(T) = \frac{\lambda_i \cdot \alpha_i(T) \cdot \theta_i}{\mu_i-\lambda_i} + T_i,
\end{align}
while the \off duration of queue $i$ is
\begin{align}
\label{eq:exhaustiveservice_off_duration_alpha_j_theta_j}
  \bar{L}_i(T) = 1^\top \tau + 1^\top L_{-i}(T)
  = 1^\top \tau + \sum_{j\neq i} \left(\frac{\lambda_j \cdot \alpha_j(T) \cdot \theta_j}{\mu_j-\lambda_j} + T_j\right).
\end{align}
Since $\alpha(T)$ is sufficient to determine the equilibrium outcome, we refer to $\alpha(T)$ as the \textit{equilibrium variable} or simply the \textit{equilibrium} of the exhaustive service policy $\e$. When there is no confusion, we drop the dependence on $T$ for ease of notation.

We now describe how to solve for $\alpha$.
First, to ensure that the customer arriving at
$\diamondsuit$ joins, we must have $\alpha_i \in [0,1]$. We then have the following conditions depending on whether $\bar{J}_i$ is zero or not:
\begin{align}
\text{either } \bar{J}_i > 0:\quad &\alpha_i = 1,\quad \text{$\alpha_i\theta_i$ is strictly less than $\bar{L}_i(T)$}, \label{eq:exhaustive_alpha_either}\\
\text{or } \bar{J}_i = 0:\quad &\alpha_i \in [0,1],\quad \text{$\alpha_i\theta_i$ equals $\bar{L}_i(T)$}. \label{eq:exhaustive_alpha_or}
\end{align}
If $\bar{J}_i > 0$, the customer arriving at $\clubsuit$ must have a waiting time strictly greater than the waiting patience, and thus the customer arriving at $\diamondsuit$'s waiting time is exactly equal to the waiting patience. 
On the other hand, if $\bar{J}_i = 0$, customers always join, and the one arriving at $\diamondsuit$ (coinciding with $\clubsuit$) must have a waiting time of exactly $\bar{L}_i(T)$.
(A similar phenomenon appears in the exogenous setting; see \Cref{lem:fixed_duration_waitingtime}.)

By substituting \eqref{eq:exhaustiveservice_off_duration_alpha_j_theta_j} into \eqref{eq:exhaustive_alpha_either} and \eqref{eq:exhaustive_alpha_or}, we obtain the constraints on the equilibrium variable $\alpha$, which can be written in a more compact way in \ref{prob:alpha}. That is, any solution of the following problem \ref{prob:alpha} is an equilibrium of the exhaustive service policy $\e$:
\begin{subequations}
    \begin{align}
\text{Find}\quad &\alpha \in \mathbb{R}^n \tag{$\mathsf{EQ}(T)$}\label{prob:alpha}\\
\text{s.t.}\quad &\alpha \geq 0, \label{con:non_neg}\\
&\alpha\leq 1, 
\label{con:ir}\\
&A\alpha\leq b(T), 
\label{con:period}\\
&(1 - \alpha)^\top (b(T) - A\alpha) = 0,\label{con:comp}
\end{align}
\end{subequations}
where $b(T) \in \mathbb{R}^n$ and $A \in \mathbb{R}^{n \times n}$ are defined by
\begin{align}
\label{eq:matrix_A_b}
    b(T) = (1^\top \tau + 1^\top T)1 - T,\quad
    A_{ij} = 
    \begin{cases}
        \theta_i, & \text{if } i=j\in \N\\
        -c_j, & \text{if } i\neq j\in\N
    \end{cases},
\end{align}
with $c_j = \lambda_j\theta_j/(\mu_j-\lambda_j)=\rho_j\theta_j/(1-\rho_j)$ for all $j\in\N$.
We put $c=(c_j)_{j\in\N}$. Observe that the vector $b(T)$ depends solely on the switchover times $1^\top \tau$ and the post-clearance durations $T$, while the matrix $A$ depends only on the utilization rates $\rho$ and the waiting patience $\theta$.
From \cref{eq:exhaustiveservice_on_duration_alpha_j_theta_j}, it can be seen that $c_i\alpha=\alpha_i\rho_i\theta_i/(1-\rho_i)$ represents the \on duration of queue $i$ minus its post-clearance duration, $L_i(T) - T_i$. Substituting \eqref{eq:exhaustiveservice_off_duration_alpha_j_theta_j} into $\bar{L}_i(T)$, %
constraint \eqref{con:period} ensures $\bar{L}_i(T)\geq \alpha_i \theta_i$ (by definition, see \Cref{fig:exhaustiveservice_equilibrium}). Constraint \eqref{con:comp} stipulates the ``either or" conditions in \eqref{eq:exhaustive_alpha_either} and \eqref{eq:exhaustive_alpha_or}, i.e., $\alpha_i=1$ or $\alpha_i\theta_i=\bar{L}_i(T)$. %

In the following, we use $x_{\I}\in\mathbb{R}^{|\I|}$ to denote the subvector of $x\in\mathbb{R}^{n}$ whose components are indexed by a non-empty $\I\subseteq \N$. Likewise, $A_{\I}\in\mathbb{R}^{|\I|\times|\I|}$ is the submatrix of $A\in\mathbb{R}^{n\times n}$ whose rows and columns are indexed by $\I$.
We use $1_{\I}$ (or $0_{\I}$) to denote the $|\I|$-dimensional vector of all ones (or all zeros), and we may drop the subscript $\I$ when it is clear from the context.

Finding an equilibrium $\alpha$, i.e., a feasible solution of problem \ref{prob:alpha}, is equivalent to finding a set $\I(T)\subseteq\N$ such that 
\begin{subequations}
    \begin{align}
       &(A\alpha)_{\I(T)}=b_{\I(T)} , \label{eq:alljoingset_1}\\
       &\alpha_{\barI(T)}=1, \label{eq:alljoingset_2}
    \end{align}
\end{subequations}
where $\barI(T) := \N\setminus\I(T)$, and the induced $\alpha$ satisfies all constraints in \ref{prob:alpha}. 
We call such an $\I(T)$ as the \textit{equilibrium set} of exhaustive service policy $\e$. Intuitively, $\I(T)$ is a set of queues where customers always join ($\bar{J}_i=0,~\forall i\in\I(T)$).
We may write $\I(T)$ as $\I$ when there is no confusion.

Given $\I\subseteq\N$, if $A_{\I}$ is invertible, by \eqref{eq:alljoingset_1}--\eqref{eq:alljoingset_2} it can be shown that the corresponding unique $\alpha$ is given by
\begin{subequations}
    \begin{align}
&\alpha_{\I} = A_{\I}^{-1}\left(b_{\I}+(c_{\barI}^\top 1_{\barI}) 1_{\I}\right) = \left(1^\top \tau+1^\top T +c^\top_{\barI}1_{\barI} \right)A_{\I}^{-1}1_{\I} - A_{\I}^{-1}T_{\I}, \label{eq:alpha_m}  \\
& \alpha_{\barI} = 1_{\barI} \label{eq:alpha_m_n}.
\end{align}
\end{subequations}
The $A_{\I}^{-1}$ exists if and only if $1^\top \rho_{\I}\neq 1$ as shown in \Cref{lem:inverse_mat} of Appendix \ref{app_sub:inversematrix}, where we also provide a closed-form expression of $A_{\I}^{-1}$.
Our next result shows that at equilibrium, $1^\top \rho_{\I}\neq 1$ must hold and thus the matrix $A_{\I}$ is indeed invertible. 
\begin{lemma}
\label{lem:equilibriumset}
Given an exhaustive service policy $\e$, if an equilibrium set $\I(T)\subseteq \N$  exists, we have $1^\top \rho_{\I(T)}<1$.
\end{lemma}

\Cref{lem:equilibriumset} can be intuitively understood: if $1^\top \rho_{\I}(T)<1$ is violated, the queue becomes unstable if all customers of $\I(T)$ join and the length of at least one queue in $\I(T)$ would grow to infinity. Consequently, some customers in these queues would experience infinite waiting times, contradicting the definition of the all-joining set $\I(T)$.

Our next result establishes the existence and uniqueness of equilibrium $\alpha$ and shows that it can be obtained by \Cref{alg:pivoting} in at most $\bar{n}$ steps. Here $\bar{n}$ is the maximum number of queues such that the sum of these queues' utilization rates is smaller than one, i.e., $\bar{n}=\max_{\I\subseteq\N} |\I|$ such that $\rho_{\I}^\top 1_{\I}<1$.
Note that $\bar{n}$ can be easily calculated by first sorting utilization rates in ascending order; then, find $\bar{n}$ such that $\sum_{i=1}^{\bar{n}}\rho_i<1$ and $\sum_{i=1}^{\bar{n}+1}\rho_i\geq 1$. For instance, in a system with identical arrival rate $\lambda_0$ and service rate $\mu_0$ across all queues, the parameter \(\bar{n}\) is simply given by $\min\{n,\floor{\mu_0/\lambda_0}\}$. In scenarios where the service rate is very close to the arrival rate, \Cref{alg:pivoting} can identify the equilibrium in as few as \emph{one} step.
Besides, notice that $\bar{n}$ is independent of $T$.

\begin{proposition}[Equilibrium Computation and Uniqueness]
    \label{prop:n_step}
There exists a unique equilibrium $\alpha(T)$ for any exhaustive service policy $\e$. 
Besides, \Cref{alg:pivoting} finds such a unique $\alpha(T)$ in at most $\bar{n}$ steps with closed-form updates.
\end{proposition}

\begin{algorithm}
\small
\caption{Finding the Equilibrium of Exhaustive Service Policy $\e$}
\begin{algorithmic}[1] %
\Require model parameters $\{\lambda_i,\mu_i,\tau_i,\theta_i: i\in\N\}$ and post-clearance durations $T$.
\Ensure equilibrium $\alpha(T)$ of $\e$.
\State Initialization: $\I\gets\emptyset$, $\nu\gets0$, and $q^{(0)}\gets b(T)-A1$. 
\Loop
\If{$q_{\barI}^{(\nu)} \geq 0$}
            \State \Return $\I$ and $\alpha(T) = (\alpha_{\I},\alpha_{\barI})$ using \eqref{eq:alpha_m}--\eqref{eq:alpha_m_n}.
        \Else
            \State $\J = \{i\in \bar{\I}\mid q_i^{(\nu)}<0\}$. 
            \Comment{\textsf{{\footnotesize{
            This is the set of not-all-joining queues violating the constraint \eqref{con:period} with $\alpha$ determined by the current all-joining set $\I$.}}}}
            \State $\I\gets \I\cup \J$, $\nu\gets\nu+1$. 
            \State Update $q_{\barI}^{(\nu)} \gets q_{\barI}^{(0)} - A_{\barI\I}A_{\I}^{-1}q_{\I}^{(0)}$, where $A_{\I}^{-1}$ is provided in closed form in \Cref{lem:inverse_mat}.
            \textsf{\footnotesize{\Comment{$A_{\I}$ is invertible at each iteration.}}}
\EndIf   
\EndLoop
\end{algorithmic}
\label{alg:pivoting}
\end{algorithm}

The development of \Cref{alg:pivoting} relies crucially on pivoting algorithms developed for solving linear complementarity problems. To see that, without constraint \eqref{con:non_neg}, %
problem \ref{prob:alpha} is equivalent to a standard LCP problem $\textsf{LCP}$($b(T)-A1, A$) by treating $z=1-\alpha$. %
Interestingly, in the proof of \Cref{prop:n_step}, we show that any $\alpha$ satisfying constraints $\eqref{con:ir}$, $\eqref{con:period}$, and $\eqref{con:comp}$ must also be non-negative. Thus, constriant \eqref{con:non_neg} can be dropped, and thus solving \ref{prob:alpha} is equivalent to solving $\textsf{LCP}$($b(T)-A1, A$).  %
In addition, note that our matrix $A$ defined in \eqref{eq:matrix_A_b} belongs to the class of $Z$-matrices, i.e., the off-diagonal elements of matrix $A$ are all non-positive.\footnote{Our matrix $A$ is a special class of $Z$-matrix, where the diagonal elements are positive. \cite{chandrasekaran_1970_LCP_Z} refers to such a matrix as $L$-matrix. The LCP problem with $Z$-matrix also appears in other operations settings, e.g., \cite{federgruen2015multi,hu_2016_OR_sequentialprice_LCP}.}

Based on the additional condition that the polyhedron composed of \eqref{con:ir} and \eqref{con:period} is non-empty (which indeed holds in our setting), this guarantees the equilibrium existence, see Theorem 3.11.6 in \cite{cottle_2009_LCPBook}.
Notice that an LCP problem with a $Z$-matrix may admit multiple solutions. The uniqueness in our problem arises from the special structure of the matrix $A$ in \eqref{eq:matrix_A_b}, namely that the elements in each column (except for the diagonal one) are the same. We provide further discussion of uniqueness later.%

Our \Cref{alg:pivoting} adapts the %
seminal Chandrasekaran's Algorithm \citep{chandrasekaran_1970_LCP_Z} which solves a general $\lcp$ with $A$ being a $Z$-matrix in at most $n$ steps. %
Compared with the Chandrasekaran Algorithm, our \Cref{alg:pivoting} simplifies the pivoting process with an analytical formula of $A_{\I}^{-1}$ (\Cref{lem:inverse_mat}), and of $\alpha_{\I}$ in \eqref{eq:alpha_m}--\eqref{eq:alpha_m_n} (in the proof, we show that the matrix $A_{\I}$ at each iteration must be invertible). The property that \Cref{alg:pivoting} terminates in at most $\bar{n}$ steps is mainly based on \Cref{lem:equilibriumset}.

We now discuss the intuition behind \Cref{alg:pivoting}. Recall that $\I$ is the set of all-joining queues, i.e., $\bar{J}_i=0, \forall i \in \I$; and conversely, $\barI$ is the set of not-all-joining queues, i.e., $\bar{J}_i>0, \forall i\in\barI$. 
\Cref{alg:pivoting} aims to return the equilibrium (all-joining) set $\I$ by expanding it with the queues in the not-joining set violating the constraint \eqref{con:period}, see lines~6--7. 
Its convergence relies on a key property: once a queue transitions from not-all-joining to all-joining, it \emph{never} reverts back during the pivoting process. We now go through the major steps of the algorithm.
\begin{enumerate}
    \item Initially, it sets $\I$ to be empty, implying that \textit{all} queues are initialized to be not-all-joining. This makes $\alpha=1$ according to \eqref{eq:alpha_m_n} and thus constraints \eqref{con:ir} are automatically satisfied for all queues. If the resulting $\alpha=1$ also satisfies constraint \eqref{con:period}, it constitutes an equilibrium.
    \item If not, the algorithm finds the set of not-all-joining queues $\J$ that violate the constraint \eqref{con:period}, i.e., $\bar{L}_j(T)\geq \alpha_j\theta_j$, in line 6 and then forms a new set $\I\cup \J$ in line 7.
    Line 8 is to update the condition about constraints $\bar{L}_j(T)\geq \alpha_j\theta_j$ under the new set.
    Interestingly, the algorithm then \textit{only} checks if $\bar{L}_j(T)\geq \alpha_j\theta_j$ is satisfied for not-all-joining queues (line 3), while disregarding constraint \eqref{con:ir} for all-joining queues. This is based on the property of matrix $A$ being a $Z$-matrix so that constraint \eqref{con:ir} for all-joining queues always holds. If $\bar{L}_j(T)\geq \alpha_j\theta_j$ is satisfied for every not-all-joining queue, the resulting $\alpha$ is deemed an equilibrium.
    \item If any of the conditions is not met, the algorithm repeats step (ii) until it reaches an equilibrium.
\end{enumerate}

We now discuss the uniqueness of the equilibrium $\alpha(T)$, which is based on the following lemma.
\begin{lemma}[Equivalence between Equilibrium and the Greatest Element]
\label{lem:lp_equilirbium}
For any post-clearance durations $T \in \mathbb{R}_+^{n}$, 
the polyhedron $\mathcal{A}(b(T),A):=\{\alpha\in\mathbb{R}^n: 0\leq \alpha\leq 1, A\alpha\leq b(T)\}$ admits a (unique) greatest element,\footnote{An element \( x \in \mathcal{X} \), where \( \mathcal{X} \subseteq \mathbb{R}^n \), is called the \emph{greatest element} of \( \mathcal{X} \) if and only if \( y \leq x \) (component-wise) for all \( y \in \mathcal{X} \).}
which coincides with the unique equilibrium $\alpha(T)$ of the exhaustive service policy $\e$.
\end{lemma}

Observe that the polyhedron $\mathcal{A}(b(T),A)$ in \Cref{lem:lp_equilirbium} is composed of constraints \eqref{con:non_neg}, \eqref{con:ir}, and \eqref{con:period} 
in problem \ref{prob:alpha} but in absence of constraints \eqref{con:comp}. Although not every polyhedron admits a greatest element, the polyhedron in \Cref{lem:lp_equilirbium} has one (and must be unique) and coincides with the equilibrium.\footnote{Note that the existence of the least element (our problem is equivalent to $\textsf{LCP}(b(T) - A1, A)$ by setting $z = 1 - \alpha$ with the additional constraint $z \leq 1$, so the greatest element becomes the least element) in the LCP literature with a $Z$-matrix is well established; see Theorem~1.11.6 in \cite{cottle_2009_LCPBook}. Moreover, this least element is a feasible solution to the problem $\textsf{LCP}(b(T) - A1, A)$. However, the uniqueness of the equilibrium is new and mainly comes from the special structure of the matrix $A$ defined in \eqref{eq:matrix_A_b}, namely that all elements except for the diagonal one in each column are the same.} We prove \Cref{lem:lp_equilirbium} based on duality theory. A visualization of the greatest element for the two-queue setting is provided in \Cref{fig:greatest_element_polyhedron} in Appendix \ref{app_sec:subsec_twoqueues}.
Although \Cref{lem:lp_equilirbium} suggests that one can find the equilibrium $\alpha(T)$ by solving an LP (e.g., by maximizing $1^\top \alpha$ subject to $\alpha\in\mathcal{A}(b(T),A))$, \Cref{alg:pivoting} offers a more efficient method. %

\subsection{Optimal Exhaustive Service Policy}
\label{sec:opt_ex}

In this subsection, we proceed to characterize the structure of an optimal exhaustive service policy and provide an efficient algorithm to compute an optimal one in at most $2\bar{n}$ steps. (Recall that, by \Cref{prop:n_step}, the equilibrium of any exhaustive service policy can be found in at most $\bar{n}$ steps.)

The throughput $\tp(T)$ of an exhaustive service policy $\e$ with its equilibrium $\alpha(T)$ is
\begin{align}
\label{eq:thp_dfn}
\tp(T) := \frac{\sum_{i\in\N}\left(\mu_i\cdot c_i\alpha_i(T)  + \lambda_iT_i\right)}{1^\top \tau + \sum_{i\in\N}\left(c_i\alpha_i(T)+T_i\right)}.
\end{align}
Recall that $c_i\alpha_i(T)$ is the \on duration minus the post-clearance duration at the equilibrium outcome, %
during which the server’s effective service rate is $\mu_i$. In contrast, the effective service rate is $\lambda_i$ during the post-clearance durations.

To find throughput-maximizing post-clearance durations, we solve:
\begin{align}
\max_{T\in{\mathbb{R}^n_{+}}}&\quad \tp(T) \text{ defined in } \eqref{eq:thp_dfn}\tag{$\mathsf{P}_{\mathsf{e}}$}
 \label{prob:opt_e}\\
\text{s.t.}&\quad \text{$\alpha(T)$ is the solution of problem \ref{prob:alpha}} \nonumber.
\end{align}
Problem \eqref{prob:opt_e} belongs to the class of linear-fractional optimization problems with linear complementarity constraints, which is generally NP-hard \citep{luo_1996_mpec}.
Our approach is to use the equilibrium of the pure exhaustive service policy $\pe$ as a bridge. We first prove the monotonicity of the equilibrium variable with respect to post-clearance durations $T$ (\Cref{lem:piecewise_linear_increasing} below). Then, based on this equilibrium monotonicity and the equilibrium of the pure exhaustive service policy $\pe$, we derive the structure of the optimal solution. Finally, based on this structure and \Cref{alg:pivoting}, we provide an algorithm that finds optimal post-clearance durations in at most $2\bar{n}$ steps (\Cref{alg:opt} below).

One might expect that increasing queue $i$'s post-clearance duration would attract more customers, thus leading to a decrease in the equilibrium not-joining duration $\bar{J}_i$. Consequently, the equilibrium variable $\alpha_i(T)$ would be decreasing in $T_i$, given $T_{-i}$. Somewhat surprisingly, the opposite is true, as established below.

\begin{lemma}[Piecewise Linearity and Monotonicity of Equilibrium]
\label{lem:piecewise_linear_increasing}
The equilibrium $\alpha(T)$ is a non-decreasing piecewise linear function of $T$.
\end{lemma}

The piecewise linearity stems from equations \eqref{eq:alpha_m}--\eqref{eq:alpha_m_n}; that is, given an equilibrium set $\I(T)$, the function $\alpha(T)$ is linear in $T$. The non-decreasing property of $\alpha_{-i}(T)$ with respect to $T_i$ (for fixed $T_{-i}$) can be intuitively understood: a higher $T_i$ increases the waiting times for customers in the other queues, imposing a greater negative externality. Consequently, more customers in these queues may opt not to join, which implies that $\alpha_{-i}(T)$ is non-decreasing in $T_i$.

Regarding the monotonicity of $\alpha_i(T)$ with respect to $T_i$, note that since the server must clear the queue upon leaving, whether customers join or not depends primarily on the \off duration of queue $i$.
Since \(\alpha_{-i}(T)\) increases with $T_i$ as explained above, the maximum queue length of other queues would increase (see \Cref{fig:exhaustiveservice_equilibrium}), and thus the server requires more time to clear these queues' customers. This leads to an extended \off duration of queue $i$, which in turn causes more customers of queue \(i\) not to join. Therefore, \(\alpha_i(T)\) is also non-decreasing with \(T_i\).

\smallskip

Leveraging the equilibrium set $\I(0)$ of $\pe$ as a bridge, we obtain the structure of an optimal exhaustive service policy as follows.

\begin{theorem}[Structure of an Optimal Exhaustive Service Policy]
        \label{thm:structure_opt_ex}
    We have
  \begin{enumerate}
  \item If $\I(0)=\N$, the pure exhaustive service policy $\pe$ achieves the first-best throughput $1^\top \lambda$.
      \item If $\I(0)\neq \N$, for any $j\in\argmax_{i\in \bar{\I}(0)}\lambda_i$, we can set $T^\ast_{-j}=0$ without loss of optimality. 
      Furthermore, if $\I(0)=\emptyset$, either the pure exhaustive service policy $\pe$ %
      or always serving the queue with the highest arrival rate is optimal.
  \end{enumerate}
\end{theorem}

\Cref{thm:structure_opt_ex} shows that, under (one of) the optimal exhaustive service policies, \textit{at most one} queue can have a \emph{non-zero} post-clearance duration %
and it must be the queue with the highest arrival rates among all not-all-joining queues $\bar{\I}(0)$ under pure exhaustive policy $\pe$. %

\Cref{thm:structure_opt_ex} (i) is straightforward since %
$\I(0)=\N$ suggests that all customers of all queues join the system. This leads to the first best throughput $1^\top \lambda$. For the proof of the first part of \Cref{thm:structure_opt_ex} (ii), we first show that, without loss of optimality, we can set post-clearance durations of all-joining queues under the pure exhaustive service policy $\pe$ as zero, i.e., $T_{\I(0)}=0$. Based on this, the second term in \eqref{eq:alpha_m} becomes zero, and thus $\alpha_{\I(0)}$ only depends on $1^\top T_{\barI(0)}$.
More interestingly, we then show that the throughput can be written as a term \textit{only} depending on $1^\top T_{\barI(0)}$ plus a term proportional to $\lambda_{\barI(0)}^\top T_{\barI(0)}$. Thus, for any $T_{\barI(0)}$, we can construct another solution ($\tilde{T}_{j}=1^\top T_{\barI(0)},\tilde{T}_{-j}=0)$ for some $j\in\arg\max_{i\in\bar{\mathcal{I}}(0)}\lambda_{i}$ achieving (weakly) higher throughput. The second part of \Cref{thm:structure_opt_ex} (ii) is an interesting corollary of the first part: when $\I(0)=\emptyset$, without loss of optimality, we can set all queues' post-clearance durations as zero except for a queue $j$ with the \emph{highest} arrival rate. Besides, when $\I(0)=\emptyset$, i.e., $\alpha(0)=1$, by the non-decreasing property of $\alpha(T)$ (\Cref{lem:piecewise_linear_increasing}), we must have 
$\alpha(T)=1$ for any $T\geq 0$. Based on this, it can be shown that the throughput $\tp(T)$ is 
monotone (either decreasing or increasing) in $T_{j}$ given $T_{-j}=0$. As a consequence, either $T^\ast_{j}=\infty$ or $T^\ast_{j}=0$ is optimal.

\begin{example}[The Homogeneous System]
\label{example:homogeneous_system_1}
    Consider a system where all queues have identical service rates, arrival rates, and waiting patience.  Then, either the pure exhaustive service policy $\pe$ or always serving any one of the queues is an optimal exhaustive service policy. To see this, in the homogeneous system, %
    by symmetry, all elements of $\alpha(0)$ are the same. Thus, %
either $\mathcal{I}(0)=\N$ or $\mathcal{I}(0)=\emptyset$.
\end{example}

By \Cref{thm:structure_opt_ex}, to find an optimal exhaustive service policy, i.e., optimal post-clearance durations $T^\ast$, we only need to know $T_{j}^\ast$ since all other queues' post-clearance durations can be set to zero, where $j$ is a queue with the highest arrival rate among $\barI(0)$. We now introduce \Cref{alg:opt} that finds such an optimal $T_{j}^\ast$ efficiently. 

Intuitively, \Cref{alg:opt} works as follows. It begins by computing the equilibrium under the pure exhaustive service policy $\pe$, which can be efficiently obtained using \Cref{alg:pivoting}. This equilibrium reveals which queue~$j$ may have a non-zero post-clearance duration. Next, since the equilibrium $\alpha(T_{j}, 0_{-j})$ is non-decreasing in~$T_{j}$, increasing~$T_{j}$ cannot reduce the values of $\alpha_{\I(0)}(T_{j}, 0_{-j})$, while the values of $\alpha_{\barI(0)}(T_{j}, 0_{-j})$ remain fixed at one.  A key consequence is that as~$T_{j}$ increases from zero, the number of all-joining queues, denoted $|\I(T_{j}, 0_{-j})|$, cannot increase. The order in which queues transition from all-joining to not-all-joining is determined in line~4 of the algorithm. Moreover, as shown in the proof, the throughput is monotonic (though not linear) in~$T_{j}$ within any interval where $|\I(T_{j}, 0_{-j})|$ remains constant. This means that it is sufficient to evaluate the throughput only at the \emph{boundary} values of~$T_{j}$ where $|\I(T_{j}, 0_{-j})|$ changes (see lines~5--9). Finally, lines~12--15 of the algorithm check whether it is optimal to serve only one queue---specifically, queue~$j$---at all times.

\begin{algorithm}[htbp]
\caption{Computing Optimal Exhaustive Service Policy}
\small
\begin{algorithmic}[1]
\Require model parameters $\{\lambda_i, \mu_i, \tau_i, \theta_i : i \in \N\}$.
\Ensure optimal exhaustive service policy $T^\ast$.
\State Initialize: $\I \gets \I(0)$, $\tp \gets \tp(0)$ by \eqref{eq:thp_dfn} and \Cref{alg:pivoting}, $T \gets 0$.
\If {$\I = \N$} 
    \Return $T^\ast = T = 0$ \Comment{{\footnotesize{\textsf{Pure exhaustive policy is optimal.}}}}
\EndIf
\State Choose any queue $j \in \argmax_{i \in \barI} \lambda_i$. \Comment{{\footnotesize{\textsf{Select a queue with potentially positive post-clearance duration.}}}}
\State Order queues in  $\mathcal{I}(0)$ in a descending order of $(\mu_i - \lambda_i)/(\mu_i \theta_i)$, denoted by $\mathcal{K} = [k_1, k_2, \dots, k_{|\mathcal{I}(0)|}]$.
\Comment{{\footnotesize{\textsf{If $\I(0)$ is empty, skip this step and the following ``for loop'' procedure from line 5 to 11.}}}}
\For{$k$ in $\mathcal{K}$}
    \State  $T_j \gets (\mu_k \theta_k) \cdot (1 - \rho_{\I}^\top 1_{\I})/(\mu_k - \lambda_k) - 1^\top \tau - c^\top_{\barI} 1_{\barI}$. \Comment{\textsf{\footnotesize{The boundary value that exactly makes $\alpha_k=1$ (see \Cref{lem:monoton_tp_hat_T}).}}}
    \State $\I \gets \I \setminus \{k\}$. \Comment{{\footnotesize{\textsf{The new equilibrium (all-joining) set (see \Cref{lem:monoton_tp_hat_T}}}}}).
    \State Calculate equilibrium $\alpha(T_j,0_{-j})$ of $\pi(T_j,0_{-j})$ using \eqref{eq:alpha_m} and \eqref{eq:alpha_m_n}:
\begin{align*}
    \alpha_i(T_j,0_{-j})=\frac{1^\top \tau+T_j+c_{\barI}^\top 1_{\barI}}{1-\rho_{\I}^\top 1_{\I}}\cdot\frac{\mu_i-\lambda_i}{\mu_i\theta_i},~\forall i\in \I; \quad \alpha_{\barI}(T_j,0_{-j})=1_{\barI}. 
\end{align*}
\State Calculate throughput $\tp(T_j,0_{-j})$ using \eqref{eq:thp_dfn}:
\begin{align*}
     \tp(T_j,0_{-j}) = \frac{\lambda_jT_j+ \sum_{i\in\N} \mu_i c_i  \alpha_i(T_j,0_{-j})}{1^\top \tau+c^\top  \alpha(T_j,0_{-j})+T_j}.
\end{align*}
    \If{$\tp < \tp(T_j, 0_{-j})$} 
        \State Update $\tp^\ast = \tp(T_j, 0_{-j})$ and $T_j^\ast = T_j$. \Comment{{\footnotesize{\textsf{Keep the best solution so far.}}}}
    \EndIf
\EndFor
\If{$\tp^{\ast} \geq \lambda_j$}
   \State \Return $T^\ast$
\Else
   \State \Return $T_j^\ast = +\infty$, $T_{-j}^\ast = 0$, and $\tp^\ast = \lambda_j$. \Comment{{\footnotesize{\textsf{Always serving queue $j$ is optimal.}}}}
\EndIf
\end{algorithmic}
\label{alg:opt}
\end{algorithm}

\begin{theorem}%
\label{thm:opt_ex_algo}
\Cref{alg:opt} finds an optimal exhaustive service policy in at most $2\bar{n}$ steps with closed-form updates.
\end{theorem}

\Cref{thm:opt_ex_algo} is grounded in three pivotal properties. First, by \Cref{thm:structure_opt_ex}, we can, without loss of optimality, set $T^\ast_{-j} = 0$ for some $j\in\argmax_{i\in \bar{\I}(0)}\lambda_i$, significantly simplifying our analysis.
Second, the variable $\alpha$ (which may not yet be at equilibrium) provided in \eqref{eq:alpha_m}--\eqref{eq:alpha_m_n} exhibits a non-decreasing relationship with $T_j$. This property allows for an efficient method to calculate the equilibrium set (and the corresponding variable) under a new $T_j$, as shown in lines~7--8. 
This effectively avoids re-running \Cref{alg:pivoting} to derive  $\alpha(T_j, 0_{-j})$. 
Third, as mentioned before, it is sufficient to only check boundary values of $T_j$ that cause the equilibrium set to change. Based on Theorems \ref{thm:structure_opt_ex} and \ref{thm:opt_ex_algo}, we also provide a closed-form expression of the optimal policy when $n=2$ in Appendix \ref{app_sec:subsec_twoqueues}.

Regarding the $2\bar{n}$-step convergence, the algorithm requires at most $\bar{n}$ steps to identify the equilibrium set $\I(0)$ in line~1, as established in \Cref{prop:n_step}. Since $|\I(0)| \leq \bar{n}$, the ``for loop'' in lines~5--11 also completes in at most $\bar{n}$ steps. Consequently, \Cref{alg:opt} terminates within at most $2\bar{n}$ steps.

\Cref{fig:exhaustive_hat_T_j_numerical} illustrates how the throughput varies with the post-clearance duration $T_j$, where $j \in \argmax_{i \in \barI(0)} \lambda_i$ denotes the queue with the highest arrival rate among those eligible for a non-zero duration. As discussed earlier, \Cref{alg:opt} effectively compares throughput only at the boundary values of $T_j$ (indicated by the vertical lines), since the throughput is monotonic between consecutive boundaries. \Cref{subfig:exh-0.5} presents the case with relatively low waiting patience. Here, the equilibrium joining set under the pure exhaustive service policy is empty. For this set of parameters, the throughput decreases with increasing $T_j$, indicating that the optimal policy is the pure exhaustive policy. \Cref{subfig:exh-1.5} shows the case with moderate waiting patience, where $\I(0) = \{1, 3\}$. In this example, the throughput initially increases and then decreases, with the optimal post-clearance duration $T_j^\ast$ around $0.66$.  Finally, \Cref{subfig:exh-3.0} considers the case with higher waiting patience, showing a similar pattern to \Cref{subfig:exh-1.5}, but under larger patience levels.

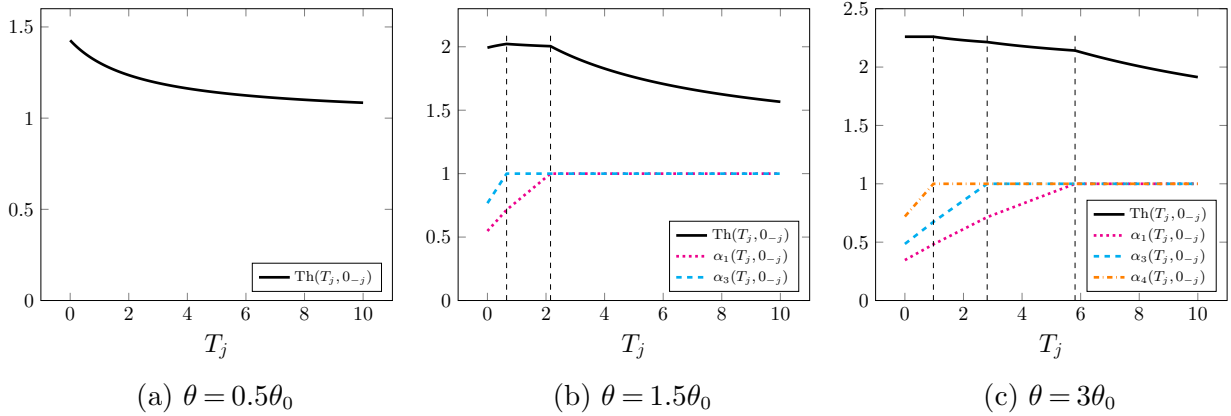
\begin{figure*}[htbp]
    \centering
    \subfloat[$\theta=0.5\theta_0$ \label{subfig:exh-0.5}]{%
        \resizebox{.32\textwidth}{!}{\begin{tikzpicture}
\begin{axis}[
    xlabel={$T_j$},
    legend style={legend pos=south east, font=\scriptsize},
    label style={font=\large},
    ymin=0,
    ymax=1.6,
    xtick={0,2,4,6,8,10},
    xticklabels={0,2,4,6,8,10}, %
    yticklabel style={/pgf/number format/assume math mode} %
]

\addplot[
    solid,
    line width=1.5pt %
] table [
    x=Hat_T_j_values,
    y=Throughput,
    col sep=space %
]{numerical_data/exhaustive_service_policy/exhaustive_data_I_empty.txt};
\addlegendentry{$\tp(T_j,0_{-j})$};

\end{axis}
\end{tikzpicture}}
    }
    \hfill
    \subfloat[$\theta=1.5\theta_0$ \label{subfig:exh-1.5}]{%
        \resizebox{.32\textwidth}{!}{\begin{tikzpicture}
\begin{axis}[
    xlabel={$T_j$},
    legend style={legend pos=south east, font=\scriptsize},
    label style={font=\large},
    ymin=0,
    ymax=2.3,
    xtick={0,2,4,6,8,10},
    xticklabels={0,2,4,6,8,10},
    yticklabel style={/pgf/number format/assume math mode}
]

\addplot[
    solid,
    line width=1.5pt
] table [
    x=Hat_T_j_values,
    y=Throughput,
    col sep=space
]{numerical_data/exhaustive_service_policy/exhaustive_data_I_2_0.txt};
\addlegendentry{$\tp(T_j,0_{-j})$};

\addplot[
    dotted,
    line width=1.5pt,
    color=magenta
] table [
    x=Hat_T_j_values,
    y=Alpha0,
    col sep=space
]{numerical_data/exhaustive_service_policy/exhaustive_data_I_2_0.txt};
\addlegendentry{$\alpha_1(T_j,0_{-j})$};

\addplot[
    dashed,
    line width=1.5pt,
    color=cyan
] table [
    x=Hat_T_j_values,
    y=Alpha2,
    col sep=space
]{numerical_data/exhaustive_service_policy/exhaustive_data_I_2_0.txt};
\addlegendentry{$\alpha_3(T_j,0_{-j})$};

\draw [dashed, line width=0.1pt] (axis cs:0.65460526,0) -- (axis cs:0.65460526,2.1);
\draw [dashed, line width=0.1pt] (axis cs:2.15460526,0) -- (axis cs:2.15460526,2.1);

\end{axis}
\end{tikzpicture}}
    }
    \hfill
    \subfloat[$\theta=3\theta_0$ \label{subfig:exh-3.0}]{%
        \resizebox{.32\textwidth}{!}{\begin{tikzpicture}
\begin{axis}[
    xlabel={$T_j$},
    legend style={legend pos=south east, font=\scriptsize},
    label style={font=\large},
    ymin=0,
    ymax=2.5,
    xtick={0,2,4,6,8,10},
    xticklabels={0,2,4,6,8,10},
    yticklabel style={/pgf/number format/assume math mode}
]

\addplot[
    solid,
    line width=1.5pt
] table [
    x=Hat_T_j_values,
    y=Throughput,
    col sep=space
]{numerical_data/exhaustive_service_policy/exhaustive_data_I_3_2_0.txt};
\addlegendentry{$\tp(T_j,0_{-j})$};

\addplot[
    dotted,
    line width=1.5pt,
    color=magenta
] table [
    x=Hat_T_j_values,
    y=Alpha0,
    col sep=space
]{numerical_data/exhaustive_service_policy/exhaustive_data_I_3_2_0.txt};
\addlegendentry{$\alpha_1(T_j,0_{-j})$};

\addplot[
    dashed,
    line width=1.5pt,
    color=cyan
] table [
    x=Hat_T_j_values,
    y=Alpha2,
    col sep=space
]{numerical_data/exhaustive_service_policy/exhaustive_data_I_3_2_0.txt};
\addlegendentry{$\alpha_3(T_j,0_{-j})$};

\addplot[
    dash dot,
    line width=1.5pt,
    color=orange
] table [
    x=Hat_T_j_values,
    y=Alpha3,
    col sep=space
]{numerical_data/exhaustive_service_policy/exhaustive_data_I_3_2_0.txt};
\addlegendentry{$\alpha_4(T_j,0_{-j})$};

\draw [dashed, line width=0.1pt] (axis cs:0.97368421,0) -- (axis cs:0.97368421,2.3);
\draw [dashed, line width=0.1pt] (axis cs:2.80921053,0) -- (axis cs:2.80921053,2.3);
\draw [dashed, line width=0.1pt] (axis cs:5.80921053,0) -- (axis cs:5.80921053,2.3);

\end{axis}
\end{tikzpicture}}
    }
    \caption{Throughput and Equilibrium of the Exhaustive Service Policy 
    $T=(T_j,T_{-j}=0)$, where $j \in \argmax_{i \in \barI(0)} \lambda_i$. \emph{Note}. %
    The figures are based on the parameters $n = 4$, $\lambda = [1.0, 1.2, 0.4, 0.2]$, $\mu = [5.0, 3.0, 2.0, 4.0]$, $1^\top \tau = 1.5$, and $\theta_0 = [3.5, 0.5, 2.5, 2.0]$.
    It turns out that in all three cases, the queue $j \in \argmax_{i \in \barI(0)} \lambda_i$ is queue $2$.
    The values of the non-plotted variables $\alpha_i$ remain fixed at~1 throughout. The equilibrium joining sets under the pure exhaustive service policy are $\I(0) = \emptyset$ in~(a), $\I(0) = \{1, 3\}$ in~(b), and $\I(0) = \{1, 3, 4\}$ in~(c).}
    \label{fig:exhaustive_hat_T_j_numerical}
\end{figure*}

It is worth emphasizing that the set of all-joining queues $\I(0)$ under the pure exhaustive service policy $\pe$, as described in \Cref{thm:structure_opt_ex} and \Cref{alg:opt}, serves as a meaningful indicator of ``patient'' queues whose customers who are more inclined to join. This is important because the patience parameters $\theta_i$ alone do not provide a complete picture. Just as an example, a high patience level $\theta_i$ combined with very low service rates $\mu_i$ can possibly be similar to a case with lower patience but higher service rates. In this sense, $\I(0)$ offers a more informative measure by capturing the interplay between patience, utilization, as well as strategic interactions among customers across different queues.

\subsection{Connection to the Exogenous Regime}
\label{subsec:connection}

We now relate our results back to the exogenous regime analyzed in \Cref{sec:exogenous_regime}. Notice that the equilibrium under an exhaustive service policy (as shown in \Cref{fig:exhaustiveservice_equilibrium}) shares the same structural characteristics as the equilibrium outcomes in Figures~\ref{subfig:exh_case1} and~\ref{subfig:exh_case4}. The key difference lies in the control mechanism: in \Cref{fig:exhaustiveservice_equilibrium}, the server controls only the post-clearance durations~$T$, which in turn endogenously determine the \on-\off\ durations. In contrast, in Figures~\ref{subfig:exh_case1} and~\ref{subfig:exh_case4}, the planner directly specifies the \on-\off\ durations, and the post-clearance durations are induced as a consequence.

Among the four exhaustive equilibrium patterns identified in the exogenous regime (see \Cref{fig:fixed_duration_exhaustive_equilibrium}), an exhaustive service policy can replicate only the patterns shown in Figures~\ref{subfig:exh_case1} and~\ref{subfig:exh_case4}. As established in \Cref{thm:polling_identical_service_rates}, when service rates are identical, the optimal equilibrium outcome in the exogenous regime corresponds precisely to these two patterns and, therefore, can be implemented through the optimal exhaustive service policy. In such cases, this leads to a significant computational advantage, as the problem can be solved efficiently using \Cref{alg:opt}, rather than through a linear program~\eqref{prob:LP_fixed_duration}.

In fact, for any equilibrium outcome induced by a given service policy, there exists an exogenous \on-\off\ duration setting that replicates the same equilibrium outcome. This idea is reminiscent of the \textit{revelation principle} in mechanism design \citep{myerson1979incentive}. Specifically, given an equilibrium outcome under some service policy, we can extract the \on-\off\ durations realized at equilibrium and use them to construct a service policy with exogenous \on-\off\ durations considered in \Cref{sec:exogenous_regime}. %

Starting from the same system state as in the original equilibrium, customers face identical waiting times, resulting in the same equilibrium outcome. Consequently, focusing on the exogenous regime involves no loss of generality, as any equilibrium outcome can be reproduced through a suitable choice of exogenous \on-\off\ durations. Nevertheless, the endogenous regime provides additional insights into customers' strategic behavior and often leads to optimal policies that are easier to compute and require fewer parameters to implement (e.g., the pure exhaustive policy).

\section{Closing Remarks}
\label{sec:closing}

In summary, this paper analyzes \on-\off systems with strategic customers under two distinct sources of \on-\off durations. In both the exogenous and endogenous regimes, equilibrium behaviors feature herding cycles, with both follow-the-crowd (FTC) and avoid-the-crowd (ATC) dynamics emerging. For the exogenous regime, we present a compact linear program to compute the optimal \on-\off durations. For the endogenous regime, we design an efficient algorithm that identifies the optimal exhaustive service policies within at most $2n$ iterations.

An interesting direction for future research is to investigate how the planner’s choice of information disclosure policy influences equilibrium outcomes and to characterize the optimal disclosure strategy. Another promising avenue is to extend our model to a network setting with multiple servers.

\ACKNOWLEDGMENT{The authors thank Li Jin from Shanghai Jiao Tong University for his support and the valuable early-stage discussions that helped shape this project. The first author acknowledges funding from the Imperial College President's PhD Scholarship programme. The third author also gratefully acknowledges the Schloss Dagstuhl Seminar 24281 \textit{Dynamic Traffic Models in Transportation Science} for inspiring ideas related to dynamic congestion games.}

\bibliographystyle{informs2014}

\bibliography{references.bib}

\newpage

\ECSwitch

\ECHead{\centering{\textsc{Online Appendix}}}

\setcounter{equation}{0}
\renewcommand{\theequation}{\Alph{section}.\arabic{equation}}%

\begin{APPENDICES}

\section{Auxiliary Results}
\label{app_sec:auxiliary_result}

In Appendix \ref{app_sec:lambda_i_greater_mu_i}, we present the results for the case $\lambda_i \geq \mu_i$. 
In Appendix \ref{app_subsec:optimal_exo_n=1}, we discuss the optimal exogenous \on-\off durations in the single-queue setting with additional operational constraints.
In Appendix \ref{app_sec:subsec_twoqueues}, we consider a two-queue example for the endogenous regime, where we provide explicit analytical characterizations of
the equilibrium and the optimal policy. Additional technical results are provided in Appendix \ref{app_sub:inversematrix}.

\subsection{The case of $\lambda_i\geq \mu_i$}
\label{app_sec:lambda_i_greater_mu_i}

\subsubsection{Equilibrium Analysis\\}

\hspace{-4mm}\noindent\textit{Exogenous Regime.} Given an exogenous \on-\off duration $(L_i,\bar{L}_i)$ with $\mu_i \leq \lambda_i$, the primary distinction from the case of $\mu_i>\lambda_i$ is that customers adopt a mixed joining strategy with probability $\frac{\mu_i}{\lambda_i}$ each time they consider joining. Despite this adjustment, the pattern of herding cycles remains intact.
The waiting time dynamics for the case of $\mu_i\leq \lambda_i$ is illustrated in \Cref{fig:waiting_time_fixed_duration_mui_smaller}.

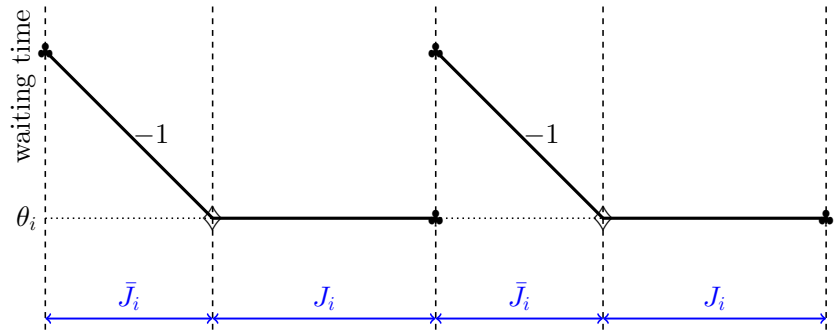
\begin{figure*}[htbp]
  \centering
  \resizebox{.68\textwidth}{!}{\begin{tikzpicture}[font=\large, line width=1pt] 

\draw (-0.4,6.9) node[left, rotate=90] {waiting time};

\draw[line width=1.6pt] (0,6) -- (3,3) -- (7,3);
\draw[line width=1.6pt] (7,6) -- (10,3) -- (14,3);

\draw[thick,dashed] (0,1) -- (0,6.8);
\draw[thick,dashed] (3,1) -- (3,6.8);
\draw[thick,dashed] (7,1) -- (7,6.8);
\draw[thick,dashed] (10,1) -- (10,6.8);
\draw[thick,dashed] (14,1) -- (14,6.8);

\draw [to-to,blue] (0,1.2) --(3,1.2);
\draw (1.5,1.2) node[right, color=blue, above] {$\bar{J}_i$};

\draw [to-to,blue] (7,1.2) --(10,1.2);
\draw (8.5,1.2) node[right, color=blue, above] {$\bar{J}_i$};

\draw [to-to,blue] (3,1.2) --(7,1.2);
\draw (5,1.2) node[right, color=blue, above] {$J_i$};

\draw [to-to,blue] (10,1.2) --(14,1.2);
\draw (12,1.2) node[right, color=blue, above] {$J_i$};

\draw  (1.9,4.5) node {$-1$};
\draw  (8.9,4.5) node {$-1$};

\draw[thick,dotted] (0,3) -- (14,3);
\draw (0,3) node[left] {$\theta_i$};

\draw (0,6) node {\small{$\clubsuit$}};
\draw (7,6) node {\small{$\clubsuit$}};
\draw (7,3) node {\small{$\clubsuit$}};
\draw (14,3) node {\small{$\clubsuit$}};

\draw (3,3) node {$\diamondsuit$};
\draw (10,3) node {$\diamondsuit$};

\end{tikzpicture}}
  \caption{Illustration of the Waiting Time Dynamics under Exogenous \on-\off Duration ($\mu_i\leq \lambda_i$). 
  \textit{Note}. The amount of the upward jump at \small{$\clubsuit$} is $\bar{L}_i$. 
  }
\label{fig:waiting_time_fixed_duration_mui_smaller}
\end{figure*}

\begin{theorem}
\label{thm:fixed_duration_equilibrium_mui_smaller}
Given an exogenous \on-\off duration $(L_i, \bar{L}_i)$ with $\mu_i\leq \lambda_i$:
\begin{enumerate}
    \item When $\bar{L}_i\geq \theta_i$, there exists a unique equilibrium outcome, which is exhaustive, as characterized in \Cref{fig:fixed_duration_exhaustive_equilibrium_mui_smaller}.
    \item When $\bar{L}_i<\theta_i$, there exists a unique equilibrium outcome, which is non-exhaustive, as characterized in \Cref{fig:fixed_duration_nonexhaustive_equilibrium_mui_smaller} with four different cases further depending on $L_i$.
    Let $k_i := \floor{\theta_i/(L_i + \bar{L}_i)}$.
\begin{enumerate}
    \item If $ L_i \geq  \bar{L}_i$, the equilibrium outcome depends on which interval $k_i$ falls into:\footnote{The boundary values can belong to either case.}
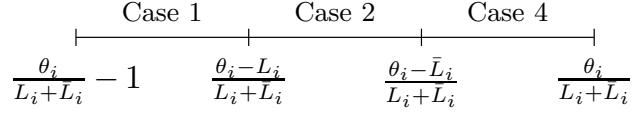
    \begin{center}
        \resizebox{0.53\textwidth}{!}{\begin{tikzpicture}
  \draw (1,0) -- (7,0);
  
  \foreach \x/\lab in {1/{$\frac{\theta_i}{L_i+\bar{L}_i}-1$}, 
                       3/{$\frac{\theta_i-L_i}{L_i+\bar{L}_i}$}, 
                       5/{$\frac{\theta_i-\bar{L}_i}{L_i+\bar{L}_i}$}, 
                       7/{$\frac{\theta_i}{L_i+\bar{L}_i}$}} {
    \draw (\x,0.1) -- (\x,-0.1) node[below] {\lab}; %
  }

  \node at (2,0.3) {\footnotesize{\textnormal{Case 1}}};
  \node at (4,0.3) {\footnotesize{\textnormal{Case 2}}};
  \node at (6,0.3) {\footnotesize{\textnormal{Case 4}}};
\end{tikzpicture}}
    \end{center}
    \item If $ L_i < \bar{L}_i$, again, the equilibrium outcome depends on which interval $k_i$ falls into:
    \begin{center}
        \resizebox{0.53\textwidth}{!}
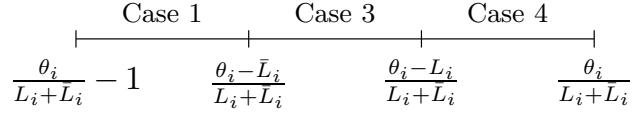
{\begin{tikzpicture}
    \draw (1,0) -- (7,0);
  \foreach \x/\lab in {1/{$\frac{\theta_i}{L_i+\bar{L}_i}-1$}, 
                       3/{$\frac{\theta_i-\bar{L}_i}{L_i+\bar{L}_i}$}, 
                       5/{$\frac{\theta_i-L_i}{L_i+\bar{L}_i}$}, 
                       7/{$\frac{\theta_i}{L_i+\bar{L}_i}$}} {
    \draw (\x,0.1) -- (\x,-0.1) node[below] {\lab}; %
  }

  \node at (2,0.3) {\footnotesize{\textnormal{Case 1}}};
  \node at (4,0.3) {\footnotesize{\textnormal{Case 3}}};
  \node at (6,0.3) {\footnotesize{\textnormal{Case 4}}};

\end{tikzpicture}}
    \end{center}
\end{enumerate}
The corresponding variables $\zeta_i$ and $\underline{q}_i$ in \Cref{fig:fixed_duration_nonexhaustive_equilibrium_mui_smaller} are uniquely determined as follows:
\begin{enumerate}[label=Case \arabic*:, leftmargin=1.8cm]
    \item $\zeta_i = (k_i + 1)(L_i + \bar{L}_i) - \theta_i$ and $\underline{q}_i = k_i \mu_i L_i + \mu_i (L_i - \zeta_i)$;
    \item $\zeta_i = k_i (L_i + \bar{L}_i) +  L_i - \theta_i$ and $\underline{q}_i = k_i \mu_i L_i + \mu_i \left(  L_i - \bar{L}_i - \zeta_i \right)$;
    \item $\zeta_i = k_i (L_i + \bar{L}_i) + \bar{L}_i - \theta_i$ and $\underline{q}_i = k_i \mu_i L_i$;
    \item $\zeta_i = k_i (L_i + \bar{L}_i) + \bar{L}_i - \theta_i$ and $\underline{q}_i = k_i \mu_i L_i$.
\end{enumerate}
\end{enumerate}
\end{theorem}

\begin{figure*}[ht]
\centering
  \subfloat[{\small  $\bar{L}_i\geq \theta_i$ and $L_i\geq \frac{\lambda_i\theta_i}{\mu_i}$}]{  \resizebox{.48\textwidth}{!}{\begin{tikzpicture}[font=\large, line width=1pt] 

\draw (-3.4,6) node[left, rotate=90] {{\small{queue length}}};

\draw[line width=1.6pt] (-3,1) -- (-2,1) -- (0,4) -- (2,4)-- (4,1) -- (5,1) -- (7,4);

\draw[thick,dashed] (0,0) -- (0,6);
\draw[thick,dashed] (-3,0) -- (-3,6);
\draw[thick,dashed] (4,0) -- (4,6);
\draw[thick,dashed] (7,0) -- (7,6);

\draw [to-to,black] (4,5) --(7,5);
\draw (5.5, 5) node[right, color=black, above] {$\bar{L}_i$};

\draw [to-to,black] (-3,5) --(0,5);
\draw (-1.5, 5) node[right, color=black, above] {$\bar{L}_i$};
\draw [to-to,black] (4,5) --(0,5);
\draw (2,5) node[right, color=black, above] {$L_i$};

\draw	(0,4.1) node[anchor=east,red] {$\mu_i\theta_i$};
\draw	(-3,1.2) node[anchor=east,red] {$0$};

\draw[thick,dotted] (-2,1) -- (-2,0);
\draw[thick,dotted] (2,4) -- (2,0);
\draw[thick,dotted] (5,0) -- (5,1);

\draw [to-to,blue] (-2,0.2) --(2,0.2);
\draw (0,0.2) node[right, color=blue, above] {$J_i$};

\draw [to-to,blue] (2,0.2) --(5,0.2);
\draw (3.5,0.2) node[right, color=blue, above] {$\bar{J}_i$};

\draw [to-to,red] (5,1) --(7,1);
\draw (6,1) node[right, color=red, above] {$\theta_i$};

\draw [above] (-0.9,2.8) node {\small{$\mu_i$}};
\draw [above] (5.9,2.5) node {\small{$\mu_i$}};
\draw (3.5,2.3) node {\small{$-\mu_i$}};

\draw (2,4) node {\small{$\clubsuit$}};

\draw (-2,1) node {$\diamondsuit$};
\draw (5,1) node {$\diamondsuit$};

\end{tikzpicture}}}
  \hfill
  \subfloat[{\small  $\bar{L}_i\geq \theta_i$ and $ L_i<\frac{\lambda_i\theta_i}{\mu_i}$}]{  \resizebox{.48\textwidth}{!}{\begin{tikzpicture}[font=\large, line width=1pt]

\draw[line width=1.6pt](-3,1) -- (-2,1) -- (0,4) --(1,4) -- (3,1) -- (4,1) -- (6,4) --(7,4);

\draw[thick,dashed] (-3,0) -- (-3,6);
\draw[thick,dashed] (1,0) -- (1,6);
\draw[thick,dashed] (3,0) -- (3,6);
\draw[thick,dashed] (7,0) -- (7,6);

\draw [to-to,black] (-3,5) --(1,5);
\draw (-1, 5) node[right, color=black, above] {$\bar{L}_i$};
\draw [to-to,black] (1,5) --(3,5);
\draw (2,5) node[right, color=black, above] {$L_i$};
\draw [to-to,black] (3,5) --(7,5);
\draw (5, 5) node[right, color=black, above] {$\bar{L}_i$};

\draw	(1,4.2) node[anchor=west,red] {$\mu_i L_i$};
\draw	(-3,1.2) node[anchor=east,red] {$0$};

\draw[thick,dotted] (-2,1) -- (-2,0);
\draw[thick,dotted] (0,4) -- (0,0);
\draw[thick,dotted] (4,1) -- (4,0);
\draw [thick,dotted] (6,4) --(6,0);

\draw [to-to,blue] (-2,0.2) --(0,0.2);
\draw (-1,0.2) node[right, color=blue, above] {$J_i$};
\draw [to-to,blue] (4,0.2) --(6,0.2);
\draw (5,0.2) node[right, color=blue, above] {$J_i$};

\draw [to-to,blue] (0,0.2) --(4,0.2);
\draw (2,0.2) node[right, color=blue, above] {$\bar{J}_i$};

\draw [to-to,red] (4,1) --(7,1);
\draw (5.5,1) node[right, color=red, above] {$\theta_i$};

\draw [above] (5.3,3.2) node {\small{$\mu_i$}};
\draw [above] (-1,2.8) node {\small{$\mu_i$}};
\draw (2.2,2.9) node {\small{$-\mu_i$}};

\draw (6,4)  node {\small{$\clubsuit$}};
\draw  (0,4) node {\small{$\clubsuit$}};

\draw (-2,1) node {$\diamondsuit$};
\draw (4,1) node {$\diamondsuit$};

\end{tikzpicture}}}

  \caption{Queueing Dynamics of Exhaustive Equilibrium Outcomes under \on-\off Duration $(L_i,\bar{L}_i)$ with $\mu_i\leq \lambda_i$.}
\label{fig:fixed_duration_exhaustive_equilibrium_mui_smaller}
\end{figure*}

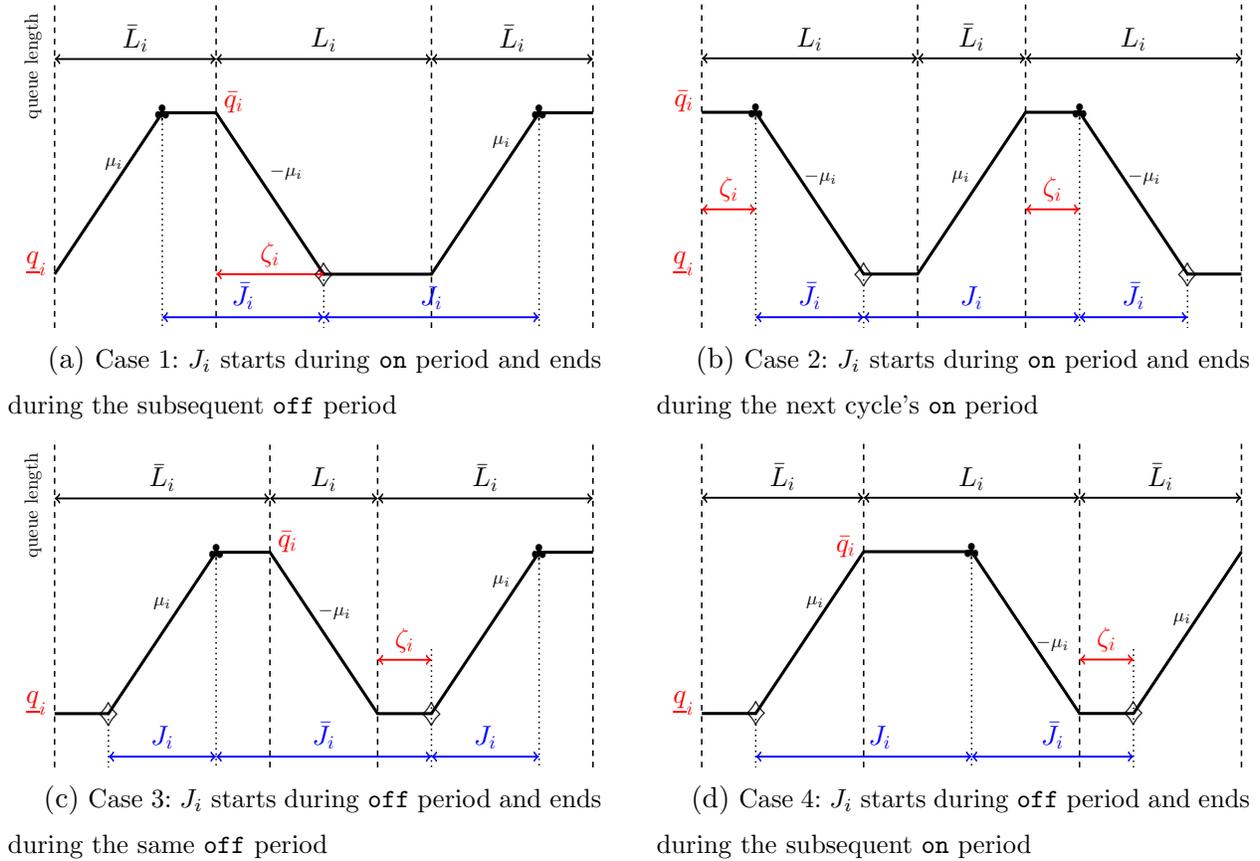
\begin{figure*}[ht]
\centering
 \subfloat[{\small{Case 1: $J_i$ starts during \on period and ends during the subsequent \off period}}]{  \resizebox{.47\textwidth}{!}{\begin{tikzpicture}[font=\large, line width=1pt] 

\draw (-3.4,6) node[left, rotate=90] {{\small{queue length}}};

\draw[thick,dashed] (0,0) -- (0,6);
\draw[thick,dashed] (-3,0) -- (-3,6);
\draw[thick,dashed] (4,0) -- (4,6);
\draw[thick,dashed] (7,0) -- (7,6);

\draw	(0,4.2) node[anchor=west,red] {$\Bar{q}_i$};
\draw	(-3,1.2) node[anchor=east,red] {$\underline{q}_i$};

\draw[line width=1.6pt] (-3,1) -- (-1,4) --(0,4) -- (2,1) -- (4,1) -- (6,4) --(7,4);

\draw [above] (5.3,3.2) node {\small{$\mu_i$}};
\draw [above] (-1.9,2.8) node {\small{$\mu_i$}};
\draw (1.3,2.9) node {\small{$-\mu_i$}};

\draw [to-to,black] (4,5) --(7,5);
\draw (5.5, 5) node[right, color=black, above] {$\bar{L}_i$};

\draw [to-to,black] (-3,5) --(0,5);
\draw (-1.5, 5) node[right, color=black, above] {$\bar{L}_i$};
\draw [to-to,black] (4,5) --(0,5);
\draw (2,5) node[right, color=black, above] {$L_i$};

\draw[thick,dotted] (-1,4) -- (-1,0);
\draw[thick,dotted] (2,1) -- (2,0);
\draw [thick,dotted] (6,4) --(6,0);

\draw [to-to,red] (2,1) --(0,1);
\draw (1,1) node[right, color=red, above] {$\zeta_i$};

\draw [to-to,blue] (-1,0.2) --(2,0.2);
\draw (0.5,0.2) node[right, color=blue, above] {$\bar{J}_i$};

\draw [to-to,blue] (6,0.2) --(2,0.2);
\draw (4,0.2) node[right, color=blue, above] {$J_i$};

\draw (-1,4) node {\small{$\clubsuit$}};
\draw (6,4) node {\small{$\clubsuit$}};

\draw (2,1) node {$\diamondsuit$};

\end{tikzpicture}}}
 \hfill
  \subfloat[{\small{Case 2: $J_i$ starts during \on period and ends during the next cycle's \on period}}]{  \resizebox{.47\textwidth}{!}{\begin{tikzpicture}[font=\large, line width=1pt] 

\draw[thick,dashed] (-3,0) -- (-3,6);
\draw[thick,dashed] (1,0) -- (1,6);
\draw[thick,dashed] (3,0) -- (3,6);

\draw[thick,dashed] (7,0) -- (7,6);

\draw	(-3,4.2) node[anchor=east,red] {$\Bar{q}_i$};
\draw	(-3,1.2) node[anchor=east,red] {$\underline{q}_i$};

\draw[line width=1.6pt] (-3,4) -- (-2,4) -- (0,1) -- (1,1) -- (3,4) -- (4,4) -- (6,1) --  (7,1);

\draw (-0.8,2.8) node {\small{$-\mu_i$}};
\draw (5.2,2.8) node {\small{$-\mu_i$}};

\draw (1.8,2.8) node {\small{$\mu_i$}};

\draw [to-to,black] (1,5) --(3,5);
\draw (2, 5) node[right, color=black, above] {$\bar{L}_i$};

\draw [to-to,black] (-3,5) --(1,5);
\draw (-1, 5) node[right, color=black, above] {$L_i$};
\draw [to-to,black] (3,5) --(7,5);
\draw (5,5) node[right, color=black, above] {$L_i$};

\draw[thick,dotted] (-2,4) -- (-2,0);
\draw[thick,dotted] (0,1) -- (0,0);
\draw[thick,dotted] (4,4) -- (4,0);
\draw[thick,dotted] (6,1) -- (6,0);

\draw [to-to,blue] (-2,0.2) --(0,0.2);
\draw (-1,0.2) node[right, color=blue, above] {$\bar{J}_i$};
\draw [to-to,blue] (4,0.2) --(6,0.2);
\draw (5,0.2) node[right, color=blue, above] {$\bar{J}_i$};

\draw [to-to,blue] (0,0.2) --(4,0.2);
\draw (2,0.2) node[right, color=blue, above] {$J_i$};

\draw [to-to,red] (-3,2.2) --(-2,2.2);
\draw (-2.5,2.2) node[right, color=red, above] {$\zeta_i$};
\draw [to-to,red] (3,2.2) --(4,2.2);
\draw (3.5,2.2) node[right, color=red, above] {$\zeta_i$};

\draw (-2,4) node {\small{$\clubsuit$}};
\draw (4,4) node {\small{$\clubsuit$}};
\draw (0,1) node {$\diamondsuit$};
\draw (6,1) node {$\diamondsuit$};
\end{tikzpicture}}}
\vspace{0.7em}
   \subfloat[{\small{Case 3: $J_i$ starts during \off period and ends during the same \off period}}]{  \resizebox{.47\textwidth}{!}{\begin{tikzpicture}[font=\large, line width=1pt] 

\draw (-3.4,6) node[left, rotate=90] {{\small{queue length}}};

\draw[line width=1.6pt](-3,1) -- (-2,1) -- (0,4) --(1,4) -- (3,1) -- (4,1) -- (6,4) --(7,4);

\draw[thick,dashed] (-3,0) -- (-3,6);
\draw[thick,dashed] (1,0) -- (1,6);
\draw[thick,dashed] (3,0) -- (3,6);
\draw[thick,dashed] (7,0) -- (7,6);

\draw [to-to,black] (-3,5) --(1,5);
\draw (-1, 5) node[right, color=black, above] {$\bar{L}_i$};
\draw [to-to,black] (1,5) --(3,5);
\draw (2,5) node[right, color=black, above] {$L_i$};
\draw [to-to,black] (3,5) --(7,5);
\draw (5, 5) node[right, color=black, above] {$\bar{L}_i$};

\draw	(1,4.2) node[anchor=west,red] {$\Bar{q}_i$};
\draw	(-3,1.2) node[anchor=east,red] {$\underline{q}_i$};

\draw[thick,dotted] (-2,1) -- (-2,0);
\draw[thick,dotted] (0,4) -- (0,0);
\draw[thick,dotted] (4,2.2) -- (4,0);
\draw [thick,dotted] (6,4) --(6,0);

\draw [to-to,red] (3,2) --(4,2);
\draw (3.5,2) node[right, color=red, above] {$\zeta_i$};

\draw [to-to,blue] (-2,0.2) --(0,0.2);
\draw (-1,0.2) node[right, color=blue, above] {$J_i$};
\draw [to-to,blue] (4,0.2) --(6,0.2);
\draw (5,0.2) node[right, color=blue, above] {$J_i$};

\draw [to-to,blue] (0,0.2) --(4,0.2);
\draw (2,0.2) node[right, color=blue, above] {$\bar{J}_i$};

\draw [above] (5.3,3.2) node {\small{$\mu_i$}};
\draw [above] (-1,2.8) node {\small{$\mu_i$}};
\draw (2.2,2.9) node {\small{$-\mu_i$}};

\draw (6,4) node {\small{$\clubsuit$}};
\draw (0,4) node {\small{$\clubsuit$}};

\draw (-2,1) node {$\diamondsuit$};
\draw (4,1) node {$\diamondsuit$};

\end{tikzpicture}}}
 \hfill
  \subfloat[{\small{Case 4: $J_i$ starts during \off period and ends during the subsequent \on period}}]{  \resizebox{.47\textwidth}{!}{\begin{tikzpicture}[font=\large, line width=1pt]

\draw[line width=1.6pt] (-3,1) -- (-2,1) -- (0,4) -- (2,4)-- (4,1) -- (5,1) -- (7,4);

\draw[thick,dashed] (0,0) -- (0,6);
\draw[thick,dashed] (-3,0) -- (-3,6);
\draw[thick,dashed] (4,0) -- (4,6);
\draw[thick,dashed] (7,0) -- (7,6);

\draw [to-to,black] (4,5) --(7,5);
\draw (5.5, 5) node[right, color=black, above] {$\bar{L}_i$};

\draw [to-to,black] (-3,5) --(0,5);
\draw (-1.5, 5) node[right, color=black, above] {$\bar{L}_i$};
\draw [to-to,black] (4,5) --(0,5);
\draw (2,5) node[right, color=black, above] {$L_i$};

\draw	(0,4.1) node[anchor=east,red] {$\Bar{q}_i$};
\draw	(-3,1.2) node[anchor=east,red] {$\underline{q}_i$};

\draw[thick,dotted] (-2,1) -- (-2,0);
\draw[thick,dotted] (2,4) -- (2,0);
\draw[thick,dotted] (5,0) -- (5,2.2);

\draw [to-to,blue] (-2,0.2) --(2,0.2);
\draw (0.3,0.2) node[right, color=blue, above] {$J_i$};

\draw [to-to,blue] (2,0.2) --(5,0.2);
\draw (3.5,0.2) node[right, color=blue, above] {$\bar{J}_i$};

\draw [to-to,red] (4,2) --(5,2);
\draw (4.5,2) node[right, color=red, above] {$\zeta_i$};

\draw [above] (-0.9,2.8) node {\small{$\mu_i$}};

\draw [above] (5.9,2.5) node {\small{$\mu_i$}};
\draw (3.5,2.3) node {\small{$-\mu_i$}};

\draw (2,4) node {\small{$\clubsuit$}};

\draw (5,1) node {$\diamondsuit$};
\draw (-2,1) node {$\diamondsuit$};

\end{tikzpicture}}}
  
  \vspace{0.5em}
  \caption{Queueing Dynamics of Non-Exhaustive Equilibrium Outcomes under \on-\off Duration $(L_i,\bar{L}_i)$ with $\mu_i\leq \lambda_i$.}
\label{fig:fixed_duration_nonexhaustive_equilibrium_mui_smaller}
\end{figure*}

\noindent \textit{Endogenous Regime.}
When there exists a queue $i \in \N$ such that $\lambda_i \geq \mu_i$, the equilibrium under an exhaustive service policy exhibits peculiar behavior. Specifically, because the server never departs from a non-empty queue under an exhaustive service policy, if $\lambda_i \geq \mu_i$ for some queue $i$, the equilibrium joining rate of this queue eventually equals the service rate. As a result, the server never clears this queue and continuously serves it. This scenario is not of interest, and we omit further discussion of the optimal exhaustive service policy in this case.

\subsubsection{Optimal Exogenous Durations}

In this section, we present an LP formulation to determine the optimal exogenous \on-\off durations. Define $\N_1$ as the set of queues for which the arrival rate is strictly less than the service rate, and let $\N_2$ be the rest of the queues.

For each queue $i \in \N_1$, the induced equilibrium post-clearance duration $T_i(L_i,\bar{L}_i)$ is given by \Cref{coro:fixed_duration_post_clearance_duration} based on the analysis in Sections \ref{sec:equilibrium_analysis_fixed_duration} and \ref{sec:opt_fixed_duration}. For each queue $i \in \N_2$, we have $T_i(L_i,\bar{L}_i)=0$ as a direct corollary of \Cref{thm:fixed_duration_equilibrium_mui_smaller}. Using the same reasoning as in the derivation of the LP formulation \eqref{prob:LP_fixed_duration}, we obtain the following LP, which finds the optimal \on-\off durations when there exist queues with arrival rates higher than their service rates.
\begin{align}
\max_{x,\bar{x}\in \mathbb{R}^n,y\in\mathbb{R}^{|\N_1|}_{+},g\in\mathbb{R}_{+}}&  
   \mu^\top x -\sum_{i\in\N_1}(\mu_i-\lambda_i)\cdot y_i  \nonumber\\
\textrm{s.t.}\qquad & y_i\geq 0,\quad \forall i\in\N_1, \nonumber\\
  & y_i \geq x_i -  \frac{\lambda_i\theta_i}{\mu_i-\lambda_i} g ,\quad \forall i\in\N_1, \nonumber \\
  & y_i \geq x_i -  \frac{\lambda_i}{\mu_i-\lambda_i}\bar{x}_i ,\quad \forall i\in\N_1, \nonumber \\
  & x_i + \bar{x}_i = 1,\quad \forall i \in \N, \nonumber \\
 & \bar{x}_i=1^\top x_{-i}+1^\top \tau \cdot g,\quad \forall i \in \N. \nonumber
\end{align}

Let $(x^\ast,\bar{x}^\ast,y^\ast,g^\ast)$ be the optimal solution of the above program. The optimal exogenous \on-\off durations $(L^\ast,\bar{L}^\ast)$ and the corresponding equilibrium post-clearance duration $T_i(L_i^\ast,\bar{L}_i^\ast)$ can then be recovered as follows:
\begin{align*}
   & L_i^\ast = \frac{x_i^\ast }{ g^\ast}, \quad 
    \bar{L}_i^\ast =   \frac{\bar{x}_i^\ast}{g^\ast},\quad \forall i \in \N;\\
   & T_i(L_i^\ast,\bar{L}_i^\ast) = \frac{y_i^\ast}{g^\ast},\quad \forall i \in \N_1 ,\quad T_i(L_i^\ast,\bar{L}_i^\ast) = 0,\quad \forall i \in \N_2.
\end{align*}

\subsection{Optimal Exogenous Durations in the Single-Queue Setting}
\label{app_subsec:optimal_exo_n=1}

We now discuss the optimal exogenous \on-\off durations in the single-queue setting that incorporates additional side constraints. We impose the following two constraints on the cycle durations:
\begin{enumerate}
\item \textit{Work limit}: 
\begin{align}
    L_i \leq L^{\textrm{max}}_i,\label{constraint:work_limit}
\end{align}
where $L^{\textrm{max}}_i > 0$ denotes the maximum allowable \on duration within a single cycle.
\item \textit{Forced off}: 
\begin{align}
    \bar{L}_i \geq \beta_i L_i,\label{constraint:forced_vacation} 
\end{align}
where $\beta_i > 0$ specifies the minimum required ratio between the \off and \on durations in a cycle.
\end{enumerate}
These constraints ensure that the \on duration remains bounded and that the \off duration is sufficiently long relative to the preceding \on period.
 For example, an electric vehicle must periodically stop operating to recharge and typically does not resume service until its battery is sufficiently replenished.

Putting it all together, in the single-queue setting the planner is to solve problem \eqref{prob:fixed_duration} with \( n=1 \), subject to the additional operational constraints \eqref{constraint:work_limit}--\eqref{constraint:forced_vacation} but without the constraints \eqref{constraint:one_cycle} (since there is only one queue).

\begin{proposition}
\label{prop:vacation_opt_duration}
For $n=1$, the following is an optimal solution:
  \begin{align*}
        L_i^\ast = \min\left\{ \frac{\lambda_i}{\mu_i-\lambda_i}\theta_i,~ L^{\textrm{max}}_i\right\} \quad \text{and} \quad \bar{L}_i^\ast = \beta_i L_i^\ast.
    \end{align*}
\begin{enumerate}
    \item When $\beta_i> (\mu_i-\lambda_i)/\lambda_i$, we have the equilibrium post-clearance duration $T_i(L_i^\ast,\bar{L}_i^\ast)=0$ and the optimal throughput $\mu_i/(1+\beta_i)$.
    \begin{enumerate}
        \item If $L^{\textrm{max}}_i\geq \lambda_i\theta_i/(\mu_i-\lambda_i)$, the corresponding optimal equilibrium outcome is exhaustive.
        \item  If $L^{\textrm{max}}_i< \lambda_i\theta_i/(\mu_i-\lambda_i)$ and $\beta_i\in ((\mu_i-\lambda_i)/\lambda_i, \theta_i/L^{\textrm{max}}_i)$, the corresponding equilibrium outcome is non-exhaustive. If $L^{\textrm{max}}_i< \lambda_i\theta_i/(\mu_i-\lambda_i)$ and $\beta_i\geq \theta_i/L^{\textrm{max}}_i$, the corresponding equilibrium outcome is exhaustive.
    \end{enumerate}
    \item  When $\beta_i\leq (\mu_i-\lambda_i)/\lambda_i$, we have $T_i(L_i^\ast,\bar{L}_i^\ast)=L_i^\ast - \lambda_i \bar{L}_i^\ast/(\mu_i-\lambda_i)$ and the system achieves the first-best throughput $\lambda_i$.
    Besides, the corresponding optimal equilibrium outcome is exhaustive.
\end{enumerate}
\end{proposition}

The relationship between $\beta_i$ and the ratio $(\mu_i-\lambda_i)/\lambda_i$ is crucial. To see that, recall from our earlier equilibrium characterization (Theorems \ref{thm:equilibrium_fixedduration_exhaustive} and \ref{thm:fixed_duration_nonexhaustive_equilibrium}) that achieving the first-best throughput requires keeping the \off duration sufficiently small relative to the \on duration. Specifically, to attain the first-best outcome, the equilibrium must be the scenario depicted in \Cref{subfig:exh_case4}, where $\bar{L}_i \leq L_i(\mu_i-\lambda_i)/\lambda_i$. Thus, if we impose the constraint \eqref{constraint:forced_vacation}, i.e., $\bar{L}_i \geq \beta_i L_i$, it follows that to achieve the first-best throughput, we must have $\beta_i \leq (\mu_i-\lambda_i)/\lambda_i$.
Intuitively, a larger $\beta_i$ restricts the planner’s ability to shorten \off durations, thus harming system throughput.

Another observation is that the optimal equilibrium outcome can be non-exhaustive. Specifically, it occurs when $ L^{\textrm{max}}_i < \lambda_i \theta_i/(\mu_i - \lambda_i)$ and $\beta_i \in ((\mu_i - \lambda_i)/\lambda_i, \theta_i/L^{\textrm{max}}_i)$. In this case, the work limit $L^{\textrm{max}}_i$ is small, resulting in a relatively short \on duration. As a result, customers cannot be cleared before the server leaves the system, leading to a non-exhaustive equilibrium outcome.

\subsection{Endogenous Regime with Two Queues}
\label{app_sec:subsec_twoqueues}

In this subsection, we focus on the endogenous regime with only two queues. 
We will first characterize the equilibrium under the pure exhaustive service policy $\pe$, which is useful in characterizing the optimal exhaustive service policy (see \Cref{alg:opt}), and then derive the optimal post-clearance duration in closed-form expressions.

Since there are two queues, the equilibrium set $\I(0)$ under the pure-exhaustive service policy $\pe$ can take one of $2^2$ possible forms, namely $\I(0)=\emptyset$, $\I(0)=\{1\}$, $\I(0)=\{2\}$, or $\I(0)=\{1,2\}$, depending on the model parameters. We completely and analytically characterize the conditions on the model parameters that lead to each possibility and determine the optimal post-clearance duration, as detailed below.

\begin{proposition}[Equilibrium and Optimal Exhaustive Service Policy for Two Queues]
\label{prop:twoqueues}
Suppose there are two queues.
\begin{enumerate}
    \item When $\theta_1\leq 1^\top \tau+\lambda_2\theta_2/(\mu_2-\lambda_2)$ and $\theta_2\leq 1^\top \tau+\lambda_1\theta_1/(\mu_1-\lambda_1)$,
    \begin{align*}
        \I(0)=\emptyset,\quad \alpha(0)=1.
    \end{align*}
    Furthermore, 
    \begin{enumerate}
        \item  when $\lambda_1> \lambda_2$: if $\theta_1<1^\top \tau+ (\lambda_2/(\mu_2-\lambda_2))(1- \mu_2/\lambda_1)\cdot \theta_2$, always serving queue $1$ is optimal; otherwise, the pure exhaustive service policy $\pe$ is optimal.
    \item    
    when $\lambda_2\geq \lambda_1$: if $\theta_2<1^\top \tau+ (\lambda_1/(\mu_1-\lambda_1))(1-\mu_1/\lambda_2)\cdot \theta_1$, always serving queue $2$ is optimal; otherwise, $\pe$ is optimal.
    \end{enumerate}

    \item When $\theta_1 > 1^\top \tau+ \lambda_2\theta_2/(\mu_2-\lambda_2)$ and $\theta_2\leq 1^\top \tau(1-\rho_2)/(1-(\rho_1+\rho_2))$,
    \begin{align}
     \I(0)=\{1\},\quad \alpha_1(0) = \left(1^\top \tau+\frac{\lambda_2}{\mu_2-\lambda_2}\theta_2\right)\cdot \frac{1}{\theta_1} ,\quad \alpha_2(0)=1.  \nonumber
    \end{align}
    Furthermore, 
    \begin{enumerate}
        \item when $\theta_2< 1^\top \tau$:  if $\lambda_2>\lambda_1$ and $\theta_2<1^\top \tau+\frac{\lambda_2-\mu_1}{\lambda_2}\cdot \frac{\lambda_1}{\mu_1-\lambda_1}\theta_1$, always serving queue $2$ is optimal; otherwise, $T_1^\ast=0$ and $T_2^\ast=\theta_1-\frac{\lambda_2}{\mu_2-\lambda_2}\theta_2-1^\top \tau$.
        \item  when $\theta_2\geq 1^\top \tau$: 
        if $\lambda_2>\mu_1$ and $\theta_2<\frac{\mu_1(\lambda_2-\lambda_1) (\mu_2-\lambda_2)}{\lambda_2\cdot \left[\mu_2(\mu_1-\lambda_1)-\mu_1(\lambda_2-\lambda_1)\right]}\cdot 1^\top \tau$,
        always serving queue $2$ is optimal; otherwise, $\pe$ is optimal.
    \end{enumerate}

    \item When $\theta_2 > 1^\top \tau+\lambda_1\theta_1/(\mu_1-\lambda_1)$ and $\theta_1\leq 1^\top \tau(1-\rho_1)/(1-(\rho_1+\rho_2))$,
    \begin{align}
     \I(0)=\{2\},\quad \alpha_2(0) = \left(1^\top \tau+\frac{\lambda_1}{\mu_1-\lambda_1}\theta_1\right)\cdot \frac{1}{\theta_2} ,\quad \alpha_1(0)=1.  \nonumber
    \end{align}
    Furthermore, 
    \begin{enumerate}
        \item when $\theta_1< 1^\top \tau$:  if $\lambda_1>\lambda_2$ and $\theta_1<1^\top \tau+\frac{\lambda_1-\mu_2}{\lambda_1}\cdot \frac{\lambda_2}{\mu_2-\lambda_2}\theta_2$, always serving queue $1$ is optimal; otherwise, $T_2^\ast=0$ and $T_1^\ast=\theta_2-\frac{\lambda_1}{\mu_1-\lambda_1}\theta_1-1^\top \tau$.
        \item  when $\theta_1\geq 1^\top \tau$: if $\lambda_1>\mu_2$ and $\theta_1<\frac{\mu_2(\lambda_1-\lambda_2) (\mu_1-\lambda_1)}{\lambda_1\cdot \left[\mu_1(\mu_2-\lambda_2)-\mu_2(\lambda_1-\lambda_2)\right]}\cdot 1^\top \tau$, always serving queue $1$ is optimal; otherwise, $\pe$ is optimal.
    \end{enumerate}

     \item When $\theta_1>1^\top \tau(1-\rho_1)/(1-(\rho_1+\rho_2))$ and $\theta_2> 1^\top \tau(1-\rho_2)/(1-(\rho_1+\rho_2))$,
    \begin{align*}
        \I(0)=\{1,2\},\quad \alpha_1(0)=\frac{1^\top \tau}{\theta_1}\cdot\frac{1-\rho_1}{1-(\rho_1+\rho_2)},\quad \alpha_2(0)=\frac{1^\top \tau}{\theta_2}\cdot\frac{1-\rho_2}{1-(\rho_1+\rho_2)}.
    \end{align*}
    Furthermore, $\pe$ is optimal.
\end{enumerate}
\end{proposition}

We will visualize \Cref{prop:twoqueues} and provide further discussions about the equilibrium set $\I(0)$ and the optimal post-clearance durations in the following two subsections.
Before that, in
\Cref{fig:greatest_element_polyhedron}, we illustrate \Cref{lem:lp_equilirbium} regarding the equivalence between the greatest element and the equilibrium via an example of the pure exhaustive service policy. The conditions leading to each scenario are given by \Cref{prop:twoqueues}.

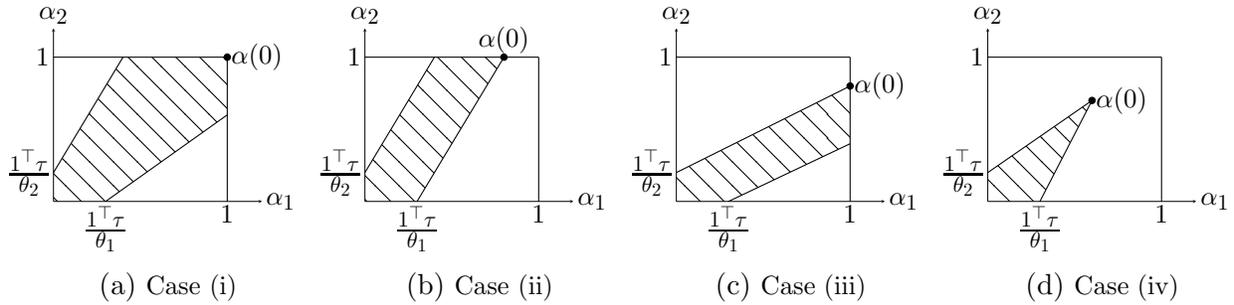
\begin{figure*}[ht]
\centering
    \subfloat[{\small{Case (i)}}]
  {  \resizebox{.24\textwidth}{!}{\begin{tikzpicture}
\begin{axis}
[
    axis lines = left,
    xlabel = \Huge\(\alpha_1\),
    ylabel = \Huge\(\alpha_2\),
    xmin=0, xmax=1.2,
    ymin=0, ymax=1.2,
    xtick=\empty,
    ytick=\empty,
    xticklabels={,,},
    yticklabels={,,},
    tick style={draw=none},
    line width=0.5pt,
    clip=false, %
    every axis x label/.style={at={(ticklabel* cs:1.0)}, anchor=west},
    every axis y label/.style={at={(ticklabel* cs:1.0)}, anchor=south}
]

\coordinate (A) at (0,0.2);
\coordinate (B) at (0.3,0);
\coordinate (E) at (0,1);
\coordinate (F) at (1,0);
\coordinate (G) at (1,1); %

\draw node[left] at (A)  {\Huge$\frac{1^\top \tau}{\theta_2}$};
\draw node[below] at (B)  {\Huge$\frac{1^\top \tau}{\theta_1}$};
\draw node[left] at (E)  {\Huge$1$};
\draw node[below] at (F)  {\Huge$1$};

\draw[line width=1pt] (E) -- (G);
\draw[line width=1pt] (F) -- (G);

\coordinate (C) at (0.4,1);
\coordinate (D) at (1,0.6);

\draw[line width=1pt] (A) -- (C);
\draw[line width=1pt] (B) -- (D);

\fill[pattern=light south west lines] 
(A) -- (C) --(G) -- (D) -- (B) --(0,0) -- cycle;

\filldraw[black] (G) circle (3pt); 
\draw node[right] at (G)  {\Huge$\alpha(0)$};

\end{axis}
\end{tikzpicture}}}
  \hfill
    \subfloat[{\small{Case (ii)}}]
  {  \resizebox{.24\textwidth}{!}{\begin{tikzpicture}
\begin{axis}
[
    axis lines = left,
    xlabel = \Huge\(\alpha_1\),
    ylabel = \Huge\(\alpha_2\),
    xmin=0, xmax=1.2,
    ymin=0, ymax=1.2,
    xtick=\empty,
    ytick=\empty,
    xticklabels={,,},
    yticklabels={,,},
    tick style={draw=none},
    line width=0.5pt,
    clip=false, %
    every axis x label/.style={at={(ticklabel* cs:1.0)}, anchor=west},
    every axis y label/.style={at={(ticklabel* cs:1.0)}, anchor=south}
]

\coordinate (A) at (0,0.2);
\coordinate (B) at (0.3,0);
\coordinate (E) at (0,1);
\coordinate (F) at (1,0);
\coordinate (G) at (1,1); %

\draw node[left] at (A)  {\Huge$\frac{1^\top \tau}{\theta_2}$};
\draw node[below] at (B)  {\Huge$\frac{1^\top \tau}{\theta_1}$};
\draw node[left] at (E)  {\Huge$1$};
\draw node[below] at (F)  {\Huge$1$};

\draw[line width=1pt] (E) -- (G);
\draw[line width=1pt] (F) -- (G);

\coordinate (C) at (0.4,1);
\coordinate (D) at (0.8,1);

\draw[line width=1pt] (A) -- (C);
\draw[line width=1pt] (B) -- (D);

\fill[pattern=light south west lines] 
(A) -- (C) --(G) -- (D) -- (B) --(0,0) -- cycle;

\filldraw[black] (D) circle (3pt); 
\draw node[above] at (D)  {\Huge$\alpha(0)$};

\end{axis}
\end{tikzpicture}}}
  \hfill
        \subfloat[{\small{Case (iii)}}]
    {  \resizebox{.24\textwidth}{!}{\begin{tikzpicture}
\begin{axis}
[
    axis lines = left,
    xlabel = \Huge\(\alpha_1\),
    ylabel = \Huge\(\alpha_2\),
    xmin=0, xmax=1.2,
    ymin=0, ymax=1.2,
    xtick=\empty,
    ytick=\empty,
    xticklabels={,,},
    yticklabels={,,},
    tick style={draw=none},
    line width=0.5pt,
    clip=false, %
    every axis x label/.style={at={(ticklabel* cs:1.0)}, anchor=west},
    every axis y label/.style={at={(ticklabel* cs:1.0)}, anchor=south}
]

\coordinate (A) at (0,0.2);
\coordinate (B) at (0.3,0);
\coordinate (E) at (0,1);
\coordinate (F) at (1,0);
\coordinate (G) at (1,1); %

\draw node[left] at (A)  {\Huge$\frac{1^\top \tau}{\theta_2}$};
\draw node[below] at (B)  {\Huge$\frac{1^\top \tau}{\theta_1}$};
\draw node[left] at (E)  {\Huge$1$};
\draw node[below] at (F)  {\Huge$1$};

\draw[line width=1pt] (E) -- (G);
\draw[line width=1pt] (F) -- (G);

\coordinate (C) at (1,0.8);
\coordinate (D) at (1,0.4);

\draw[line width=1pt] (A) -- (C);
\draw[line width=1pt] (B) -- (D);

\fill[pattern=light south west lines] 
(A) -- (C) --(G) -- (D) -- (B) --(0,0) -- cycle;

\filldraw[black] (C) circle (3pt); 
\draw node[right] at (C)  {\Huge$\alpha(0)$};

\end{axis}
\end{tikzpicture}}}
  \hfill
        \subfloat[{\small{Case (iv)}}]
    {  \resizebox{.24\textwidth}{!}{\begin{tikzpicture}
\begin{axis}
[
    axis lines = left,
    xlabel = \Huge\(\alpha_1\),
    ylabel = \Huge\(\alpha_2\),
    xmin=0, xmax=1.2,
    ymin=0, ymax=1.2,
    xtick=\empty,
    ytick=\empty,
    xticklabels={,,},
    yticklabels={,,},
    tick style={draw=none},
    line width=0.5pt,
    clip=false, %
    every axis x label/.style={at={(ticklabel* cs:1.0)}, anchor=west},
    every axis y label/.style={at={(ticklabel* cs:1.0)}, anchor=south}
]

\coordinate (A) at (0,0.2);
\coordinate (B) at (0.3,0);
\coordinate (E) at (0,1);
\coordinate (F) at (1,0);
\coordinate (G) at (1,1); %

\draw node[left] at (A)  {\Huge$\frac{1^\top \tau}{\theta_2}$};
\draw node[below] at (B)  {\Huge$\frac{1^\top \tau}{\theta_1}$};
\draw node[left] at (E)  {\Huge$1$};
\draw node[below] at (F)  {\Huge$1$};

\draw[line width=1pt] (E) -- (G);
\draw[line width=1pt] (F) -- (G);

\coordinate (C) at (0.6,0.7);
\coordinate (D) at (0.6,0.7);

\draw[line width=1pt] (A) -- (C);
\draw[line width=1pt] (B) -- (D);

\fill[pattern=light south west lines] 
(A) -- (C) --(G) -- (D) -- (B) --(0,0) -- cycle;

\filldraw[black] (D) circle (3pt); 
\draw node[right] at (D)  {\Huge$\alpha(0)$};

\end{axis}
\end{tikzpicture}}}
  \caption{Equilibrium Variable $\alpha(0)$ under the Pure Exhaustive Service Policy $\pe$ as the Greatest Element of the Polyhedron $\mathcal{A}(b(0),A)$ Defined in \Cref{lem:lp_equilirbium} with $n=2$. 
  } 
  \label{fig:greatest_element_polyhedron}
\end{figure*}

\subsubsection{How Equilibrium Varies with Waiting Patience}

In this subsection, we examine in detail the equilibrium variables under the pure exhaustive service policy. \Cref{fig:twoqueue_pe_eq} visualizes \Cref{prop:twoqueues} about the equilibrium set $\I(0)$ for different waiting patience parameters.

In \Cref{subfig:twoqueue_pe_eq_a}, we have $\rho_1 + \rho_2 < 1$, while in \Cref{subfig:twoqueue_pe_eq_b}, $\rho_1 + \rho_2 \geq 1$. Here, $\rLL$ denotes the regime where both queues have relatively low waiting patience, resulting in $\I(0)=\emptyset$, i.e., the not-joining duration $\bar{J}_i>0$ for both queues. In the $\rHL$ regime, where queue 1’s customers have relatively high waiting patience and queue 2’s customers have relatively low waiting patience, we have $\I(0)=\{1\}$. That is, all arrivals of queue 1 join, while queue 2 still exhibits a not-joining duration $\bar{J}_2>0$. Similarly, in the $\rLH$ regime, $\I(0)=\{2\}$, so all arrivals of queue 2 join, and queue 1 experiences a not-joining duration $\bar{J}_1>0$. Finally, in the $\rHH$ regime, where both queues have relatively high waiting patience, we have $\I(0)=\{1,2\}$, i.e., all arrivals of both queues join, which is the first-best outcome. Importantly, when $\rho_1 + \rho_2 < 1$, it is possible to achieve the first-best throughput $\lambda_1+\lambda_2$, as noted in \Cref{lem:equilibriumset};
whereas this does not occur in \Cref{subfig:twoqueue_pe_eq_b} with $\rho_1+\rho_2\geq 1$.

\begin{figure*}[ht]
\centering
  \subfloat[{\small{$\rho_1+\rho_2<1$}} \label{subfig:twoqueue_pe_eq_a}]{  \resizebox{.52\textwidth}{!}{\begin{tikzpicture}
\begin{axis}[
    axis lines = left,
    xlabel = \(\theta_1\),
    ylabel = \(\theta_2\),
    xmin=0, xmax=10.5,
    ymin=0, ymax=10.5,
    xtick=\empty,
    ytick=\empty,
    xticklabels={,,},
    yticklabels={,,},
    tick style={draw=none},
    line width=0.5pt,
    clip=false, %
    every axis x label/.style={at={(ticklabel* cs:1.0)}, anchor=west},
    every axis y label/.style={at={(ticklabel* cs:1.0)}, anchor=south}
]

\addplot [
    domain=0:2.5, 
    samples=100, 
    color=black,
    line width=1pt
]
{x+1.5};

\addplot [
    domain=1.5:2.5, 
    samples=100, 
    color=black,
    line width=1pt
]
{4*x-6};

\addplot [
    domain=2.5:10, 
    samples=100, 
    color=black,
    line width=1pt
]
{4};

\draw [line width=1pt] (axis cs:2.5,4) -- (axis cs:2.5,10);

\draw [line width=0.5pt,dashed] (axis cs:0,4) -- (axis cs:2.5,4);
\draw [line width=0.5pt,dashed] (axis cs:2.5,0) -- (axis cs:2.5,4);

\draw node[below] at (axis cs:3.2,0)  {$\frac{1^\top \tau(1-\rho_1)}{1-(\rho_1+\rho_2)}$};
\draw node[left] at (axis cs:0,4)  {$\frac{1^\top \tau(1-\rho_2)}{1-(\rho_1+\rho_2)}$};
\draw node[left] at (axis cs:-0.2,1.5)  {$1^\top \tau$};
\draw node[below] at (axis cs:1.3,-0.2)  {{\small{$1^\top \tau$}}};

\draw node[above] at (axis cs:1,2.7)  {\small$\ell_2$};
\draw node[right] at (axis cs:1.5,0.5)  {\small$\ell_1$};

\draw node at (axis cs:1,1.5)  {\small$\rLL$};
\draw node at (axis cs:1.25,6)  {\small$\rLH$};
\draw node at (axis cs:6,2)  {\small$\rHL$};
\draw node at (axis cs:6,7)  {\small$\rHH$};

\filldraw[black] (axis cs:2.5,4) circle (2pt); 
\draw node[right,above] at (axis cs:2.8,4)  {\small$\mathsf{Q}$};

\end{axis}
\end{tikzpicture}}}
  \hfill
\subfloat[{\small{$\rho_1+\rho_2\geq 1$}} \label{subfig:twoqueue_pe_eq_b}]{  \resizebox{.46\textwidth}{!}{\begin{tikzpicture}
\begin{axis}[
    axis lines = left,
    xlabel = \(\theta_1\),
    ylabel = \(\theta_2\),
    xmin=0, xmax=10.5,
    ymin=0, ymax=10.5,
    xtick=\empty,
    ytick=\empty,
    xticklabels={,,},
    yticklabels={,,},
    tick style={draw=none},
    line width=0.5pt,
    clip=false, %
    every axis x label/.style={at={(ticklabel* cs:1.0)}, anchor=west},
    every axis y label/.style={at={(ticklabel* cs:1.0)}, anchor=south}
]

\addplot [
    domain=0:4.5, 
    samples=100, 
    color=black,
    line width=1pt, %
]
{2*x+1};

\addplot [
    domain=1:23/3, 
    samples=100, 
    color=black,
    line width=1pt, %
]
{1.5*x-1.5};

\draw node[above left] at (axis cs:3.8,8)  {$\ell_2$};

\draw node[above] at (axis cs:6.5,8.5)  {$\ell_1$};

\draw node[left] at (axis cs:-0.3,1)  {$1^\top \tau$};
\draw node[below] at (axis cs:1,-0.3)  {$1^\top \tau$};

\draw node at (axis cs:3,4.5)  {\small$\rLL$};

\draw node at (axis cs:1.3,7)  {\small$\rLH$};

\draw node at (axis cs:7,2.5)  {\small$\rHL$};

\end{axis}
\end{tikzpicture}}}
  \caption{Equilibrium of the Pure Exhaustive Service Policy $\pe$ under Different Waiting Patience. \textit{Note}. $(\mathsf{L}\mathsf{L})$: $\I(0)=\emptyset$; \quad $(\mathsf{H}\mathsf{L})$: $\I(0)=\{1\}$; \quad $(\mathsf{L}\mathsf{H})$: $\I(0)=\{2\}$; \quad $(\mathsf{H}\mathsf{H})$: $\I(0)=\{1,2\}$.
The slope of the line $\ell_1$ is $(\mu_2-\lambda_2)/\lambda_2$, and of the line $\ell_2$ is $\lambda_1/(\mu_1-\lambda_1)$.
  }
  \label{fig:twoqueue_pe_eq}
\end{figure*}
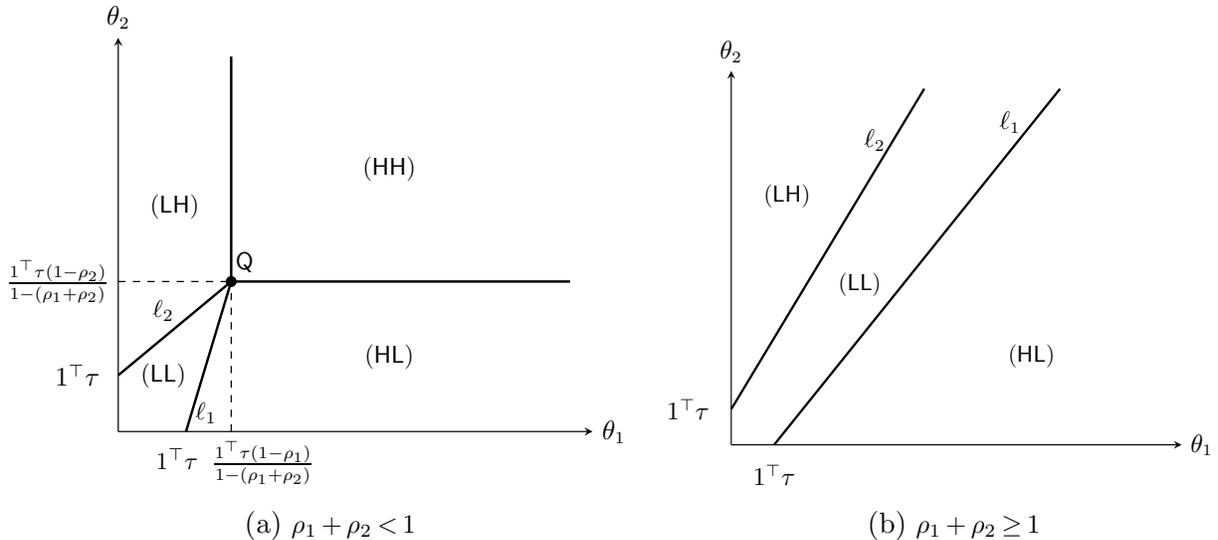

Different regions in \Cref{fig:twoqueue_pe_eq} cover all possible waiting patience values, which confirms the existence of an equilibrium for any given set of model parameters. Along the boundaries, multiple equilibrium sets $\I(0)$ may exist; however, they induce the same equilibrium variable $\alpha(0)$. For example, at point $\mathsf{Q}$ in \Cref{subfig:twoqueue_pe_eq_a}, there are four distinct equilibrium sets $\I(0)$. In fact, at $\mathsf{Q}$, we have $\alpha=1$ and $A\alpha=1^\top\tau1$, so all constraints \eqref{con:ir} and \eqref{con:period} are binding. Thus, any subset of $\N=\{1,2\}$ qualifies as an equilibrium set, yet the equilibrium variable $\alpha(0)$ remains unique.

\begin{corollary} \label{coro:twoqueue_pureex_utilizationrates} 
For $n=2$, under the pure exhaustive service policy, given the same waiting patience for both queues, if the queue with the lower utilization rate $\rho_i$ is all-joining at equilibrium, i.e., the equilibrium not-joining duration of this queue is zero, then the queue with a higher utilization rate must also be all-joining.
\end{corollary}

Corollary~\ref{coro:twoqueue_pureex_utilizationrates} may seem counterintuitive, since a queue with a lower utilization rate is typically expected to have shorter waiting times. However, under the pure exhaustive service policy, the server spends more time in a queue with a higher utilization rate because it takes longer to clear that queue. Consequently, the server remains in the high-utilization-rate queue for a larger fraction of time, which leads to shorter waiting times for its customers. In effect, Corollary~\ref{coro:twoqueue_pureex_utilizationrates} implies that the externality imposed on the low-utilization-rate queue by the high-utilization-rate queue is greater than the reverse.

\subsubsection{How Optimal Exhaustive Service Policy Varies with Waiting Patience}

\Cref{fig:twoqueue_opt_ex} visualizes how the optimal exhaustive service policy, i.e., the optimal post-clearance duration, varies with waiting patience. We focus on the case where $\rho_1+\rho_2>1$ in both Figures \ref{subfig:twoqueue_opt_ex_a} and \ref{subfig:twoqueue_opt_ex_b}; the case $\rho_1+\rho_2 \leq 1$ is similar but does not have the red region.

In the black regime, which corresponds to parts of the $\rLL$ and $\rHL$ regions in \Cref{fig:twoqueue_pe_eq}, the optimal exhaustive service policy always serves the queue with the higher arrival rate (queue 2). In the magenta regime, a subset of the $\rLH$ region in \Cref{fig:twoqueue_pe_eq}, the optimal policy sets $T_1^\ast=\theta_2 - 1^\top \tau - c_2$ and $T_2^\ast=0$ to force $\alpha_2(0)$ to equal one. In the red regime, which encompasses parts of the $\rLL$, $\rLH$, $\rHL$, and the entire $\rHH$ regions in \Cref{fig:twoqueue_pe_eq}, the pure exhaustive service policy $\pe$ is optimal. Finally, in the blue regime (if it exists), a subset of the $\rHL$ region in \Cref{fig:twoqueue_pe_eq}, the optimal policy requires $T_1^\ast=0$ and $T_2^\ast=\theta_1 - 1^\top \tau - c_2$ so that $\alpha_1(0)=1$.

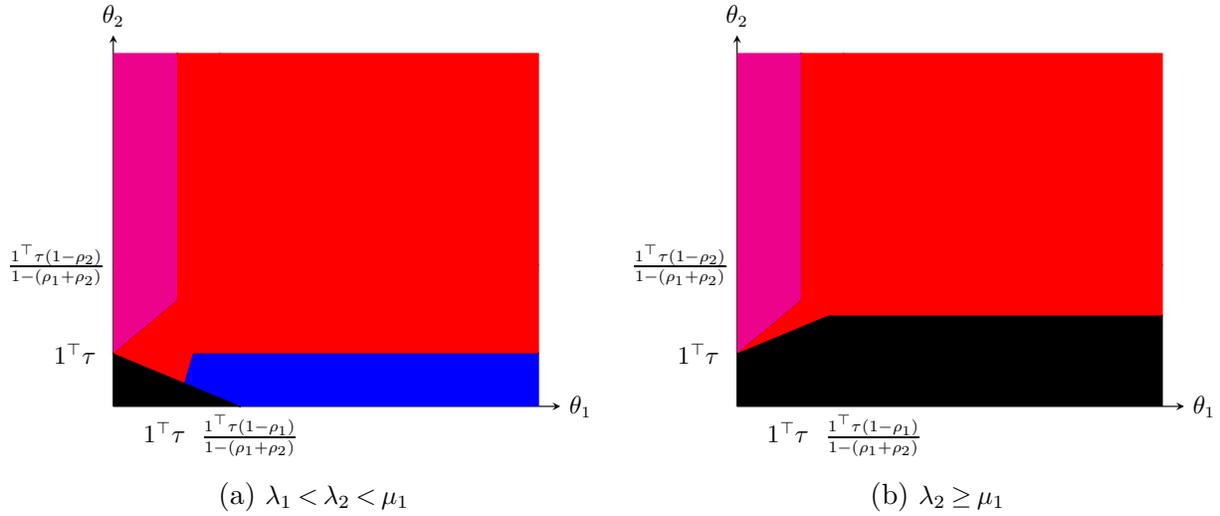
\begin{figure*}[ht]
\centering
  \subfloat[{\small{$\lambda_1<\lambda_2<\mu_1$}} 
  \label{subfig:twoqueue_opt_ex_a}]{  \resizebox{.49\textwidth}{!}{\begin{tikzpicture}
\begin{axis}[
    axis lines = left,
    xlabel = \(\theta_1\),
    ylabel = \(\theta_2\),
    xmin=0, xmax=10.5,
    ymin=0, ymax=10.5,
    xtick=\empty,
    ytick=\empty,
    xticklabels={,,},
    yticklabels={,,},
    tick style={draw=none},
    line width=0.5pt,
    clip=false, %
    every axis x label/.style={at={(ticklabel* cs:1.0)}, anchor=west},
    every axis y label/.style={at={(ticklabel* cs:1.0)}, anchor=south}
]

\addplot [
    domain=0:1.5, 
    samples=100, 
    color=black,
    line width=1pt
]
{x+1.5};
\addplot [
    domain=1.5:2.5, 
    samples=100, 
    color=black,
    line width=0.5pt,
    dashed
]
{x+1.5};

\addplot [
    domain=5/3:15/8, 
    samples=100, 
    color=black,
    line width=1pt
]
{4*x-6};
\addplot [
    domain=1.5:5/3, 
    samples=100, 
    color=black,
    line width=0.5pt,
    dashed
]
{4*x-6};
\addplot [
    domain=15/8:2.5, 
    samples=100, 
    color=black,
    line width=0.5pt,
    dashed
]
{4*x-6};

\addplot [
    domain=0:3, 
    samples=100, 
    color=black,
    line width=1pt
]
{-0.5*x+1.5};

\addplot [
    domain=2.5:10, 
    samples=100, 
    color=black,
    line width=0.5pt,
    dashed
]
{4};

\draw [dashed,line width=0.5pt] (axis cs:2.5,4) -- (axis cs:2.5,10);

\draw [line width=1pt] (axis cs:1.5,3) -- (axis cs:1.5,10);
\draw [line width=0.5pt,dashed] (axis cs:1.5,3) -- (axis cs:1.5,0);
\draw [line width=1pt] (axis cs:15/8,1.5) -- (axis cs:10,1.5);
\draw [line width=0.5pt,dashed] (axis cs:15/8,1.5) -- (axis cs:0,1.5);

\draw [line width=0.5pt,dashed] (axis cs:0,4) -- (axis cs:2.5,4);
\draw [line width=0.5pt,dashed] (axis cs:2.5,0) -- (axis cs:2.5,4);

\draw node[below] at (axis cs:3.2,0)  {$\frac{1^\top \tau(1-\rho_1)}{1-(\rho_1+\rho_2)}$};
\draw node[left] at (axis cs:0,4)  {$\frac{1^\top \tau(1-\rho_2)}{1-(\rho_1+\rho_2)}$};
\draw node[left] at (axis cs:-0.2,1.5)  {$1^\top \tau$};
\draw node[below] at (axis cs:1.2,-0.2)  {$1^\top \tau$};

\draw node[above] at (axis cs:0.8,2.4)  {\small$\ell_2$};
\draw node[right] at (axis cs:1.8,2.6)  {\small$\ell_1$};
\draw node at (axis cs:0.5,1)  {\small$\ell_3$};

\fill[black,opacity=0.5] 
(axis cs:0,0) --(axis cs:3,0) --(axis cs:0,1.5)  -- cycle;

\fill[blue,opacity=0.5] 
(axis cs:10,0) --(axis cs:10,1.5) --(axis cs:15/8,1.5) --(axis cs:5/3,2/3) --(axis cs:3,0)  -- cycle;

\fill[red,opacity=0.5] 
(axis cs:0,1.5) -- (axis cs:5/3,2/3) -- (axis cs:15/8,1.5)
-- (axis cs:10,1.5) -- (axis cs:10,10) -- (axis cs:1.5,10)  -- (axis cs:1.5,3)    -- cycle;

\fill[magenta,opacity=0.5] 
(axis cs:1.5,3) --(axis cs:1.5,10)--(axis cs:0,10) --(axis cs:0,1.5)  -- cycle;

\end{axis}

\end{tikzpicture}}}
  \hfill
\subfloat[{\small{$\lambda_2\geq \mu_1$}}
 \label{subfig:twoqueue_opt_ex_b}]{  \resizebox{.49\textwidth}{!}{\begin{tikzpicture}
\begin{axis}[
    axis lines = left,
    xlabel = \(\theta_1\),
    ylabel = \(\theta_2\),
    xmin=0, xmax=10.5,
    ymin=0, ymax=10.5,
    xtick=\empty,
    ytick=\empty,
    xticklabels={,,},
    yticklabels={,,},
    tick style={draw=none},
    line width=0.5pt,
    clip=false, %
    every axis x label/.style={at={(ticklabel* cs:1.0)}, anchor=west},
    every axis y label/.style={at={(ticklabel* cs:1.0)}, anchor=south}
]

\addplot [
    domain=0:1.5, 
    samples=100, 
    color=black,
    line width=1pt
]
{x+1.5};
\addplot [
    domain=1.5:2.5, 
    samples=100, 
    color=black,
    line width=0.5pt,
    dashed
]
{x+1.5};

\addplot [
    domain=1.5:2.5, 
    samples=100, 
    color=black,
     line width=0.5pt,
    dashed
]
{4*x-6};

\addplot [
    domain=0:15/7, 
    samples=100, 
    color=black,
    line width=1pt
]
{0.5*x+1.5};

\addplot [
    domain=15/7:10, 
    samples=100, 
    color=black,
    line width=0.75pt
]
{18/7};

\addplot [
    domain=2.5:10, 
    samples=100, 
    color=black,
    line width=0.5pt,
    dashed
]
{4};

\draw [dashed,line width=0.5pt] (axis cs:2.5,4) -- (axis cs:2.5,10);

\draw [line width=1pt] (axis cs:1.5,3) -- (axis cs:1.5,10);
\draw [line width=0.5pt,dashed] (axis cs:1.5,3) -- (axis cs:1.5,0);

\draw [line width=0.5pt,dashed] (axis cs:0,4) -- (axis cs:2.5,4);
\draw [line width=0.5pt,dashed] (axis cs:2.5,0) -- (axis cs:2.5,4);

\draw node[below] at (axis cs:3.2,0)  {$\frac{1^\top \tau(1-\rho_1)}{1-(\rho_1+\rho_2)}$};
\draw node[left] at (axis cs:0,4)  {$\frac{1^\top \tau(1-\rho_2)}{1-(\rho_1+\rho_2)}$};
\draw node[left] at (axis cs:-0.2,1.5)  {$1^\top \tau$};
\draw node[below] at (axis cs:1.2,-0.2)  {$1^\top \tau$};

\draw node[above] at (axis cs:0.8,2.4)  {\small$\ell_2$};
\draw node[right] at (axis cs:1.5,1.5)  {\small$\ell_1$};
\draw node[above] at (axis cs:0.5,0.8)  {\small$\ell_3$};

\fill[black,opacity=0.5] 
(axis cs:0,0) --(axis cs:10,0) --(axis cs:10,18/7) --(axis cs:15/7,18/7) --(axis cs:0,1.5)  -- cycle;

\fill[red,opacity=0.5] 
(axis cs:0,1.5) -- (axis cs:15/7,18/7) -- (axis cs:10,18/7) -- (axis cs:10,10) -- (axis cs:1.5,10)  -- (axis cs:1.5,3)    -- cycle;

\fill[magenta,opacity=0.5] 
(axis cs:1.5,3) --(axis cs:1.5,10)--(axis cs:0,10) --(axis cs:0,1.5)  -- cycle;

\end{axis}
\end{tikzpicture}}}
  \caption{Optimal Exhaustive Service Policy under Different Waiting Patience ($\rho_1+\rho_2>1$). \textit{Note}. Black regime: $T_2^\ast=+\infty, T_1^\ast=0$; 
Magenta regime: $T_1^\ast=\theta_2-1^\top \tau - c_2,T_2^\ast=0$;
Red regime: $T^\ast=0$; 
Blue regime: $T_1^\ast=0,T_2^\ast=\theta_1-1^\top\tau-c_2$.
The slope of the line $\ell_1$ is $(\mu_2-\lambda_2)/\lambda_2$, the slope of the line $\ell_2$ is $\lambda_1/(\mu_1-\lambda_1)$, and the slope of $\ell_3$ is $((\lambda_2-\mu_1)/\lambda_2)\cdot(\lambda_1/(\mu_1-\lambda_1))$. 
Lines $\ell_1$ and $\ell_2$ are the same as the ones in \Cref{fig:twoqueue_pe_eq}.
}
\label{fig:twoqueue_opt_ex}
\end{figure*}

Comparing \Cref{subfig:twoqueue_opt_ex_a} and ($\lambda_2 < \mu_1$)  \Cref{subfig:twoqueue_opt_ex_b} ($\lambda_2 \geq \mu_1$) reveals notable changes: the blue regime disappears, the red regime shrinks, the black regime expands, while the magenta regime remains unchanged. 
Another interesting observation from \Cref{fig:twoqueue_opt_ex} is that the boundary
of the optimal $T^\ast$ appears piecewise linear, despite the nonlinear nature of throughput. As shown in the proof of \Cref{prop:twoqueues}, certain nonlinear terms cancel out under equilibrium conditions, resulting in a piecewise linear relationship between $T^\ast$ and the waiting patience $\theta$.

\begin{remark}
For $n=2$, the optimal post-clearance durations $T^\ast$ are piecewise linear in the waiting patience $\theta$, given all other model parameters.
\end{remark}

\subsection{The Inverses of Matrix $A$}
\label{app_sub:inversematrix}

\begin{lemma}%
\label{lem:inverse_mat}
Given any non-empty set $\I\subseteq \N$, the matrix $A_{\I}$, where $A$ is defined in \eqref{eq:matrix_A_b}, is invertible if and only if $1^\top\rho_{\I} \neq 1$.
If $A_{\I}$ is invertible, the inverse matrix $A_{\I}^{-1}$ is given as follows: 
\begin{subequations}
        \begin{align}
    \textrm{for $i=j\in\I$:\quad}&A_{\I}^{-1}(i,j) = \frac{\mu_i-\lambda_i}{\mu_i\theta_i}\cdot\frac{1}{\bar{\rho}_{\I}}\cdot \left(1 - \sum_{k\in\I\setminus\{i\}}\rho_k\right),\label{eq:invermatrix_1} \\
    \textrm{for $i\neq j\in\I$:\quad}&A_{\I}^{-1}(i,j) = \frac{\mu_i-\lambda_i}{\mu_i\theta_i}\cdot \frac{\rho_j}{\bar{\rho}_{\I}},\label{eq:invermatrix_2}
    \end{align}
\end{subequations}
where $\bar{\rho}_{\I}:=1-1^\top\rho_{\I}$. Besides, the sum of elements in the $i^{\textrm{th}}$ row of $A_{\I}^{-1}$ is
    \begin{align*}
     \frac{\mu_i-\lambda_i}{\mu_i\theta_i}\cdot \frac{1}{\bar{\rho}_{\I}}.
    \end{align*}
\end{lemma}

\medskip

\section{Proofs}
\label{app:pf}

\subsection{Proofs in \Cref{sec:exogenous_regime}}

\begin{proof}{\textbf{Proof of \Cref{lem:waiting_time}.}}
This directly follows from the discussion after \Cref{lem:waiting_time}.  \Halmos
\end{proof}

\medskip

\begin{proof}{\textbf{Proof of \Cref{lem:fixed_duration_waitingtime}.}}

Define the \off waiting time at time $t$ as:
\begin{align*}
 W_i^{\off}(t)
 = (1-\iota_i(t))\cdot r_i(t) + z_i(t)\cdot \bar{L}_i.
\end{align*}
By the waiting time equation \eqref{eq:waiting_time_fixed_duration}, we have
\begin{align}
\label{eq:waiting_off_waiting}
 W_i(t) = \frac{q_i(t)}{\mu_i} + W_i^{\off}(t).   
\end{align}

We use the following results, whose proofs are provided immediately after the proof of \Cref{lem:fixed_duration_waitingtime}. We use customer arriving at time $t^+$ ($t^-$) to denote the customer arriving immediately after (before) time $t$.

\begin{lemma}
\label{lem:waiting_time_at_t_fixed_duration}
The following statements hold true.
 \begin{enumerate}
\item If the waiting time is not continuous at time $t$, then $W_i(t^+)=W_i(t^-) + \bar{L}_i$.
\item 
The waiting time $W_i(t)$ is continuous at time $t$ if and only if the \off waiting time $W_i^{\off}(t)$ is continuous at time $t$.
\item The \off waiting time $W_i^{\off}(t)$ is continuous at time $t$ if and only if customers arriving at $t^-$ and $t^+$ would be served in the same \on duration should they join the system.
 \end{enumerate}
\end{lemma}

\begin{lemma}
\label{lem:waiting_time_derivative}
 At any time $t$ where $W_i^\prime(t)$ exists, 
 \begin{align}
 \nonumber
  W_i^\prime(t)=
  \begin{cases}
     -1 + \rho_i \cdot \mathbbm{1}\{f_i(t)=1\}<0,
     \quad & \textrm{if $t$ does not belong to a post-clearance period},\\
   0,  & \textrm{otherwise}.
  \end{cases}
 \end{align}

\end{lemma}

\begin{lemma}
\label{coro:waiting_time_decreasing_same_onduration}
If a customer arriving at time $t_1$ is served in the same \on period as a customer arriving at a later time $t_2>t_1$, then $W_i(t)$ is decreasing in $t\in [t_1,t_2]$.
\end{lemma}

\medskip

We now prove each part of \Cref{lem:fixed_duration_waitingtime}. 
We may sometimes use \textit{customer $t$} to refer to a customer arriving at time $t$ when there is no confusion.

\begin{enumerate}

\item 

Without loss of generality, assume that in one cycle there are $K\geq 1$ consecutive contiguous not-joining and joining periods:
\[
\bar{\J}_i^{(1)},\,\J_i^{(1)},\,\bar{\J}_i^{(2)},\,\J_i^{(2)},\,\ldots,\,\bar{\J}_i^{(K)},\,\J_i^{(K)},
\]
where $\bar{\J}_i^{(k)}$ and $\J_i^{(k)}$ denote the $k^{\mathsf{th}}$ not-joining and joining \textit{periods} in one cycle (assuming that $t=0$ marks the beginning of a not-joining period), respectively. In other words, the total not-joining duration in one cycle is $\bar{J}_i = \sum_{k=1}^K \bar{J}_i^{(k)}$ and the total joining duration is $J_i = \sum_{k=1}^K J_i^{(k)}$, where $\bar{J}_i^{(k)}$ and $J_i^{(k)}$ denote the $k^{\mathsf{th}}$ not-joining and joining \textit{durations}, respectively. Proving \Cref{lem:fixed_duration_waitingtime} (i) is equivalent to showing that $K=1$.

\smallskip
\noindent
\underline{\textit{Claim}}: The customer arriving at $\bar{t} = \bar{J}_i^{(1)}$ must be served at the \emph{beginning} of some \on period, denoted as $\LL_i^{(\bar{t})}$.

To see why, note that by the definition of $\bar{t}$, the customer arriving at $\bar{t}^+$ joins the system, while the one arriving at $\bar{t}^-$ does not since we assume $t=0$ is the beginning of a not-joining period. Hence $W_i(\bar{t}^+)\le W_i(\bar{t}^-)$. By \Cref{lem:waiting_time_at_t_fixed_duration} (i), the waiting time must be continuous at $\bar{t}$, or else $W_i(\bar{t}^+)$ would exceed $W_i(\bar{t}^-)$ by $\bar{L}_i$, which contradicts $W_i(\bar{t}^+)\le W_i(\bar{t}^-)$. Besides, $W_i(\bar{t}) = \theta_i$ since the time $\bar{t}$ is the transition from not joining to joining.

Suppose instead that customer $\bar{t}$ is \emph{not} served at the beginning of some $\LL_i^{(\bar{t})}$. Then, there must be another customer, arriving earlier than the customer $\bar{t}$, who is served at the beginning of $\LL_i^{(\bar{t})}$. Denote by $\tilde{t}$ the arrival time of this earlier customer. By \Cref{coro:waiting_time_decreasing_same_onduration}, $W_i(t)$ is strictly decreasing for $t\in(\tilde{t}, \bar{t})$. Since $W_i(\bar{t})=\theta_i$, any arrival in $(\tilde{t},\bar{t})$ would have waiting time strictly less than $\theta_i$. However, this contradicts the fact that customers arriving in $(0,\bar{t})$ did not join since $(0,\bar{t})$ is the not-joining period $\bar{\J}_i^{(1)}$.
Therefore, the customer arriving at $\bar{t}=\bar{J}_i^{(1)}$ must be served at the beginning of $\LL_i^{(\bar{t})}$.

Next, define $\hat{t} := \bar{J}_i^{(1)} + J_i^{(1)}$, which is the transition 
time from the first joining period $\J_i^{(1)}$ to the second not-joining period $\bar{\J}_i^{(2)}$.
Since the customer arriving at $\hat{t}^+$ does not join, while the one at $\hat{t}^-$ does, we have $W_i(\hat{t}^-)\le \theta_i < W_i(\hat{t}^+)$ and, by \Cref{lem:waiting_time_at_t_fixed_duration} (i), $W_i(\hat{t}^+)=W_i(\hat{t}^-)+\bar{L}_i$. By \Cref{lem:waiting_time_at_t_fixed_duration} (iii), customer $\hat{t}^-$ and customer $\hat{t}^+$ (if she joins) are served in different \on periods.
To make this happen, the customer $\hat{t}^-$ must be served at the \emph{end} of some \on period.

Combining these observations, serving the arrivals from $\bar{t}$ to $\hat{t}$ (i.e., those who joined during the first joining period $\J_i^{(1)}$) requires \emph{at least one complete} \on period, i.e., requires at least $L_i$ units of \on time. 
On the other hand, by the periodicity assumption,
serving all customers who arrive during the \textit{entire} joining periods uses \textit{exactly} $L_i$ units of \on time.  
Therefore, there can only be one contiguous joining period in one cycle, i.e.,  $K=1$. 

\item 
At time $t=0^-$, customers join, but at $t=0^+$, they do not (since $t=0$ is the beginning of the not-joining period). Hence $W_i(0^+)> \theta_i \ge W_i(0^-)$. By \Cref{lem:waiting_time_at_t_fixed_duration} (i), any upward jump in $W_i(t)$ must be exactly $\bar{L}_i$, so $W_i(0^+) = W_i(0^-)+\bar{L}_i$.

\item 
From part (i), all customers who join during the joining period are served within the same \on period. Repeated applications of \Cref{lem:waiting_time_at_t_fixed_duration} (ii) and (iii) then imply that $W_i(t)$ is continuous in $t \in (\bar{J}_i,\,\bar{J}_i + J_i)$.

We now verify the continuity at $t=\bar{J}_i$ and in $(0,\bar{J}_i)$:

\begin{enumerate}
\item 
At $t=\bar{J}_i^-$, customers do not join, while at $t=\bar{J}_i^+$ they do. Thus, $W_i(\bar{J}_i^+)\le W_i(\bar{J}_i^-)$, and by \Cref{lem:waiting_time_at_t_fixed_duration} (i), $W_i(\bar{J}_i^+)=W_i(\bar{J}_i^-)$ (otherwise there would be an upward jump of $\bar{L}_i$).

\item 
From part~(i), the first joining customer arriving at $(\bar{J}_i,\bar{J}_i + J_i)$, i.e., the contiguous joining period, is served at the beginning of an \on period, and the last one is served at the end of the same \on period. By the first-come-first-serve rule, any customer arriving during the not-joining period $(0,\bar{J}_i)$ (if they were to join) would also be served at the beginning of some \on period, ensuring the continuity of $W_i(t)$ by \Cref{lem:waiting_time_at_t_fixed_duration} (iii). 
\end{enumerate}

\item 
The derivative part in
\Cref{lem:fixed_duration_waitingtime} (iv) follows directly from \Cref{lem:waiting_time_derivative}.
The result that any non-empty post-clearance period must lie at the end of the joining period follows from \Cref{thm:equilibrium_fixedduration_exhaustive}, which characterizes the exhaustive equilibrium outcome and confirms this property. Note that the proof of \Cref{thm:equilibrium_fixedduration_exhaustive} does not rely on \Cref{lem:fixed_duration_waitingtime}.
\Halmos
\end{enumerate}

\end{proof}

\medskip

\begin{proof}{\textbf{Proof of \Cref{lem:waiting_time_at_t_fixed_duration}.}}
We prove each result separately.
\begin{enumerate}
    \item By \Cref{lem:waiting_time}, if $t$ is not a transition epoch of the server's status, the term $q_i(t)/\mu_i + (1-\iota_i(t))\cdot r_i(t)$ is continuous at time $t$. Then, any discontinuity of the waiting time in \eqref{eq:waiting_time_fixed_duration} can only come from $z_i(t)$, which involves the floor function. 
    For the time $t$ that is not the transition epoch of the server's status, the term inside the floor function is continuous at time $t$.
    Then, if $z_i(t^+)\neq z_i(t^-)$, we must have $z_i(t^+) = z_i(t^-)+1$. This leads to $W_i(t^+) = W_i(t^-)+\bar{L}_i$ if the waiting time is discontinuous at time $t$.

We now discuss the cases when $t$ is a transition time of the server's status.
\begin{enumerate}
    \item  Consider the transition time $t$ from \on to \off, i.e., $(\iota_i(t^-),r_i(t^-))=(1,0)$ and $(\iota_i(t^+),r_i(t^+)) = (0,\bar{L}_i)$. Since the queue length is continuous, i.e., $q_i(t^-)=q_i(t^+)$, by \eqref{eq:waiting_time_fixed_duration}, we have
\begin{align}
\label{eq:waitingtime_leftlimit_on_off_1}
 W_i(t^-)
 &=\frac{q_i(t^-)}{\mu_i} + (1-1)\cdot 0+ \underbrace{\left(\left\lfloor\frac{q_i(t^-)}{\mu_i L_i}+ 1\right\rfloor \right)}_{z_i(t^-)}\cdot \bar{L}_i
 =\frac{q_i(t)}{\mu_i}+
 \underbrace{\left(\left\lfloor\frac{q_i(t^-)}{\mu_i L_i}\right\rfloor + 1\right)\cdot \bar{L}_i}_{\textrm{\off waiting time}}
 ,
\end{align}
and
\begin{align}
\label{eq:waitingtime_leftlimit_on_off_2}
W_i(t^+)
=\frac{q_i(t^+)}{\mu_i} +  (1-0)\cdot \bar{L}_i +  \underbrace{\left\lfloor\frac{q_i(t^+)}{\mu_i L_i}\right\rfloor}_{z_i(t^+)}\cdot \bar{L}_i
=\frac{q_i(t)}{\mu_i}+
\underbrace{\left(\left\lfloor\frac{q_i(t^+)}{\mu_i L_i}\right\rfloor + 1\right)\cdot \bar{L}_i}_{\textrm{\off waiting time}}.
\end{align}
Combining \eqref{eq:waitingtime_leftlimit_on_off_1} and \eqref{eq:waitingtime_leftlimit_on_off_2}, if $W_i(t^+)\neq W_i(t^-)$, we must have $\left\lfloor\frac{q_i(t^-)}{\mu_i L_i}\right\rfloor \neq \left\lfloor\frac{q_i(t^+)}{\mu_i L_i}\right\rfloor$. 
By the property of the floor function, we have $\left\lfloor\frac{q_i(t^+)}{\mu_i L_i}\right\rfloor = \left\lfloor\frac{q_i(t^-)}{\mu_i L_i}\right\rfloor + 1$. 
This leads to $W_i(t^+)=W_i(t^-) + \bar{L}_i$.

\item Consider the transition time $t$ from \off to \on, i.e., $(\iota_i(t^-),r_i(t^-))=(0,0)$ and $(\iota_i(t^+),r_i(t^+))=(1,L_i)$. 
By \eqref{eq:waiting_time_fixed_duration}, we have
\begin{align}
\label{eq:waitingtime_leftlimit_off_on_1}
 W_i(t^-)
 &=\frac{q_i(t^-)}{\mu_i} + (1-0)\cdot 0+ \underbrace{\left\lfloor\frac{q_i(t^-)}{\mu_i L_i}\right\rfloor}_{z_i(t^-)}\cdot \bar{L}_i
 =\frac{q_i(t)}{\mu_i}+
 \underbrace{\left\lfloor\frac{q_i(t^-)}{\mu_i L_i}\right\rfloor\cdot \bar{L}_i}_{\textrm{\off waiting time}},
\end{align}
and 
\begin{align}
\label{eq:waitingtime_leftlimit_off_on_2}
 W_i(t^+)
 &=\frac{q_i(t^+)}{\mu_i} + (1-1)\cdot L_i+ \underbrace{\left\lfloor\frac{q_i(t^+)}{\mu_i L_i} + \frac{1\cdot (L_i-L_i)}{L_i}\right\rfloor}_{z_i(t^+)}\cdot \bar{L}_i
 =\frac{q_i(t)}{\mu_i}+
 \underbrace{\left\lfloor\frac{q_i(t^+)}{\mu_i L_i}\right\rfloor\cdot \bar{L}_i}_{\textrm{\off waiting time}}.
\end{align}
Combining \eqref{eq:waitingtime_leftlimit_off_on_1} and \eqref{eq:waitingtime_leftlimit_off_on_2}, if $W_i(t^+)\neq W_i(t^-)$, we have $\left\lfloor\frac{q_i(t^+)}{\mu_i L_i}\right\rfloor \neq \left\lfloor\frac{q_i(t^-)}{\mu_i L_i}\right\rfloor$. 
This leads to $\left\lfloor\frac{q_i(t^+)}{\mu_i L_i}\right\rfloor = \left\lfloor\frac{q_i(t^-)}{\mu_i L_i}\right\rfloor + 1$.
We thus have $W_i(t^+)=W_i(t^-) + \bar{L}_i$.
\end{enumerate}
Putting it all together, we have \Cref{lem:waiting_time_at_t_fixed_duration} (i).

\item \Cref{lem:waiting_time_at_t_fixed_duration} (ii) follows directly from \eqref{eq:waiting_off_waiting} and the fact that $q_i(t)$ is always continuous.

\item \Cref{lem:waiting_time_at_t_fixed_duration} (iii) follows from the reasoning of \Cref{lem:waiting_time}: the only source of discontinuity in $W_i(t)$ is that the customer arriving at $t^+$ must wait for one more \off duration compared with a customer arriving at $t^-$. This is reflected in the floor function in $z_i(t)$ given by \eqref{eq:z_i(t)}.\Halmos
\end{enumerate}
\end{proof}

\medskip

\begin{proof}{\textbf{Proof of \Cref{lem:waiting_time_derivative}.}}
We first prove the first part.
For any $t$ not in a post-clearance period, we have
\begin{align}
\label{eq:q_i_derivative}
 q_i^\prime(t) = -\mu_i\,\iota_i(t) + \lambda_i\,\mathbbm{1}\{f_i(t)=1\}.
\end{align}
This is because when the server is serving queue $i$, the service rate is $\mu_i$, leading to a decrease of $-\mu_i$ per unit time. When a customer arriving at time $t$ joins the system (i.e., $f_i(t)=1$), there is an increment of rate $\lambda_i$. Because $z_i(t)$ in \eqref{eq:z_i(t)} involves a floor operator, its derivative (if exists) is zero, so $z_i^\prime(t)=0$.

Using \eqref{eq:q_i_derivative} and $z_i^\prime(t)=0$, when $\iota_i(t)=1$, differentiating \eqref{eq:waiting_time_fixed_duration} with respect to $t$ gives
\begin{align*}
 W_i^\prime(t) = \frac{q_i^\prime(t)}{\mu_i} 
 = -1 + \rho_i\,\mathbbm{1}\{f_i(t)=1\}.
\end{align*}
When $\iota_i(t)=0$, differentiating \eqref{eq:waiting_time_fixed_duration} with respect to $t$ yields
\begin{align*}
 W_i^\prime(t) 
 = \frac{q_i^\prime(t)}{\mu_i} + r_i^\prime(t)
 = \rho_i\,\mathbbm{1}\{f_i(t)=1\} - 1,
\end{align*}
where the last equality holds because $r_i^\prime(t)=-1$. Since $\lambda_i<\mu_i$ by our assumption in the main text, we have $W_i^\prime(t)<1$. 

The second part is straightforward since, during the post-clearance duration, the waiting time is always zero, and so is its derivative. This completes the proof.
\Halmos
\end{proof}

\medskip

\begin{proof}{\textbf{Proof of \Cref{coro:waiting_time_decreasing_same_onduration}.}}
By repeatedly applying \Cref{lem:waiting_time_at_t_fixed_duration} (ii) and (iii) to $t\in(t_1,t_2)$, we see that $W_i(t)$ is continuous over this interval. 
On the other hand, by following the same argument used in the proof of \Cref{lem:waiting_time_derivative}, one can show that the left-hand derivative $W'(t^-)$ for $t \in (t_1, t_2)$, which exists everywhere, %
is non-positive. Combining these observations, we conclude that $W_i(t)$ is (weakly) decreasing in $t \in (t_1, t_2)$.
\Halmos
\end{proof}

\medskip

\begin{proof}{\textbf{Proof of \Cref{thm:equilibrium_fixedduration_exhaustive}.}}
We prove the results case by case.

\begin{enumerate}
    \item \textit{When $\bar{L}_i \geq \theta_i$, i.e., Figures \ref{subfig:exh_case1}--\ref{subfig:exh_case3}.}

    We claim that at equilibrium, when the server leaves queue~$i$, the queue length must be zero. Suppose not, and let the remaining queue length be $\delta q > 0$. Then, for the last customer among these $\delta q$ customers, the waiting time must be greater than $\bar{L}_i + \delta q/\mu_i > \bar{L}_i \geq \theta_i$, contradicting the equilibrium constraint. Thus, the equilibrium outcome must be exhaustive.

Then, using the waiting time formula \eqref{eq:waiting_time_fixed_duration}, one can show that the waiting time decreases from the epoch when the server leaves queue $i$ until time $\diamondsuit$. At $\diamondsuit$, the waiting time is exactly $\theta_i$ by \Cref{lem:fixed_duration_waitingtime}\,(iii). Hence, no arriving customers during these intervals join.

Following time $\diamondsuit$, if an arriving customer, assuming they join, can be served before the server leaves queue $i$ again, they would be willing to join since $W_i(\diamondsuit)=\theta_i$ and, by \Cref{coro:waiting_time_decreasing_same_onduration}, the waiting time decreases. Thus, customers will join until time $\clubsuit$, which satisfies that the customer arriving at $\diamondsuit$ is served at the beginning of an \on period while the customer arriving at $\clubsuit$ is served at the end of the same \on period. In other words, customers arriving during $(\diamondsuit,\clubsuit)$ require exactly one complete \on duration to be served. Customers arriving after $\clubsuit$ (but before the server leaves queue $i$ again) will not join; if they did, they would not be served before the server departs, resulting in a non-zero queue length. This contradicts the argument in the first paragraph. %

    \item \textit{When $\bar{L}_i < \theta_i$ and $L_i \geq \lambda_i \bar{L}_i/(\mu_i - \lambda_i)$, i.e., Figure \ref{subfig:exh_case4}.}

We first show that, in this case, all customers join the system at equilibrium.
From the condition $L_i \geq \lambda_i \bar{L}_i/(\mu_i - \lambda_i)$, it follows that 
$\mu_i L_i \geq \lambda_i\bigl(L_i + \bar{L}_i\bigr)$, so the server can clear all customers even if they all join.
Suppose the queue length is zero when the server leaves queue $i$. Since $\bar{L}_i < \theta_i$, the customer arriving at that moment is willing to join. Because the server can clear the queue with all customers joining, then by \Cref{coro:waiting_time_decreasing_same_onduration}, the waiting time decreases from the time the server leaves queue $i$ until it leaves again. Consequently, all customers are indeed willing to join. Thus, $J_i = L_i + \bar{L}_i$ and $\bar{J}_i = 0$, and the first-best throughput $\lambda_i$ is achieved. This leads to an essentially unique equilibrium outcome.

We now show that the equilibrium outcome is unique when $L_i > \lambda_i \bar{L}_i/(\mu_i - \lambda_i)$.
We claim that when $L_i > \lambda_i \bar{L}_i/(\mu_i - \lambda_i)$, the equilibrium post-clearance duration $T_i$ must be strictly positive. To see this, suppose $T_i = 0$. Then the number of customers served in one cycle is $\mu_i L_i$, which exceeds the number of customers joining in one cycle, $\lambda_i(L_i + \bar{L}_i)$ by the condition $L_i > \frac{\lambda_i \bar{L}_i}{\mu_i - \lambda_i}$.
Hence, the post-clearance duration $T_i$ must be strictly positive and can be uniquely determined by the flow-balance equation:
\begin{align*}
  \lambda_i\cdot (L_i+\bar{L}_i) \;=\; \mu_i\cdot (L_i - T_i) \;+\; \lambda_i\cdot T_i.
\end{align*}
Since the post-clearance duration $T_i$ is strictly positive and unique, we have the unique equilibrium outcome, which is exhaustive, as plotted in \Cref{subfig:exh_case4}.\Halmos

\end{enumerate}

\end{proof}

\medskip

\begin{proof}{\textbf{Proof of Theorem~\ref{thm:fixed_duration_nonexhaustive_equilibrium}.}}

Define
\begin{align}
\label{eq:m_i(t)}
 m_i(t) :=  \frac{q_i(t)/\mu_i - \iota_i(t)\cdot r_i(t)}{L_i}.
\end{align}
That is, $m_i(t)$ is the term inside the floor function in \eqref{eq:z_i(t)} minus $\iota_i(t)\in \{0,1\}$.
Since $\floor{x+\iota_i(t)}=\floor{x}+\iota_i(t)$ for any $x$, by the definition of $m_i(t)$, we have $z_i(t)=\floor{m_i(t)}+\iota_I(t)$.
Then, we can rewrite the waiting time in \eqref{eq:waiting_time_fixed_duration} by $m_i(t)$ as follows:
\begin{align}
\label{eq:waiting_time_m_i_t}
 W_i(t) = \frac{q_i(t)}{\mu_i} + (1-\iota_i(t))\cdot r_i(t) +   
 (\floor{m_i(t)} + \iota_i(t))\cdot \bar{L}_i. 
\end{align}

We use the following lemmas, whose proofs can be found immediately after the proof of Theorem~\ref{thm:fixed_duration_nonexhaustive_equilibrium}. %

\begin{lemma}
\label{lem:fixed_duration_non_exhaustive_all_join_isnot_outcome}
For an exogenous \on-\off duration $(L_i, \bar{L}_i)$ with $\bar{L}_i < \theta_i$ and $L_i < \lambda_i \bar{L}_i/(\mu_i - \lambda_i)$, it is not an equilibrium outcome for all customers to join the system, i.e., $\bar{J}_i(L_i,\bar{L}_i)>0$. Moreover, at any equilibrium outcome, the post-clearance duration must be zero.
\end{lemma}

\begin{lemma}
\label{lem:fixed_duration_non_exhaustive_joiningduration_waitingtime_integer}
For an exogenous \on-\off duration $(L_i, \bar{L}_i)$ with $\bar{L}_i < \theta_i$ and $L_i < \lambda_i \bar{L}_i/(\mu_i - \lambda_i)$, the following hold at equilibrium:
\begin{enumerate}
    \item $J_i = (\mu_i/\lambda_i) L_i$ and $\bar{J}_i = \bar{L}_i - ((\mu_i - \lambda_i)/\lambda_i) L_i$.
    \item $W_i(\diamondsuit) = \theta_i$, i.e., the waiting time of customers arriving at the transition epoch from not joining to joining exactly equals their patience.
    \item $k_i := m_i(\diamondsuit)$ must be a non-negative integer; that is, the first customer joining the system during a joining period must be the first to be served in some \on period.
\end{enumerate}
\end{lemma}

Our proof proceeds in three steps. First, for any given case in \{1, 2, 3, 4, 5\} presumed to represent the equilibrium outcome, we establish the existence and uniqueness of the variables $\zeta_i$ and $\underline{q}_i$, explicitly given by the formulas in Theorem~\ref{thm:fixed_duration_nonexhaustive_equilibrium}. Consequently, for each assumed equilibrium case, the corresponding equilibrium outcome is uniquely determined. Second, we demonstrate that exactly one of the five cases emerges as the equilibrium. Finally, we confirm that this unique equilibrium outcome is non-exhaustive.

In the following, we use the term \emph{line segment} to refer to a specific linear portion of the piecewise linear queueing dynamics curve depicted in \Cref{fig:fixed_duration_nonexhaustive_equilibrium}.
The \emph{duration} of a line segment denotes the time interval over which this linear segment extends. 

The equilibrium is determined primarily by the conditions outlined in Lemma~\ref{lem:fixed_duration_non_exhaustive_joiningduration_waitingtime_integer}, along with the requirement that each line segment's duration in the corresponding case of Figure~\ref{fig:fixed_duration_nonexhaustive_equilibrium} remains non-negative.

\medskip

\noindent\underline{\textit{Step 1: Determining the Unique $\zeta_i$ and $\underline{q}_i$ for Each Case}}

\medskip

\noindent\textbf{Case 1:} For Case~1 to constitute an equilibrium, each line segment's duration in \Cref{subfig:nonexh_case1} must be non-negative. Specifically, we require $0 \leq \zeta_i \leq L_i$ and that the duration of the increasing line during the \off period lies within the interval $[0, \bar{L}_i]$. This duration is explicitly given by $J_i - (L_i - \zeta_i) = \zeta_i + ((\mu_i - \lambda_i)/\lambda_i) L_i$. To summarize, the conditions are
\begin{align}
    & 0 \leq \zeta_i \leq L_i, \label{eq:fd_eqcondition1_case1} \\
    & 0 \leq \zeta_i + \frac{\mu_i - \lambda_i}{\lambda_i} L_i \leq \bar{L}_i. \label{eq:fd_eqcondition2_case1}
\end{align}
By \eqref{eq:m_i(t)}, we have
\begin{align*}
    k_i := m_i(\diamondsuit)= \frac{\underline{q}_i - \lambda_i (L_i - \zeta_i)}{\mu_i L_i}.
\end{align*}
This leads to
\begin{align}
\label{eq:fd_underq_k_xi_case1}
    \underline{q}_i = k_i \mu_i L_i + \lambda_i (L_i - \zeta_i).
\end{align}
The waiting time $W_i(\diamondsuit)$, according to equation~\eqref{eq:waiting_time_m_i_t}, is
\begin{align*}
    W_i(\diamondsuit) & = \frac{\underline{q}_i + (\mu_i - \lambda_i)(L_i - \zeta_i)}{\mu_i} + (k_i + 1) \bar{L}_i \\
    & = k_i L_i + (L_i - \zeta_i) + (k_i + 1) \bar{L}_i \\
    & = (k_i + 1)(L_i + \bar{L}_i) - \zeta_i,
\end{align*}
where the second equality follows from equation~\eqref{eq:fd_underq_k_xi_case1}. By Lemma~\ref{lem:fixed_duration_non_exhaustive_joiningduration_waitingtime_integer} (ii), we have $W_i(\diamondsuit) = \theta_i$, so
\begin{align}
\label{eq:fd_xi_k_case1}
    \zeta_i = (k_i + 1)(L_i + \bar{L}_i) - \theta_i.
\end{align}
Substituting equation~\eqref{eq:fd_xi_k_case1} into inequalities~\eqref{eq:fd_eqcondition1_case1} and~\eqref{eq:fd_eqcondition2_case1}, we obtain
\begin{align}
\label{eq:k_i_eqcondition_case1}
    \frac{\theta_i}{L_i + \bar{L}_i} - 1 \leq k_i \leq \frac{\theta_i}{L_i + \bar{L}_i} - \frac{\max\{\bar{L}_i, (\mu_i/\lambda_i) L_i\}}{L_i + \bar{L}_i}.
\end{align}
One can verify that the difference between the right-hand side and the left-hand side is strictly smaller than one.  Since $k_i$ must be a non-negative integer (Lemma~\ref{lem:fixed_duration_non_exhaustive_joiningduration_waitingtime_integer} (iii)), there either does not exist or exists a unique non-negative integer $k_i$ satisfying \eqref{eq:k_i_eqcondition_case1}.
The corresponding $\underline{q}_i$ and $\zeta_i$ are given by \eqref{eq:fd_underq_k_xi_case1} and \eqref{eq:fd_xi_k_case1}, respectively.
\medskip

\noindent\textbf{Case 2:} Similarly, for Case~2 to be an equilibrium outcome (see \Cref{subfig:nonexh_case2}), %
it requires that $0 \leq \zeta_i \leq L_i - \bar{J}_i$ and $\bar{J}_i \leq L_i$. This leads to 
\begin{align}
    & 0 \leq \zeta_i \leq \frac{\mu_i}{\lambda_i} L_i - \bar{L}_i, \label{eq:fd_eqcondition1_case2} \\
    & L_i \geq \frac{\lambda_i}{\mu_i} \bar{L}_i. \label{eq:fd_eqcondition2_case2}
\end{align}

By \eqref{eq:m_i(t)}, we have $k_i:=m_i(\diamondsuit)$ in this case as follows
\begin{align*}
    k_i = \left(\frac{\underline{q}_i+(\mu_i-\lambda_i)x_i}{\mu_i} - x_i  \right)\bigg/L_i,
\end{align*}
where $x_i = L_i - \bar{J}_i - \zeta_i$ representing the duration of the third line segment during the \on period in \Cref{subfig:nonexh_case2}. This leads to
\begin{align}
\label{eq:fd_underq_k_xi_case2}
    \underline{q}_i = k_i \mu_i L_i + \lambda_i \left( \frac{\mu_i}{\lambda_i} L_i - \bar{L}_i - \zeta_i \right).
\end{align}
The waiting time is
\begin{align*}
    W_i(\diamondsuit) & = \frac{\underline{q}_i + (\mu_i - \lambda_i) x_i}{\mu_i} + (k_i + 1) \bar{L}_i \\
    & = k_i (L_i + \bar{L}_i) + \frac{\mu_i}{\lambda_i} L_i - \zeta_i,
\end{align*}
using equation~\eqref{eq:fd_underq_k_xi_case2}. Since $W_i(\diamondsuit) = \theta_i$, we get
\begin{align}
\label{eq:fd_xi_k_case2}
    \zeta_i = k_i (L_i + \bar{L}_i) + \frac{\mu_i}{\lambda_i} L_i - \theta_i.
\end{align}
Substituting equation~\eqref{eq:fd_xi_k_case2} into inequality~\eqref{eq:fd_eqcondition1_case2}, we have
\begin{align}
\label{eq:k_i_eqcondition_case2}
    \frac{\theta_i - (\mu_i/\lambda_i) L_i}{L_i + \bar{L}_i} \leq k_i \leq \frac{\theta_i - \bar{L}_i}{L_i + \bar{L}_i} .
\end{align}
Again, the difference between the right-hand side and the left-hand side is strictly smaller than one. Thus, there either does not exist or exists a unique non-negative integer $k_i$ satisfying this inequality.
The corresponding $\underline{q}_i$ and $\zeta_i$ are given by \eqref{eq:fd_underq_k_xi_case2} and \eqref{eq:fd_xi_k_case2}, respectively.
\medskip

\noindent\textbf{Case 3:} Similarly, for Case~3 to be an equilibrium outcome (see \Cref{subfig:nonexh_case3}), it
requires that $0\leq \zeta_i\leq \bar{L}_i-J_i$ and $J_i\leq \bar{L}_i$. This leads to 
\begin{align}
    & 0 \leq \zeta_i \leq \bar{L}_i - \frac{\mu_i}{\lambda_i} L_i, \label{eq:fd_eqcondition1_case3} \\
    & L_i \leq \frac{\lambda_i}{\mu_i} \bar{L}_i. \label{eq:fd_eqcondition2_case3}
\end{align}

The $k_i:=m_i(\diamondsuit)$ defined in \eqref{eq:m_i(t)} in this case becomes
\begin{align*}
    k_i = \frac{\underline{q}_i}{\mu_iL_i}.
\end{align*}
This leads to
\begin{align}
\label{eq:fd_underq_k_xi_case3}
    \underline{q}_i = k_i \mu_i L_i.
\end{align}
The waiting time is
\begin{align*}
    W_i(\diamondsuit) & = \frac{\underline{q}_i}{\mu_i} + (\bar{L}_i - \zeta_i) + k_i \bar{L}_i \\
    & = k_i (L_i + \bar{L}_i) + \bar{L}_i - \zeta_i,
\end{align*}
using equation~\eqref{eq:fd_underq_k_xi_case3}. Since $W_i(\diamondsuit) = \theta_i$, we obtain
\begin{align}
\label{eq:fd_xi_k_case3}
    \zeta_i = k_i (L_i + \bar{L}_i) + \bar{L}_i - \theta_i.
\end{align}
Substituting equation~\eqref{eq:fd_xi_k_case3} into inequality~\eqref{eq:fd_eqcondition1_case3}, we have
\begin{align}
\label{eq:k_i_eqcondition_case3}
    \frac{\theta_i - \bar{L}_i}{L_i + \bar{L}_i}  \leq k_i \leq \frac{\theta_i - (\mu_i/\lambda_i) L_i}{L_i + \bar{L}_i}.
\end{align}
The difference between the right-hand side and the left-hand side is strictly smaller than one. Thus, there either does not exist or exists a unique non-negative integer $k_i$ satisfying this inequality.
The corresponding $\underline{q}_i$ and $\zeta_i$ are given by \eqref{eq:fd_underq_k_xi_case3} and \eqref{eq:fd_xi_k_case3}, respectively.

\medskip

\noindent\textbf{Case 4:} Similarly, for Case~4 to be an equilibrium outcome (see \Cref{subfig:nonexh_case4}), it requires that $0\leq \zeta_i\leq \bar{L}_i$ and $0\leq \bar{J}_i-\zeta_i\leq L_i$. This leads to
\begin{align*}
    & 0 \leq \zeta_i \leq \bar{L}_i, \\
    & 0 \leq \bar{L}_i - \frac{\mu_i - \lambda_i}{\lambda_i} L_i - \zeta_i \leq L_i,
\end{align*}
which yields
\begin{align}
\label{eq:fd_eqcondition_case4}
    \max\left\{ 0, \bar{L}_i - \frac{\mu_i}{\lambda_i} L_i \right\} \leq \zeta_i \leq \bar{L}_i - \frac{\mu_i - \lambda_i}{\lambda_i} L_i.
\end{align}
Similar to Case 3, we have
\begin{align}
\label{eq:fd_underq_k_xi_case4}
    k_i = \frac{\underline{q}_i}{\mu_iL_i}~\Rightarrow~\underline{q}_i = k_i \mu_i L_i.
\end{align}
The waiting time $W_i(\diamondsuit)$ in \eqref{eq:waiting_time_m_i_t}
is calculated as follows,
\begin{align*}
    W_i(\diamondsuit) = \frac{\underline{q}_i}{\mu_i} + (\bar{L}_i - \zeta_i) + k_i\bar{L}_i= z(L_i+\bar{L}_i)+\bar{L}_i - \zeta_i,
\end{align*}
where the second equality holds by \eqref{eq:fd_underq_k_xi_case4}. By \Cref{lem:fixed_duration_non_exhaustive_joiningduration_waitingtime_integer} (ii) that  $W_i(\diamondsuit) = \theta_i$, we find
\begin{align}
\label{eq:fd_xi_k_case4}
    \zeta_i = k_i \cdot (L_i + \bar{L}_i) +\bar{L}_i - \theta_i.
\end{align}

Substituting equation~\eqref{eq:fd_xi_k_case4} into inequality~\eqref{eq:fd_eqcondition_case4}, we have
\begin{subequations}
\begin{align}
    & \text{If } \frac{\mu_i}{\lambda_i} L_i \geq \bar{L}_i: \quad \frac{\theta_i - \bar{L}_i}{L_i + \bar{L}_i} \leq k_i \leq \frac{\theta_i - \frac{\mu_i - \lambda_i}{\lambda_i} L_i}{L_i + \bar{L}_i} ; \label{eq:k_i_eqcondition1_case4} \\
    & \text{If } \frac{\mu_i}{\lambda_i} L_i < \bar{L}_i: \quad \frac{\theta_i - \frac{\mu_i}{\lambda_i} L_i}{L_i + \bar{L}_i}  \leq k_i \leq \frac{\theta_i - \frac{\mu_i - \lambda_i}{\lambda_i} L_i}{L_i + \bar{L}_i} . \label{eq:k_i_eqcondition2_case4}
\end{align}
\end{subequations}
Again, the right-hand side and the left-hand side in both \eqref{eq:k_i_eqcondition1_case4} and \eqref{eq:k_i_eqcondition2_case4} are strictly smaller than one. Thus, there either does not exist or exists a unique non-negative integer $k_i$ satisfying \eqref{eq:k_i_eqcondition1_case4} or \eqref{eq:k_i_eqcondition2_case4}.
The corresponding $\underline{q}_i$ and $\zeta_i$ are given by \eqref{eq:fd_underq_k_xi_case4} and \eqref{eq:fd_xi_k_case4}, respectively.
\medskip

\noindent\textbf{Case 5:} Finally, for Case~5 to be an equilibrium, it requires that $0\leq \zeta_i\leq \bar{L}_i-\bar{J}_i$. This leads to
\begin{align}
\label{eq:fd_eqcondition_case5}
    0 \leq \zeta_i \leq \frac{\mu_i - \lambda_i}{\lambda_i} L_i.
\end{align}
We have
\begin{align}
\label{eq:fd_underq_k_xi_case5}
    k_i = \frac{\underline{q}_i+\lambda_i \zeta_i}{\mu_iL_i}~\Rightarrow~\underline{q}_i = k_i \mu_i L_i - \lambda_i\zeta_i.
\end{align}
The waiting time is
\begin{align*}
    W_i(\diamondsuit) & = \frac{\underline{q}_i + \lambda_i \zeta_i}{\mu_i} + (\bar{L}_i - \bar{J}_i - \zeta_i) + k_i \bar{L}_i \\
    & = k_i (L_i + \bar{L}_i) + \frac{\mu_i - \lambda_i}{\lambda_i} L_i - \zeta_i,
\end{align*}
using equation~\eqref{eq:fd_underq_k_xi_case5}. Setting $W_i(\diamondsuit) = \theta_i$, we find
\begin{align}
\label{eq:fd_xi_k_case5}
    \zeta_i = k_i (L_i + \bar{L}_i) + \frac{\mu_i - \lambda_i}{\lambda_i} L_i - \theta_i.
\end{align}
Substituting equation~\eqref{eq:fd_xi_k_case5} into inequality~\eqref{eq:fd_eqcondition_case5}, we have
\begin{align}
\label{eq:k_i_eqcondition_case5}
    \frac{\theta_i - \frac{\mu_i - \lambda_i}{\lambda_i} L_i}{L_i + \bar{L}_i}  \leq k_i \leq \frac{\theta_i}{L_i + \bar{L}_i}.
\end{align}
Again, the difference between the right-hand side and the left-hand side is strictly smaller than one. Thus, there either does not exist or exists a unique non-negative integer $k_i$ satisfying this inequality.
The corresponding $\underline{q}_i$ and $\zeta_i$ are given by \eqref{eq:fd_underq_k_xi_case5} and \eqref{eq:fd_xi_k_case5}, respectively.
\medskip

\noindent\underline{\textit{Step 2: The Uniqueness of the Equilibrium Case}}

\medskip

From \eqref{eq:k_i_eqcondition_case1}, \eqref{eq:k_i_eqcondition_case2}, \eqref{eq:k_i_eqcondition_case3}, \eqref{eq:k_i_eqcondition1_case4}--\eqref{eq:k_i_eqcondition2_case4} and \eqref{eq:k_i_eqcondition_case5}, we have the equilibrium conditions regarding $k_i$ outlined in Theorem~\ref{thm:fixed_duration_nonexhaustive_equilibrium} parts~(i) and~(ii).
Notice that the interval of constraints about $k_i$ has a left boundary of $\theta_i/(L_i+\bar{L}_i)-1$ and a right boundary of $\theta_i/(L_1+\bar{L}_i)$, with the length of the interval being \textit{exactly} one.
This confirms the existence of the integer $k_i$.
Besides, in each case, given $k_i$, by the analysis in Step 1, the variables $\zeta_i$ and $\underline{q}_i$ are uniquely determined. 
To conclude the proof, we are left to show the uniqueness of the equilibrium outcome. 
Notice that if $k_i$ falls on the boundary between two consecutive cases, these two cases correspond to the same equilibrium outcome. Thus, the boundary value between two consecutive cases can belong to either case.

When $\frac{\theta_i}{L_i + \bar{L}_i}$ is not an integer, there exists a unique integer $k_i$ within the interval $[\frac{\theta_i}{L_i + \bar{L}_i} - 1,\frac{\theta_i}{L_i + \bar{L}_i}]$ since the difference between $\frac{\theta_i}{L_i + \bar{L}_i} - 1$ and $\frac{\theta_i}{L_i + \bar{L}_i}$ is exactly $1$.
This ensures the uniqueness of $k_i$ and thus the uniqueness of the equilibrium case. 

When $\frac{\theta_i}{L_i + \bar{L}_i}$ is an integer, $k_i$ can be $\frac{\theta_i}{L_i + \bar{L}_i}$ or $\frac{\theta_i}{L_i + \bar{L}_i}-1$, corresponding to Case~1 (\Cref{subfig:nonexh_case1}) and Case~5 (\Cref{subfig:nonexh_case5}), respectively.
However, in this situation, the queueing dynamics under these two cases coincide.
To see this, when $\frac{\theta_i}{L_i + \bar{L}_i}$ is an integer, $\zeta_i$ in \Cref{subfig:nonexh_case1} becomes zero, and $\zeta_i$ in \Cref{subfig:nonexh_case5} becomes $\bar{L}_i - \bar{J}_i$, resulting in the duration of the third line segment during the \off period in \Cref{subfig:nonexh_case5} being zero. Therefore, Figures \ref{subfig:nonexh_case1} and \ref{subfig:nonexh_case5} coincide.
Thus, when $\frac{\theta_i}{L_i + \bar{L}_i}$ is an integer, $k_i\in \{\frac{\theta_i}{L_i + \bar{L}_i},\frac{\theta_i}{L_i + \bar{L}_i}-1\}$ can be chosen arbitrarily without affecting the equilibrium outcome.

\medskip

\noindent\underline{\textit{Step 3: The Equilibrium Outcome is Non-Exhaustive}}

\medskip

We now show that the equilibrium outcome is non-exhaustive, i.e., $\underline{q}_i > 0$. If $k_i$ is a strictly positive integer, from our earlier expressions for $\underline{q}_i$, it follows that $\underline{q}_i > 0$.

When $k_i = 0$, the equilibrium outcome cannot be Case~3 or Case~4, since in these cases $\zeta_i = \bar{L}_i - \theta_i < 0$ (because $\theta_i > \bar{L}_i$), which contradicts the requirement that $\zeta_i \geq 0$. Similarly, Case~5 is not possible when $k_i = 0$ because $\zeta_i = \frac{\mu_i - \lambda_i}{\lambda_i} L_i - \theta_i < 0$ (since $L_i < \frac{\lambda_i \bar{L}_i}{\mu_i - \lambda_i}$ implies $\frac{\mu_i - \lambda_i}{\lambda_i} L_i < \bar{L}_i$, and thus smaller than $\theta_i$). Therefore, the equilibrium must be either Case~1 or Case~2. When $k_i = 0$, these cases yield $\underline{q}_i > 0$ by the expression of $\underline{q}_i$ in \Cref{thm:fixed_duration_nonexhaustive_equilibrium}.
This completes the proof. \Halmos
\end{proof}

\medskip

\begin{proof}{\textbf{Proof of \Cref{lem:fixed_duration_non_exhaustive_all_join_isnot_outcome}.}}
We prove the two statements separately.
\begin{enumerate}
    \item \textit{All customers joining cannot be an equilibrium outcome.} Given $\bar{L}_i < \theta_i$ and $L_i < \frac{\lambda_i \bar{L}_i}{\mu_i - \lambda_i}$,  it follows that $\mu_i L_i < \lambda_i (L_i + \bar{L}_i)$. If all arriving customers were to join, the queue length would grow to infinity. Therefore, an outcome where all customers join cannot be an equilibrium.
    \item \textit{The post-clearance duration must be zero at equilibrium.}
Suppose, for contradiction, that at equilibrium, the post-clearance duration $T_i>0$. Then the server finishes serving all customers before leaving the system. Since $\bar{L}_i < \theta_i$, a customer arriving at the time when the server immediately leaves would still be willing to join, because their waiting time would be less than $\theta_i$. This would lead to subsequent customers also joining since the waiting time is decreasing by \Cref{coro:waiting_time_decreasing_same_onduration},
resulting in all customers joining the system. However, as established in the first part, this cannot be an equilibrium outcome. Therefore, the post-clearance duration $T_i$ must be zero at equilibrium. \Halmos
\end{enumerate}
\end{proof}

\medskip

\begin{proof}{\textbf{Proof of \Cref{lem:fixed_duration_non_exhaustive_joiningduration_waitingtime_integer}.}}

We now prove each result in \Cref{lem:fixed_duration_non_exhaustive_joiningduration_waitingtime_integer} separately.

\begin{enumerate}
    \item By Lemma~\ref{lem:fixed_duration_non_exhaustive_all_join_isnot_outcome}, we know that the post-clearance duration must be zero. Then, the flow-balance equation becomes $\lambda_i J_i = \mu_i L_i$, which leads to $J_i = (\mu_i/\lambda_i) L_i$. Therefore, $\bar{J}_i = L_i + \bar{L}_i - J_i = \bar{L}_i - \frac{\mu_i - \lambda_i}{\lambda_i} L_i$.

    \item This follows from \Cref{lem:fixed_duration_waitingtime} (iii).
    
    \item 
Recall that the term inside the floor function of $z_i(\diamondsuit)$, as defined in \eqref{eq:z_i(t)}, represents the additional (possibly fractional) number of full \off periods a customer arriving at time $\diamondsuit$ must wait through, excluding the current residual \off period if $\iota_i(\diamondsuit) = 0$. Since such a customer is served exactly at the start of an \on period by the proof of \Cref{lem:waiting_time} (i), the total number of complete \off periods they wait through---apart from any ongoing residual \off period---must be an integer.
    This leads to the conclusion that the term inside the floor function of $z_i(\diamondsuit)$ must be an integer and thus $k_i = m_i(\diamondsuit)$ is an integer, where $m_i(\diamondsuit)$ is defined in \eqref{eq:m_i(t)}.\Halmos
\end{enumerate}
\end{proof}

\medskip

\begin{proof}{\textbf{Proof of \Cref{coro:fixed_duration_post_clearance_duration}.}}
By \Cref{thm:equilibrium_fixedduration_exhaustive} and \Cref{thm:fixed_duration_nonexhaustive_equilibrium}, the post-clearance duration $T_i(L_i,\bar{L}_i)$ is non-zero only in the cases Figures \ref{subfig:exh_case1} and \ref{subfig:exh_case4}. In the case of \Cref{subfig:exh_case1}, we have 
\begin{align*}
  \textrm{if $\bar{L}_i\geq \theta_i$ and $L_i\geq \frac{\lambda_i\theta_i}{\mu_i-\lambda_i}$},\quad
  T_i(L_i,\bar{L}_i) = L_i - \frac{\lambda_i\theta_i}{\mu_i-\lambda_i}.
\end{align*}
In the case of \Cref{subfig:exh_case4}, we have
\begin{align*}
 \textrm{if $\bar{L}_i< \theta_i$ and $L_i\geq \frac{\lambda_i\bar{L}_i}{\mu_i-\lambda_i}$}, \quad
T_i(L_i,\bar{L}_i)=L_i - \frac{\lambda_i \bar{L}_i}{\mu_i-\lambda_i}. 
\end{align*}
Thus, the equilibrium post-clearance duration in the cases Figures \ref{subfig:exh_case1} and \ref{subfig:exh_case4} can be written as:
\begin{align*}
 T_i(L_i,\bar{L}_i)=L_i - \frac{\lambda_i}{\mu_i-\lambda_i}\min\{ \theta_i, \bar{L}_i\}
 =
\max\left\{L_i -\frac{\lambda_i}{\mu_i-\lambda_i}\theta_i, ~  L_i - \frac{\lambda_i}{\mu_i-\lambda_i}\bar{L}_i\right\}   .    
\end{align*}
Combining $ T_i(L_i,\bar{L}_i)\geq 0$, we have \eqref{eq:equilibrium_fd_hat_T_i}.
\Halmos   
\end{proof}

\medskip

\begin{proof}{\textbf{Proof of \Cref{prop:LP_fixed_duration}.}}
This follows from the discussions before \Cref{prop:LP_fixed_duration}. \Halmos
\end{proof}

\medskip

\begin{proof}{\textbf{Proof of \Cref{thm:polling_identical_service_rates}.}}
Recall that when $n\geq 2$, the \off duration of queue $i$ equals the \on durations of the other queues plus the switchover times, i.e., $\bar{L}_i = 1^\top \tau + 1^\top L_{-i}$. Thus, given the \on durations $L$, the \off durations are uniquely determined. For this reason, in what follows, we will work with \on durations for convenience.

If, under the optimal \on-\off durations, the equilibrium post-clearance duration $T_i>0$ for some $i\in \N$, then by Theorems \ref{thm:equilibrium_fixedduration_exhaustive} and \ref{thm:fixed_duration_nonexhaustive_equilibrium} the equilibrium outcome of queue $i$ must be either \Cref{subfig:exh_case1} or \Cref{subfig:exh_case4}. Thus, we do not need to worry about the situation in which $T_i>0$.

When $T_i=0$, the equilibrium outcome can be non-exhaustive (see \Cref{fig:fixed_duration_nonexhaustive_equilibrium}) or can take other forms of exhaustive outcomes (see \Cref{subfig:exh_case2} and \Cref{subfig:exh_case3}). To conclude the proof, we now show that when restricting to zero equilibrium post-clearance durations, the optimal \on-\off durations must induce an equilibrium outcome corresponding either to \Cref{subfig:exh_case1} or to \Cref{subfig:exh_case4}. Formally, for any (possibly empty) set $\I\subseteq\N$, given any $L_\I$ with $T_{\I}>0$, the optimal $L_{\barI}$, when restricted to $T_{\barI}=0$, must induce an equilibrium outcome corresponding to either \Cref{subfig:exh_case1} or \Cref{subfig:exh_case4}, where $\barI:=\N\setminus \I$ and $L_{\I}$ denotes the subvector of $L\in \mathbb{R}^{|\N|}$ indexed by $\I$ (and similarly for $T_\I$, $L_{\barI}$, and $T_{\barI}$).

When the service rates are identical (denoted by $\mu_0$) and $T_{\barI}=0$, the throughput formula in problem~\eqref{prob:lf_fixed_duration} simplifies to
\begin{align}
\frac{\mu^\top L - \sum_{i\in\N}(\mu_0 - \lambda_i) T_i}{1^\top \tau + 1^\top L}  
&= \frac{\mu_0 (1^\top \tau + 1^\top L) - \mu_0 1^\top \tau - \sum_{i\in\N} (\mu_0 - \lambda_i) T_i}{1^\top \tau + 1^\top L} \nonumber \\
&= \mu_0 - \frac{\mu_0 1^\top \tau + \sum_{i\in\N} (\mu_0 - \lambda_i) T_i}{1^\top \tau + 1^\top L} \nonumber\\
&= \mu_0 - \frac{\mu_0 1^\top \tau + \sum_{i\in\I} (\mu_0 - \lambda_i) T_i}{1^\top \tau + 1^\top L}. \nonumber
\end{align}

Given any feasible $(L_{\I},T_{\I})$ (note that the feasible region of $(L_{\I},T_{\I})$ depends on the choice of $L_{\barI}$; we drop the explicit dependence here for simplicity), the throughput maximization problem is
\begin{align}
    \max_{L \in \mathbb{R}_+^{|\barI|}}\quad & \mu_0 - \frac{\mu_0 1^\top \tau + \sum_{i\in\I} (\mu_0 - \lambda_i) T_i}{1^\top \tau + 1^\top L} \nonumber \\
    \text{s.t.} \quad & L_i \leq \frac{\lambda_i}{\mu_i - \lambda_i} \theta_i, \quad \forall i \in \barI, \label{constraint1:polling_fd_hat_T=0} \\
    & L_i - \frac{\lambda_i}{\mu_i - \lambda_i} \cdot 1^\top L_{-i} \leq \frac{\lambda_i}{\mu_i - \lambda_i} \cdot 1^\top \tau, \quad \forall i \in \barI. \label{constraint2:polling_fd_hat_T=0}
\end{align}
Constraints \eqref{constraint1:polling_fd_hat_T=0}--\eqref{constraint2:polling_fd_hat_T=0} guarantee $T_{\barI}=0$ at equilibrium (since we restrict to the space of \on-\off durations that induce $T_{\barI}=0$).
To see this, to guarantee $T_{\barI}=0$ at equilibrium, by \Cref{coro:fixed_duration_post_clearance_duration}, we require, for any $i\in \barI$, $L_i-\frac{\lambda_i}{\mu_i-\lambda_i}\theta_i\leq 0$ (i.e., the constraint \eqref{constraint1:polling_fd_hat_T=0}) and $L_i-\frac{\lambda_i}{\mu_i-\lambda_i}\bar{L}_i\leq 0$ (i.e., constraint \eqref{constraint2:polling_fd_hat_T=0}), where we used the equation $\bar{L}_i=1^\top\tau + 1^\top L_{-i}$.

Let $z_i := \frac{\lambda_i}{\mu_i - \lambda_i} \theta_i - L_i$, i.e., $L_i = \frac{\lambda_i}{\mu_i - \lambda_i} \theta_i - z_i$ for all $i \in \barI$. Then we can rewrite the objective function as:
\begin{align}
    \mu_0 - \frac{\mu_0 1^\top \tau + \sum_{i\in\I} (\mu_0 - \lambda_i) T_i}{1^\top \tau + 1^\top L_{\I} - 1^\top z + \sum_{i\in\barI} \frac{\lambda_i}{\mu_i - \lambda_i} \theta_i }, \nonumber
\end{align}
which is decreasing in $z_i$ for all $i \in \barI$.

Constraint \eqref{constraint1:polling_fd_hat_T=0} is equivalent to $z_i \geq 0$, and constraint \eqref{constraint2:polling_fd_hat_T=0} becomes:
\begin{align}
    z_i - \frac{\lambda_i}{\mu_i - \lambda_i} \cdot 1^\top z_{\barI_{-i}} + \frac{\lambda_i}{\mu_i - \lambda_i} \left(1^\top \tau + 1^\top L_\I - \theta_i + \sum_{j \in \barI_{-i}} \frac{\lambda_j}{\mu_j - \lambda_j} \theta_j\right) \geq 0, \quad \forall i \in \barI. \nonumber
\end{align}
Since $L_i \geq 0$, by the definition of $z_i$, we have $z_i \leq \frac{\lambda_i}{\mu_i - \lambda_i} \theta_i$ for all $i \in \barI$. Collecting these results, the feasible region of $z$ is given by 
\[
\mathcal{Z} := \{ z \in \mathbb{R}^{|\barI|} \mid 0 \leq z \leq c, \ m + M z \geq 0 \},
\]
where
\begin{align}
    c_i &= \frac{\lambda_i}{\mu_i - \lambda_i} \theta_i, \quad \forall i \in \barI, \nonumber \\
    m_i &= \frac{\lambda_i}{\mu_i - \lambda_i} \left(1^\top \tau + 1^\top L_\I - \theta_i + \sum_{j \in \barI_{-i}} \frac{\lambda_j}{\mu_j - \lambda_j} \theta_j\right), \quad \forall i \in \barI, \nonumber
\end{align}
and the matrix $M\in \mathbb{R^{|\barI|\times |\barI|}}$ is defined as:
\begin{align*}
    M_{ij} =
    \begin{cases}
        1, & \text{if } i = j, \\
        -\dfrac{\lambda_i}{\mu_i - \lambda_i}, & \text{if } i \ne j.
    \end{cases}
\end{align*}

Note that $M$ is a $Z$-matrix, meaning its off-diagonal elements are non-positive. Also, $\mathcal{Z}$ is non-empty, as it includes $z = c$. Then, by Theorem 3.11.6 in \cite{cottle_2009_LCPBook}, the set $\mathcal{Z}$ has a unique least element. Since the throughput decreases in $z$, the (unique) least element must be the optimal solution.
It is straightforward to show that for this least element, either $z_i = 0$ or $(m + M z)_i = 0$ for all $i\in \barI$, meaning that for each $i$, either constraint \eqref{constraint1:polling_fd_hat_T=0} or constraint \eqref{constraint2:polling_fd_hat_T=0} is binding. The binding constraint \eqref{constraint1:polling_fd_hat_T=0} corresponds to \Cref{subfig:exh_case1} with $T_i=0$, and the binding constraint \eqref{constraint2:polling_fd_hat_T=0} corresponds to \Cref{subfig:exh_case4} with $T_i=0$. This concludes the proof.
\Halmos
\end{proof}

\subsection{Proofs in \Cref{sec:exhaustive}}
\begin{proof}{\textbf{Proof of \Cref{lem:exhaustive_equilibrium_structure}.}}
We focus on a focal queue $i \in \N$. Consider the joining strategies $f_{-i}$ of the other queues (which need not be equilibrium strategies). Given the strategies $f_{-i}$ and the post-clearance durations $T_{-i}$ of the other queues, the \on durations of these queues are fixed, which in turn fixes the \off duration of queue $i$. Consequently, the equilibrium analysis of queue $i$ reduces to the case with an exogenous \off duration. Besides, queue $i$ exhibits an exhaustive equilibrium outcome by the definition of the exhaustive service policy.
Then, following an analysis analogous to that in \Cref{thm:equilibrium_fixedduration_exhaustive} for the exhaustive equilibrium outcome under exogenous \on--\off durations, it can be shown that the customer arriving at the epoch when the server departs from queue $i$ experiences the longest waiting time, with waiting times decreasing thereafter. This observation implies (i) and (ii) of \Cref{lem:exhaustive_equilibrium_structure}.
\Halmos
\end{proof}

\medskip

\begin{proof}{\textbf{Proof of \Cref{lem:equilibriumset}.}}
We write $\I(T)$ as $\I$ for simplicity. By the definition of the equilibrium (all-joining) set $\I$, all arrivals to the queues in $\I$ join the system. If $1^\top \rho_{\I} > 1$, it is clear that the server cannot finish serving all these arriving customers. If $1^\top \rho_{\I} = 1$, then because there are switchover times $1^\top \tau>0$, during which the server is idle, the server again cannot complete service for all these arriving customers in queues $i \in \I$. Consequently, the length of at least one queue in $\I$ grows to infinity. This implies that some customers in these queues experience infinite waiting times, contradicting the definition of the all-joining set $\I$.
\Halmos
\end{proof}

\medskip

\begin{proof}{\textbf{Proof of \Cref{prop:n_step}.}}

We first show that problem \ref{prob:alpha} admits feasible solutions, and that we can discard constraint \eqref{con:non_neg} without losing feasibility. 
Without constraint \eqref{con:non_neg}, problem \ref{prob:alpha} is equivalent to a standard LCP problem $\textsf{LCP}(q(T), A)$ by setting $z = 1-\alpha$, where $q(T) = b(T) - A1$. 
By Theorem 3.11.6 in \cite{cottle_2009_LCPBook}, if the polyhedron 
\begin{align*}
\mathsf{POL}(q(T),A) := \{z \in \mathbb{R}^n : Az + q(T) \geq 0,  z \geq 0\}
\end{align*}
with $Z$-matrix $A$ is non-empty, then it admits a least element, which is a feasible solution of $\textsf{LCP}(q(T), A)$. 
One can easily verify that $z = 1$ satisfies $A1 + q(T) \geq 0$, so $\mathsf{POL}(q(T),A)$ is indeed non-empty. 
As a result, there is a least element $z^\ast$ in $\mathsf{POL}(q(T),A)$, and this $z^\ast$ is feasible for $\textsf{LCP}(q(T), A)$. 
By the definition of the least element, we have $z^\ast \leq 1$ since $z = 1$ belongs to $\mathsf{POL}(q(T), A)$. 
Hence, the corresponding $\alpha^\ast := 1 - z^\ast \geq 0$ satisfies constraint \eqref{con:non_neg}. 

Based on the above argument, we conclude that problem \ref{prob:alpha} admits a feasible solution $\alpha^\ast$. 
Moreover, if we consider this solution $\alpha^\ast$, we can discard constraint \eqref{con:non_neg} in problem \ref{prob:alpha} without loss of feasibility. 
On the other hand, as shown in \Cref{lem:lp_equilirbium}, problem \ref{prob:alpha} has a unique solution. 
Thus, problem \ref{prob:alpha} is equivalent to the one obtained by removing \eqref{con:non_neg}, i.e., $\textsf{LCP}(q(T), A)$.

Our \Cref{alg:pivoting} is adapted from \cite{chandrasekaran_1970_LCP_Z} with two major changes: 
first, our algorithm does not check feasibility, since our problem always admits a feasible solution (as established above). 
Second, instead of solving a system of linear equations, our algorithm (see line 8) directly uses the closed-form inverse matrix $A_{\I}^{-1}$ (see \Cref{lem:inverse_mat}), implicitly assuming $A_{\I}$ is invertible at each iteration. Since Chandrasekaran's algorithm solves $\textsf{LCP}(q(T), A)$ and problem \ref{prob:alpha} is equivalent to $\textsf{LCP}(q(T), A)$ (as shown before), it remains to verify that $A_{\I}$ is indeed invertible for non-empty $\I$ at each iteration.

 By \Cref{lem:equilibriumset}, any equilibrium set $\I(T)$ satisfies $1^\top \rho_{\I(T)} < 1$. Consequently, there can be at most $\bar{n}$ steps for termination. Furthermore, because the set $\I$ strictly grows at each iteration, we have $1^\top \rho_{\I} < 1^\top \rho_{\I(T)} < 1$ for $\I$ at any intermediate iteration. By \Cref{lem:inverse_mat}, this implies that $A_{\I}$ is invertible at each iteration. This completes the proof. \Halmos
\end{proof}

\medskip

\begin{proof}{\textbf{Proof of \Cref{lem:lp_equilirbium}.}}

Recall that the problem \ref{prob:alpha} admits feasible solutions as shown in the proof of \Cref{prop:n_step}. 
We now prove the uniqueness. Specifically, we will show the following.
\begin{enumerate}
\item There exists an $\omega \in \mathbb{R}^n_{++}$ such that any feasible solution of problem~\ref{prob:alpha} is an optimal solution of problem~\ref{prob:lp_greatest_element} defined below.
\item For this choice of $\omega$, problem~\ref{prob:lp_greatest_element} admits a unique optimal solution.
\end{enumerate}
\begin{align}
     \max_{\alpha \in \mathcal{A}(b(T),A)}\quad& \omega^\top \alpha \tag{$\mathsf{LP}$($T$)}\label{prob:lp_greatest_element}.
\end{align}
To prove (i) and (ii), it suffices to establish that \emph{any} feasible solution of problem~\ref{prob:alpha} is an optimal solution of problem~\ref{prob:lp_greatest_element} under \emph{any} $\omega \in \mathbb{R}^n_{++}$. Indeed, if this holds, then by picking a particular feasible solution of problem~\ref{prob:alpha}, we conclude that problem~\ref{prob:lp_greatest_element} has the \emph{same} optimal solution %
for \emph{all} $\omega \in \mathbb{R}^n_{++}$. In fact, this implies that the polyhedron $\mathcal{A}(b(T), A)$ admits a greatest element, which must be unique.

In what follows, given $\I,\J\subseteq\N$, we denote by $1_{\I\J}$ (or $0_{\I\J}$) a matrix of dimension $|\I|\times |\J|$ with all entries equal to $1$ (or $0$). Given $\I,\J\subseteq \N$, let $A_{\I\J}\in\mathbb{R}^{|\I|\times |J|}$ denote the submatrix of $A\in\mathbb{R}^{n\times n}$ whose entries are those of $A$ indexed by the rows in $\I$ and the columns in $\J$. When $\I$ and $\J$ are identical, we write $A_{\I\I}$ as $A_{\I}$ for simplicity, as used in the main body. Besides, we use $\mathbb{I}$ to denote the $n\times n$ identity matrix.
Sometimes, we will use the subscript $n$ to explicitly specify its dimension.

To prove that any equilibrium $\alpha(T)$, i.e., any feasible solution of problem \ref{prob:alpha}, is an optimal solution of problem \ref{prob:lp_greatest_element} with an arbitrarily given $\omega\in \mathbb{R}^n_{++}$, by the strong duality of LP, it is sufficient to show that the objective of problem \ref{prob:lp_greatest_element}, termed primal objective, under any equilibrium $\alpha(T)$ matches some feasible dual objective of \ref{prob:lp_greatest_element}. For simplicity, we drop $T$ in $\alpha(T)$ and $b(T)$.

For any equilibrium set $\I$ of the exhaustive service policy $\e$, \Cref{lem:equilibriumset} implies $1^\top \rho_{\I} < 1$, which in turn makes $A_{\I}$ invertible by \Cref{lem:inverse_mat}. Hence, the equations \eqref{eq:alpha_m}--\eqref{eq:alpha_m_n} describing the corresponding $\alpha$ is valid. In fact, the corresponding $\alpha$ in \eqref{eq:alpha_m}--\eqref{eq:alpha_m_n} can be written as
\begin{align}
\nonumber
    \begin{pmatrix}
        \alpha_\I\\
        \alpha_{\Bar{\I}}
    \end{pmatrix}
    = 
    \begin{bmatrix}
        A_\I & A_{\I\bar{\I}}\\
  0_{\bar{\I}\I} & \mathbb{I}_{\bar{\I}}
    \end{bmatrix}^{-1}\begin{pmatrix}
        b_{\I}\\
        1_{\bar{\I}}
    \end{pmatrix}.
\end{align}
Then, the primal objective under the above $\alpha$ induced by the equilibrium set $\I$ is
\begin{align}
    \omega^\top \begin{pmatrix}
        \alpha_\I\\
        \alpha_{\Bar{\I}}
    \end{pmatrix} = \omega^\top  \begin{bmatrix}
        A_\I & A_{\I\bar{\I}}\\
  0_{\bar{\I}\I} & \mathbb{I}_{\bar{\I}}
    \end{bmatrix}^{-1}\begin{pmatrix}
        b_{\I}\\
        1_{\bar{\I}}
    \end{pmatrix}.
    \label{eq:pimary_obj}
\end{align}

Consider the dual problem of \ref{prob:lp_greatest_element} as follows:
\begin{align}
\min_{x,y\in\mathbb{R}^n}\quad &\begin{pmatrix}
b  \\
    1_n
\end{pmatrix}^\top \begin{pmatrix}
   x \\
   y 
\end{pmatrix} \tag{\textrm{DUAL}($T$)} \label{prob:dual_lp}\\
\textrm{s.t.}\quad & \begin{bmatrix}
   A \\
    \mathbb{I}_n
\end{bmatrix}^\top \begin{pmatrix}
   x \\
   y 
\end{pmatrix} \geq \omega, \nonumber\\
& x,y \geq 0\nonumber.
\end{align}
In what follows, we first construct a dual variable (which may not be dual-feasible) and show the corresponding dual objective equals the primal objective \eqref{eq:pimary_obj}. Then, we prove that the constructed dual variable is indeed dual-feasible.
\smallskip

We construct a dual variable in the following way: let $x_{\barI}=0$ and $y_{\I}=0$, and make the first constraint in \ref{prob:dual_lp} binding to solve the rest of the variables $x_{\barI}$ and $y_{\I}$, i.e.,
\begin{align}
   \begin{bmatrix}
        A\\
        \mathbb{I}
    \end{bmatrix}^\top   \begin{pmatrix}
       x \\
  y
    \end{pmatrix}=\omega.
\nonumber
\end{align}
In other words, $A^\top x+\mathbb{I}^\top y=\omega$, where
\begin{align}
    \nonumber
    A^\top x =
    \begin{bmatrix}
        A_{\I} & A_{\I\barI}\\[0ex]
        A_{\barI \I} & A_{\barI}
    \end{bmatrix}^\top
    \begin{pmatrix}
        x_{\I}\\
       x_{\barI}
    \end{pmatrix} =  \begin{bmatrix}
    A_{\I}^\top &  A_{\barI \I}^\top 
    \\[1ex] %
    A_{\I\barI}^\top    & A_{\barI}^\top
\end{bmatrix}
    \begin{pmatrix}
        x_{\I}\\
       x_{\barI}
    \end{pmatrix}
    =
        \begin{bmatrix}
       A_{\I}^\top~ x_{\I}\\[1ex]
     A_{\I\barI}^\top~  x_{\I}
    \end{bmatrix}.
\end{align}
The last equality holds by $x_{\barI}=0$. Plugging $x_{\barI}=0$, $y_{\I}=0$ and the above equation into $A^\top x+\mathbb{I}^\top y=\omega$, we have
\begin{align*}
     \begin{bmatrix}
        A_{\I} & A_{\I\barI}\\[0ex]
       0_{\barI \I} & \mathbb{I}_{\barI}
    \end{bmatrix}^\top
    \begin{pmatrix}
        x_{\I}\\
        y_{\barI}
    \end{pmatrix}
 = \omega.
\end{align*}
This yields
\begin{align}
    \label{eq:dual_solu}
   \begin{pmatrix}
        x_{\I}\\
        y_{\barI}
    \end{pmatrix} =       \left(\begin{bmatrix}
        A_{\I} & A_{\I\barI}\\[0ex]
       0_{\barI \I} & \mathbb{I}_{\barI}
    \end{bmatrix}^{-1}\right)^\top \omega.
\end{align}
Plugging $x_{\barI}=0$, $y_{\I}=0$ and \eqref{eq:dual_solu} into \ref{prob:dual_lp}, we get 
the corresponding dual objective, under the dual variable constructed above, as follows:
\begin{align}
 \begin{pmatrix}
 b  \\
    1_n
\end{pmatrix}^\top \begin{pmatrix}
   x \\
   y 
\end{pmatrix}  =  \begin{pmatrix}
 b_{\I}   \\
    1_{\barI}
\end{pmatrix}^\top \left(\begin{bmatrix}
        A_{\I} & A_{\I\barI}\\[0ex]
       0_{\barI \I} & \mathbb{I}_{\barI}
    \end{bmatrix}^{-1}\right)^\top \omega.
 \label{eq:dual_obj}   
\end{align}
Observe that, \eqref{eq:dual_obj} is the same as \eqref{eq:pimary_obj}.

\smallskip

We now turn to prove that the constructed dual variable is indeed dual-feasible. By the construction that we make the first constraint in \ref{prob:dual_lp} binding, we are left to show $x,y\geq 0$. 
Specifically, since we construct $x_{\barI}=0$ and $y_{\I}=0$, we are left to show \eqref{eq:dual_solu} $\geq 0$. Since $\omega>0$, it is sufficient to show that all elements of the inverse matrix in \eqref{eq:dual_solu} are non-negative. By Schur Complement, the inverse matrix can be simplified as 
\begin{align}
 \label{eq:proof_inversematrix_dual}
    \begin{bmatrix}
        A_\I & A_{\I\bar{\I}}\\
  0_{\bar{\I}\I} & \mathbb{I}_{\bar{\I}}
    \end{bmatrix}^{-1}=
  \begin{bmatrix}
        A_\I^{-1} & -A_\I^{-1}  A_{\I\barI}\\
  0_{\bar{\I}\I} & \mathbb{I}_{\bar{\I}}
    \end{bmatrix} .
\end{align}
Since $1^\top \rho_{\I}<1$ for any equilibrium set $\I$ by \Cref{lem:equilibriumset}, the inverse matrix $A_{\I}^{-1}$ given by \eqref{eq:invermatrix_1}--\eqref{eq:invermatrix_2} must be non-negative.
Besides, considering all elements in $A_{\I\bar{\I}}$ are negative by the definition \eqref{eq:matrix_A_b} of the matrix $A$, we conclude that all elements in \eqref{eq:proof_inversematrix_dual} are non-negative. This completes the proof. \Halmos
\end{proof}

\medskip

\begin{proof}{\textbf{Proof of \Cref{lem:piecewise_linear_increasing}.}}
The piecewise linear property comes from \eqref{eq:alpha_m}--\eqref{eq:alpha_m_n}. Since $b(T)$ is increasing with $T$, it is easy to show that the greatest element of the polyhedron $\mathcal{A}(b(T),A)$
is non-decreasing with $T$. Since this greatest element is the unique equilibrium $\alpha(T)$ by \Cref{lem:lp_equilirbium}, we conclude that $\alpha(T)$ is non-decreasing with $T$. \Halmos
\end{proof}

\medskip

\begin{proof}{\textbf{Proof of \Cref{thm:structure_opt_ex}.}}

Since customers will not renege after they join the system, the throughput can also be written as the long-run average number of customers joined:
\begin{align}
\label{eq:tp_in}
 \tp(T)=   \sum_{i\in\N}\lambda_i  - \frac{\sum_{i\in\N}\lambda_i \bar{J}_i(T)}{1^\top \tau+\sum_{i\in\N}\left(c_i\alpha_i(T)+T_i\right)}.
\end{align}
Using the geometric relationship illustrated in \Cref{fig:exhaustiveservice_equilibrium}, we have $\bar{J}_i(T)=\bar{L}_i(T) - \alpha_i(T)\theta_i$ with $\bar{L}_i(T)$ given by \eqref{eq:exhaustiveservice_off_duration_alpha_j_theta_j}.
We will use the above formula later.

We prove each case separately. To simplify notation, we write $\I(0)$ as $\I$. 

\medskip
\noindent\underline{\textit{Proof of \Cref{thm:structure_opt_ex} (i)}}
\smallskip

The first part is straightforward: if $\I = \N$, then customers from all queues choose to join the system. Consequently, under these conditions, the pure exhaustive service policy $\pe$ achieves the first-best throughput $1^\top \lambda$.

\medskip

\noindent\underline{\textit{Proof of the First Part of \Cref{thm:structure_opt_ex} (ii)}}
\smallskip

Our proof involves two steps. We first show that for any non-empty $\I\subset \N$, $T_{\I}^\ast=0$, i.e., the post-clearance durations of queues, where all customers join under $\pe$, can be set at zero without loss of optimality.
Note that, if $\I$ is an empty set, we can skip this step.
We then proceed to prove that for any $\I\subset \N$, $T^\ast_{\barI\setminus\{j\}}=0$, where $j\in\argmax_{i\in\barI}\lambda_i$, without loss of optimality. %

\smallskip

\noindent\underline{\textit{Step 1: $T_{\I}^\ast=0$ if $\I$ is non-empty}}

\smallskip
Let $\tp(0_{\I},T_{\barI})$ be the throughput under the exhaustive service policy with post-clearance durations $(0_{\I},T_{\barI})$, and $\tp(T_\J,T_{\barI})$ be the throughput with post-clearance durations $(T_\J,T_{\I\setminus\J}=0,T_{\barI})$ for some $\J\subseteq \I$. 
We will show that
\begin{align*}
 \tp(0_{\I},T_{\barI})\geq \tp(T_\J,T_{\barI})   
\end{align*}
for any $\J\subseteq \I$ and any $T_\J,T_{\barI}\geq 0$.  

For brevity, we drop $T_{\I\setminus\J}=0$ in $(T_\J,T_{\I\setminus\J}=0,T_{\barI})$. 
Let $\alpha(T_\J,T_{\barI})$ be the equilibrium of the exhaustive service policy $(T_\J,T_{\barI})$. Also, let $\bar{J}_i(T_\J,T_{\barI}),~i\in\N$ be the not-joining durations under $(T_\J,T_{\barI})$.
In addition, given $(T_\J,T_{\barI})$, let $\HH =\I\setminus \I(T_\J,T_{\barI})$ be the set of additional queues with strictly positive not-joining durations under $(T_\J,T_{\barI})$ compared with the one under the pure exhaustive service policy $\pe$. 
We have $\alpha_{\HH}(T_\J,T_{\barI})=1$ and $\alpha_{\I\setminus\HH}(T_\J,T_{\barI})<1$. 
Note that $\HH$ could be empty and $\HH$ depends on $(T_\J,T_{\barI})$.

\smallskip
Since $\alpha_{\barI}(0)=1$ by \eqref{eq:alpha_m_n} and the fact that $\alpha(T_\J,T_{\barI})$ is non-decreasing with $(T_\J,T_{\barI})$ by \Cref{lem:piecewise_linear_increasing}, we must have
\begin{align}
\label{eq:pf_alpha_not_joiningqueueus}
  \alpha_{\barI}(T_\J,T_{\barI})=1.  
\end{align}
Besides, by the definition of the set $\HH$, we have $\bar{J}_{\I\setminus\HH}(T_\J,T_{\barI})=0$ (since $\alpha_{\I\setminus\HH}(T_\J,T_{\barI})<1$). 
The not-joining durations $\bar{J}_{\barI\cup \HH}(T_\J,T_{\barI})$ under $(T_\J,T_{\barI})$ can be derived by the periodicity,
\begin{align*}
    \bar{J}_i(T_\J,T_{\barI}) + \alpha_i(T_\J,T_{\barI})\cdot  \theta_i &\overset{(a)}{=} \bar{J}_i(T_\J,T_{\barI}) +  \theta_i\\
    &\overset{(b)}{=} 1^\top \tau + c_{\I}^\top \alpha_{\I}(T_\J,T_{\barI}) + 1^\top_\J T_\J + \sum_{j\in\barI\setminus\{i\}}(c_j+T_j),\quad \forall i\in\barI\cup\HH. 
\end{align*}
The equality $(a)$ holds by \eqref{eq:pf_alpha_not_joiningqueueus} %
and the definition of $\HH$.
The equality $(b)$ holds by \eqref{eq:exhaustiveservice_off_duration_alpha_j_theta_j}, see \Cref{fig:exhaustiveservice_equilibrium} for an illustration.
Thus, we have
\begin{align}
\bar{J}_i(T_\J,T_{\barI})  = 1^\top \tau + c_{\I}^\top \alpha_{\I}(T_\J,T_{\barI}) +1_\J^\top T_\J +  1_{\barI}^\top T_{\barI}+c_{\barI}^\top 1_{\barI} -  c_i - T_i- \theta_i, ~\forall i \in \barI\cup\HH \nonumber.
\end{align}
Plugging the above term and $\bar{J}_{\I\setminus\HH}(T_\J,T_{\barI})=0$ into \eqref{eq:tp_in}, the throughput becomes
\begin{align}
\tp(T_\J,T_{\barI})
& = \sum_{i\in\N}\lambda_i - \sum_{i\in\barI\cup\HH}\lambda_i\cdot \frac{1^\top \tau + c_{\I}^\top \alpha_{\I}(T_\J,T_{\barI}) +1_\J^\top T_\J + 1_{\barI}^\top T_{\barI}+ c_{\barI}^\top 1_{\barI} -  c_i -  \theta_i-T_i}{1^\top \tau+c_{\I}^\top \alpha_{\I}(T_\J,T_{\barI})+ 1^\top_\J T_\J + 1_{\barI}^\top T_{\barI}+ c_{\barI}^\top 1_{\barI} } \nonumber\\
& = \sum_{i\in\N}\lambda_i - \sum_{i\in\barI\cup\HH}\lambda_i\cdot \underbrace{\left(1 - \frac{c_i+\theta_i+T_i}{1^\top \tau+c_{\I}^\top \alpha_{\I}(T_\J,T_{\barI})+ 1^\top_\J T_\J +1_{\barI}^\top T_{\barI}+ c_{\barI}^\top 1_{\barI}  }\right)}_{:=B(T_\J)} \label{eq:tp_ex_k}    
.
\end{align}
Since $\alpha_{\I}(T_\J,T_{\barI})$ is non-decreasing with $T_\J$ given any $T_{\barI}$ by \Cref{lem:piecewise_linear_increasing}, the cardinality of $\HH$ is also non-decreasing with $T_\J$ by the definition of $\HH$, and
the term $B(T_{\J})$, which is positive, is increasing with $T_{\J}$.
Thus, the second summation term in \eqref{eq:tp_ex_k} increases with $T_\J$ for any given $T_{\barI}$. As a result, we have $\tp(T_\J,T_{\barI}) \leq \tp(0_{\J},T_{\barI})$.

\smallskip

\noindent\underline{\textit{Step 2: $T^\ast_{\barI\setminus \{j\}}=0$, where $j\in \argmax_{i\in\barI}\lambda_i$}}

\smallskip
For now, we consider a non-empty all-joining set $\I\subset \N$.
By the analysis of step 1, we can set $T_\I=0$ without loss of optimality. Then, the throughput in \eqref{eq:tp_ex_k} becomes:
\begin{align}
\tp(0_\I,T_{\barI}) =& \left(\sum_{i\in \I\setminus\HH} \lambda_i \right)+ 
\left(\frac{\sum_{i\in\barI\cup\HH}\lambda_i(c_i+\theta_i)}{1^\top \tau+c^\top_{\I}\alpha_{\I}(0_{\I},T_{\barI})+1^\top_{\barI}T_{\barI}+c^\top_{\barI}1_{\barI}}  \right)\\
&+ \frac{\lambda_{\barI}^\top T_{\barI}}{1^\top \tau+c^\top_{\I}\alpha_{\I}(0_{\I},T_{\barI})+1^\top_{\barI}T_{\barI}+c^\top_{\barI}1_{\barI}} ,  \label{eq:tp_e_I}
\end{align}
where we use $T_{\J}=0$ and $T_{\HH}=0$ in \eqref{eq:tp_ex_k} since $\J,\HH\subseteq \I$ and we set $T_{\I}=0$. 
Observe that, given $T_\I=0$, $\alpha_\I(0_\I,T_{\barI})$ \emph{only} depends on $1^\top_{\barI}T_{\barI}$ by \eqref{eq:alpha_m}. This leads to the property that the set $\HH$ also \emph{only} depends on $1^\top_{\barI}T_{\barI}$. Thus, the first two terms and the denominator of the third term in \eqref{eq:tp_e_I} \emph{only} depend on $1^\top_{\barI}T_{\barI}$. 
Then, for any $T_{\barI}$ with throughput \eqref{eq:tp_e_I}, we can construct another post-clearance durations $\tilde{T}_{\barI}=(\tilde{T}_j=1^\top_{\barI}T_{\barI},0_{\barI\setminus\{j\}})$ to achieve a (weakly) higher throughput. 
Specifically, under $\tilde{T}_{\barI}$, all terms in \eqref{eq:tp_e_I} except for the numerator of the last term remain the same since  $1^\top \tilde{T}_{\barI}=1^\top_{\barI}T_{\barI}$ by the construction.
However, we have $\lambda_{\barI}^\top T_{\barI}\leq \lambda_j1^\top_{\barI}T_{\barI}=\lambda^\top_{\barI}\tilde{T}_{\barI}$ by the definition of $j\in \argmax_{i\in \barI}\lambda_i$, leading to $\tp(0_\I,T_{\barI})\leq \tp(0_\I,\tilde{T}_{\barI})$.

When $\I$ is an empty set, $\HH$ is also empty by definition. Thus, the first term in \eqref{eq:tp_e_I} disappears. Moreover, since $\HH$ is empty, it is independent of $T_{\barI}$. Then, by following a similar argument as above, it is easy to show that $\tp(0_\I,T_{\barI})\leq \tp(0_\I,\tilde{T}_{\barI})$, with $\tilde{T}_{\barI}=(\tilde{T}_j=1^\top_{\barI}T_{\barI},0_{\barI\setminus\{j\}})$.

\medskip
\noindent\underline{\textit{Proof of the Second Part of \Cref{thm:structure_opt_ex} (ii)}}

\smallskip

If  $\I=\emptyset$, i.e., $\alpha(0)=1$, due to the monotonicity of $\alpha(T)$ by \Cref{lem:piecewise_linear_increasing}, we know that for any $T\geq 0$, the equilibrium $\alpha(T)=1$. Then, the throughput of $\e$ in this case becomes
\begin{align*}
    \tp(T) = \frac{u^\top c+ \lambda^\top T}{1^\top \tau+1^\top c+1^\top T}.
\end{align*}

When $\I=\emptyset$, by \Cref{thm:structure_opt_ex} (ii), it is without loss of optimality to set $T_{-j}=0$, where $j$ is the queue with the largest arrival rate among all queues since $\barI(0)=\N$. Setting $T_{-j}=0$, we have the partial derivative of throughput with respect to $T_j$ as follows:
\begin{align}
\label{eq:derivative_tp_exhaustiveservice}
\frac{\partial \tp(T_j,T_{-j}=0) }{\partial T_j }  = \frac{\lambda_j\cdot 1^\top \tau-\sum_{i\in\N}(\mu_i-\lambda_i)c_i}{(1^\top \tau+1^\top c+T_j)^2},
\end{align}
whose sign is independent of $T_j$.
Thus, when the numerator is positive, the throughput is always increasing with $T_j$. In this case, always serving the queue with the largest arrival rate is optimal. Otherwise, the throughput is decreasing with $T_j$, leading to $T_j^\ast=0$. This completes the proof. \Halmos
\end{proof}

\medskip

\begin{proof}{\textbf{Proof of \Cref{thm:opt_ex_algo}.}}

We first demonstrate that \Cref{alg:opt} terminates in at most $2\bar{n}$ steps. Then, we show that the algorithm indeed finds an optimal exhaustive service policy.

The algorithm takes at most $\bar{n}$ steps to derive the equilibrium set $\I(0)$ in line 1, by \Cref{prop:n_step}. Since $\I(0) \leq \bar{n}$ (as noted in \Cref{lem:equilibriumset}), the ``for loop" in lines~5--11 takes at most $\bar{n}$ steps. This completes the proof that \Cref{alg:opt} terminates in at most $2\bar{n}$ steps.

We now proceed to show that \Cref{alg:opt} indeed finds an optimal exhaustive service policy. 
By the structural result in \Cref{thm:structure_opt_ex}, all queues have zero post-clearance durations, except possibly for queue \( j \in \argmax_{i \in \barI(0)} \lambda_i \). Thus, in line 1, we first derive the equilibrium (all-joining) set $\I(0)$ of the pure exhaustive service policy.  This can be done by \Cref{alg:pivoting}.
If $\I(0)=\N$, i.e., all customers of all queues are willing to join, the pure exhaustive service policy $\pe$ achieves the first-best $1^\top \lambda$. 
This is what line 2 does.
If the equilibrium set $\I(0)$ is empty, the algorithm skips line 4 and the subsequent ``for loop'' procedure. By the second part of \Cref{thm:structure_opt_ex} (ii), either the pure exhaustive service policy or always serving the queue with the highest arrival rate is optimal. This is implemented in lines~12--13. 
If $\I(0) \neq \N$ and non-empty, then by \Cref{thm:structure_opt_ex}, only queue $j$ can have a positive post-clearance duration. The selection of queue $j$ is done in line 3 based on \Cref{thm:structure_opt_ex}. According to \Cref{lem:monoton_tp_hat_T} below, we know that the $\mathcal{K}$ is the order that queues transition from all-joining, i.e., zero not-joining duration, to not-all-joining as $T_j$ increases. 
Again by \Cref{lem:monoton_tp_hat_T}, it is sufficient to check the throughput only at the boundary values, as done in line 6. The new equilibrium is then easily computed in lines~7--8 also based on \Cref{lem:monoton_tp_hat_T}. Lines~10--11 compare the throughput at the new boundary value with the maximum throughput found at previous boundary values. Finally, lines~12--15 check whether always serving queue $j$ is optimal.\Halmos

\end{proof}

\medskip

\begin{lemma}
\label{lem:monoton_tp_hat_T}
Consider a non-empty $\I(0)\subset \N$.
For $\ell=1,2,\dots,|\I(0)|$ and queue $j\in \argmax_{i\in \barI(0)}\lambda_i$, let $T_j^{(k_\ell)}$ be the value given by line $6$ in \Cref{alg:opt} when the queue index $k_\ell \in \mathcal{K}$, where $\mathcal{K}$ is given by line $4$ in \Cref{alg:opt}.
Let $T_j^{(k_0)}:=0$. 
Then, 
\begin{enumerate}
    \item the equilibrium set $\I(T_j=T_j^{(k_\ell)},0_{-j})=\I(0)\setminus\{k_1,\dots,k_\ell\}$ for any $\ell=1,2,\dots,|\I(0)|$;
    \item the throughput $\tp(T_j,0_{-j})$ is monotone in $T_j\in[T_j^{(k_\ell)},T_j^{(k_{\ell+1})})$ for any $\ell=0,1,2,\dots, |\I(0)|-1$.
\end{enumerate}
\end{lemma}

\begin{proof}{\textbf{Proof of \Cref{lem:monoton_tp_hat_T}.}}
The general idea behind \Cref{lem:monoton_tp_hat_T} is that when we increase $T_j$ from zero (i.e., from the pure exhaustive service policy $\pe$), the equilibrium variable $\alpha_{\I(0)}(T_j,0_{-j})$ under $(T_j,0_{-j})$ will gradually increase toward one, while $\alpha_{\barI(0)}(0,T_{-j})=1_{\barI(0)}$ will not change since they have already achieved the possibly maximum value $1$ (by monotonicity in \Cref{lem:piecewise_linear_increasing}).
The value $T_j^{(k_\ell)}$ is the one where $\alpha_{k_\ell}(T_j^{(k_\ell)},0_{-j})$ \textit{exactly} reaches one, while %
all other $\alpha_{k_{\ell^\prime}}(T_j^{(k_\ell)},0_{-j})$ ($\ell^\prime > \ell$) remain below one.

\smallskip

\noindent\underline{\textit{Proof of \Cref{lem:monoton_tp_hat_T} (i)}}

We use the following observation to assist our proof.
\begin{observation}
\label{obs:same_eq_set}
If for some queue $i$, we have $\alpha_i = 1$ and $(A\alpha)_i = b_i(T)$ at equilibrium (i.e., both constraints \eqref{con:ir} and \eqref{con:period} are binding for queue $i$), then we can either include or exclude queue $i$ from the equilibrium set $\I(T)$. Both choices result in the same value of the equilibrium variable $\alpha(T)$, since it is unique by \Cref{prop:n_step}.
\end{observation}

We provide the proof for the case $\ell = 1$; other cases follow similarly. Based on \Cref{obs:same_eq_set}, to verify that the equilibrium set $\I(T_j = T_j^{(k_1)}, 0_{-j}) = \I(0) \setminus \{k_1\}$, it is sufficient to show that the $\I(0)$-induced variable $\alpha$, i.e., \eqref{eq:alpha_m}--\eqref{eq:alpha_m_n} with $\I=\I(0)$ and $T=(T_j = T_j^{(k_1)}, 0_{-j})$, not only satisfies the equilibrium conditions, i.e., the solution of problem \ref{prob:alpha} with $T=(T_j = T_j^{(k_1)}, 0_{-j})$, but also ensures that $\alpha_{k_1} = 1$. For simplicity, we will write $\I(0)$ as $\I$. The variable $\alpha$ induced by set $\I$ under $(T_j = T_j^{(k_1)}, 0_{-j})$ is:
\begin{align}
\label{eq:pf_alpha_T_j}
    \alpha_i = \frac{1^\top \tau + T_j^{(k_1)} + c_{\barI}^\top 1_{\barI}}{1 - \rho_{\I}^\top 1_{\I}} \cdot \frac{\mu_i - \lambda_i}{\mu_i \theta_i}, \quad \forall i \in \I; \quad \alpha_{\barI} = 1_{\barI},
\end{align}
where the first equation uses \Cref{lem:inverse_mat} regarding the closed-form expression for $A^{-1}_\I 1_\I$ appearing in \eqref{eq:alpha_m}.

By the proof of \Cref{prop:n_step}, we can discard constraint \eqref{con:non_neg} in problem \ref{prob:alpha}. Thus, it is sufficient to verify that the variable $\alpha$ in \eqref{eq:pf_alpha_T_j} satisfies \eqref{con:ir}, \eqref{con:period}, \eqref{con:comp}, and $\alpha_{k_1} = 1$. 
Observe that, constraint \eqref{con:comp} is automatically satisfied since the $\I$-induced $\alpha$ is derived by setting $\alpha_{\barI} = 1$ and $(A\alpha)_\I = b_\I(T_j^{(k_1)}, 0_{-j})$.

For any $i \in \I$, we have 
\begin{align*}
\alpha_i = \underbrace{\frac{1^\top \tau + T_j^{(k_1)} + c_{\barI}^\top 1_{\barI}}{1 - \rho_{\I}^\top 1_{\I}}}_{\textrm{same for all $i \in \I$}} \cdot \frac{\mu_i - \lambda_i}{\mu_i \theta_i} \overset{(a)}{\leq} \frac{1^\top \tau + T_j^{(k_1)} + c_{\barI}^\top 1_{\barI}}{1 - \rho_{\I}^\top 1_{\I}} \cdot \frac{\mu_{k_1} - \lambda_{k_1}}{\mu_{k_1} \theta_{k_1}} \overset{(b)}{=} 1,
\end{align*}
where $(a)$ holds because $k_1 \in \argmax_{i \in \I} \frac{\mu_i - \lambda_i}{\mu_i \theta_i}$ and $(b)$ holds by substituting the expression for $T_j^{(k_1)}$. 
As a result, we have $\alpha_\I \leq 1$ and $\alpha_{k_1} = 1$.
Considering $\alpha_{\barI}=1$, we conclude that the variable $\alpha$ in \eqref{eq:pf_alpha_not_joiningqueueus} satisfies \eqref{con:ir} and $\alpha_{k_1}=1$.

Since $(A\alpha)_\I = b_\I(T_j^{(k_1)}, 0_{-j})$, we are left to show that $(A\alpha)_{\barI} \leq b_{\barI}(T_j^{(k_1)}, 0_{-j})$, Since $\alpha_{\barI} = 1$, $(A\alpha)_{\barI} \leq b_{\barI}(T_j^{(k_1)}, 0_{-j})$ is equivalent to
\begin{align}
\label{eq:to_prove_period}
    \theta_i - c^\top_{\barI \setminus \{i\}} 1 - c^\top_{\I} \alpha_{\I}  \leq 1^\top \tau + T_j^{(k_1)}, \quad \forall i \in \barI.
\end{align}
Besides, since $\I$ is also the equilibrium set of the pure exhaustive service policy $\pe$, we have $(A\alpha(0))_{\barI} \leq b_{\barI}(0)$, i.e.,
\begin{align}
\label{eq:condition_pure_ex}
    \theta_i - c^\top_{\barI \setminus \{i\}} 1 - c^\top_{\I} \alpha_{\I}(0) \leq 1^\top \tau, \quad \forall i \in \barI.
\end{align}
Additionally, we have
\begin{align*}
    \alpha_i = \frac{1^\top \tau + T_j^{(k_1)} + c_{\barI}^\top 1_{\barI}}{1 - \rho_{\I}^\top 1_{\I}} \cdot \frac{\mu_i - \lambda_i}{\mu_i \theta_i} > \alpha_i(0) = \frac{1^\top \tau + c_{\barI}^\top 1_{\barI}}{1 - \rho_{\I}^\top 1_{\I}} \cdot \frac{\mu_i - \lambda_i}{\mu_i \theta_i}, \quad \forall i \in \I,
\end{align*}
and thus \eqref{eq:condition_pure_ex} implies \eqref{eq:to_prove_period}. %

\smallskip

\noindent\textit{\underline{Proof of \Cref{lem:monoton_tp_hat_T} $(ii)$}}

It is sufficient to show that for any $\ell = 0, 1, 2, \dots, |\I(0)| - 1$, the sign of the partial derivative $\partial \tp(T_j, 0_{-j}) / \partial T_j$ does \textit{not} depend on $T_j \in [T_j^{(k_\ell)}, T_j^{(k_{\ell+1})})$.

Based on (i) and its proof, we know that when $T_j \in [T_j^{(k_\ell)}, T_j^{(k_{\ell+1})})$, the equilibrium (all-joining) set $\I(T_j, 0_{-j}) = \I(0) \setminus \{k_1, \dots, k_\ell\}$ remains unchanged. For simplicity, put $\HH =\I(0)\setminus \I(T_j, 0_{-j}) $ as the set of additional queues that are not all-joining under the exhaustive service policy $(T_j,0_{-j})$ compared with the ones under $\pe$.
We have, $\alpha_{\HH}(T_j,0_{-j})=1$ and $\alpha_{\I(0)\setminus\HH}(T_j,0_{-j})<1$.
In what follows, we write $\I(0)$ as $\I$ to simplify notation. 

Given $\HH$ under $(T_j, 0_{-j})$, by \eqref{eq:tp_e_I}, we have the throughput \eqref{eq:tp_in} as follows
\begin{align}
\nonumber
 \tp(T_j, 0_{-j})=&\left( \sum_{i\in \I \setminus \HH} \lambda_i\right) +\left( \frac{\sum_{i \in \barI \cup \HH} \lambda_i (c_i + \theta_i)}{1^\top \tau + c^\top_{\I \setminus \HH} \alpha_{\I \setminus \HH}(T_j, 0_{-j}) + c^\top_{\barI \cup \HH} 1 + T_j}\right)  \\
 &+ \frac{\lambda_j T_j}{1^\top \tau + c^\top_{\I \setminus \HH} \alpha_{\I \setminus \HH}(T_j, 0_{-j}) + c^\top_{\barI \cup \HH} 1 + T_j}.
\end{align}

By \eqref{eq:alpha_m}, we have
\begin{align*}
\alpha_{\I \setminus \HH}(T_j, 0_{-j}) &= T_j A_{\I \setminus \HH}^{-1} 1 + (1^\top \tau + c^\top_{\barI \cup \HH} 1) A_{\I \setminus \HH}^{-1} 1 \nonumber \\
&= \frac{T_j}{1 - \rho_{\I \setminus \HH}^\top 1} \cdot 1^\top \frac{\mu_{\I \setminus \HH} - \lambda_{\I \setminus \HH}}{\mu_{\I \setminus \HH} \odot \theta_{\I \setminus \HH}} + (1^\top \tau + c^\top_{\barI \cup \HH} 1) \cdot 1^\top \frac{\mu_{\I \setminus \HH} - \lambda_{\I \setminus \HH}}{\mu_{\I \setminus \HH} \odot \theta_{\I \setminus \HH}},
\end{align*}
where the second equality holds by \Cref{lem:inverse_mat}. In the above equation, $x \odot y$ denotes the \emph{element-wise} product of two vectors $x, y \in \mathbb{R}^n$. Similarly, with a little abuse of notation, in the above equation, $\frac{\mu_{\I \setminus \HH} - \lambda_{\I \setminus \HH}}{\mu_{\I \setminus \HH} \odot \theta_{\I \setminus \HH}}$ represents the \emph{element-wise} division. Therefore, we have
\begin{align*}
 c^\top_{\I \setminus \HH} \alpha_{\I \setminus \HH}(T_j, 0_{-j}) = \frac{\rho_{\I \setminus \HH}^\top 1}{1 - \rho_{\I \setminus \HH}^\top 1} T_j + (1^\top \tau + c^\top_{\barI \cup \HH} 1) \cdot \rho_{\I \setminus \HH}^\top 1.
\end{align*}

Based on this formula, the partial derivative can be calculated as follows:
\begin{align}
\frac{\partial \tp(T_j,0_{-j})}{\partial T_j}& = -\frac{1}{\Delta^2} \left[\frac{\rho_{\I\setminus\HH}^\top 1}{1-\rho_{\I\setminus\HH}^\top 1}+1\right]\cdot\left[\lambda^{\top}_{\barI\cup\HH}(c_{\barI\cup\HH}+\theta_{\barI\cup\HH})\right]\nonumber \\
&\quad  + \frac{1}{\Delta^2}\left[\lambda_jT_j\cdot\left(\frac{\rho_{\I\setminus\HH}^\top 1}{1-\rho_{\I\setminus\HH}^\top 1}+1\right)+\lambda_j\cdot\left(1^\top \tau+(1^\top \tau+c^\top_{\barI\cup\HH}1)\cdot \rho_{\I\setminus\HH}^\top 1 +c^\top_{\barI\cup\HH}1\right)\right]\nonumber\\
&\quad -  \frac{1}{\Delta^2}\cdot \lambda_jT_j\cdot \left(\frac{\rho_{\I\setminus\HH}^\top 1}{1-\rho_{\I\setminus\HH}^\top 1}+1\right)\nonumber\\
& =\frac{1}{\Delta^2}\cdot \left(\lambda_j\cdot\Bigg(1^\top \tau+(1^\top \tau+c^\top_{\barI\cup\HH}1)\cdot \rho_{\I\setminus\HH}^\top 1 +c^\top_{\barI\cup\HH}1\right) \nonumber \\ 
&\qquad\qquad - \left[\frac{\rho_{\I\setminus\HH}^\top 1}{1-\rho_{\I\setminus\HH}^\top 1}+1\right]\cdot\left[\lambda^{\top}_{\barI\cup\HH}(c_{\barI\cup\HH}+\theta_{\barI\cup\HH})\right]\Bigg),  
\end{align}
where $\Delta=1^\top \tau+c^\top_{\I\setminus\HH}\alpha_{\I\setminus\HH}(T_j,0_{-j})+c^\top_{\barI\cup\HH}1+T_j$. Observe that the sign of the partial derivative is \textit{independent} of $T_j$. This completes the proof.
\Halmos
\end{proof}

\medskip

\subsection{Proofs in Appendix \ref{app_sec:auxiliary_result}}

\begin{proof}{\textbf{Proof of \Cref{thm:fixed_duration_equilibrium_mui_smaller}.}}
The proof follows in a manner similar to that of \Cref{thm:equilibrium_fixedduration_exhaustive} and \Cref{thm:fixed_duration_nonexhaustive_equilibrium}. For brevity, we omit the details.\Halmos
\end{proof}

\medskip

\begin{proof}{\textbf{Proof of \Cref{prop:vacation_opt_duration}.}}
There are three possibilities of $T_i(L_i,\bar{L}_i)$ as stated in \Cref{coro:fixed_duration_post_clearance_duration}.
We will derive the optimal solution in each case separately and then compare them to find the globally optimal solution.

\smallskip

\noindent\underline{\textit{Case 1: $T_i(L_i,\bar{L}_i)=0$.}}
In this case, we have throughput as follows:
\begin{align*}
   \tp(L_i,\bar{L}_i) =  \frac{\mu_iL_i}{L_i+\bar{L}_i} = \frac{\mu_i}{1 + \bar{L}_i/L_i}.
\end{align*}
Maximizing throughput is thus equivalent to minimizing $\bar{L}_i/L_i$.
The constraints to make $T_i(L_i,\bar{L}_i)=0$ as the equilibrium case are
\begin{align*}
 &0 \geq L_i - \frac{\lambda_i}{\mu_i-\lambda_i}\theta_i,\\
 &0 \geq L_i - \frac{\lambda_i}{\mu_i-\lambda_i}\bar{L}_i.
\end{align*}
Besides, we have the constraints \eqref{constraint:work_limit} and \eqref{constraint:forced_vacation}. 
Since the throughout is decreasing with $\bar{L}_i/L_i$, we have $\bar{L}_i^\ast= \max\{\beta_i,~\frac{\mu_i-\lambda_i}{\lambda_i}\}L_i^\ast$. 
The \on duration $L_i$ can be chosen freely. One possible choice is $L_i^\ast=\min\left\{\frac{\lambda_i}{\mu_i-\lambda_i}\theta_i,~ L^{\textrm{max}}_i\right\}$.
The corresponding throughput is
\begin{align}
\label{eq:vacational_opt_solution_case1}
  \tp(L_i,\bar{L}_i) = 
  \begin{cases}
    \frac{\mu_i}{1+\beta_i},\quad & \textrm{$\beta_i>\frac{\mu_i-\lambda_i}{\lambda_i}$}, \\
    \lambda_i, & \textrm{$\beta_i\leq \frac{\mu_i-\lambda_i}{\lambda_i}$}.  
  \end{cases}
\end{align}

\smallskip

\noindent\underline{\textit{Case 2: $T_i(L_i,\bar{L}_i)=L_i-\frac{\lambda_i}{\mu_i-\lambda_i}\theta_i$.}}
In this case, we have throughput as follows
\begin{align*}
   \tp(L_i,\bar{L}_i) = \frac{\lambda_i\cdot \left(L_i - \frac{\lambda_i}{\mu_i-\lambda_i}\theta_i\right) + \mu_i\cdot \frac{\lambda_i}{\mu_i-\lambda_i}\theta_i}{L_i+\bar{L}_i} = \frac{\lambda_iL_i+\lambda_i\theta_i}{L_i+\bar{L}_i} = \lambda_i - \frac{\lambda_i\cdot (\bar{L}_i - \theta_i)}{L_i+\bar{L}_i}.
\end{align*}
The equilibrium constraints are
\begin{align*}
 L_i-\frac{\lambda_i}{\mu_i-\lambda_i}\theta_i \geq 0 \quad &\Longleftrightarrow \quad L_i\geq \frac{\lambda_i}{\mu_i-\lambda_i}\theta_i,\\
 L_i-\frac{\lambda_i}{\mu_i-\lambda_i}\theta_i\geq L_i - \frac{\lambda_i}{\mu_i-\lambda_i}\bar{L}_i \quad &\Longleftrightarrow \quad \bar{L}_i\geq \theta_i,
\end{align*}
and \eqref{constraint:work_limit}--\eqref{constraint:forced_vacation}. 
\begin{enumerate}
    \item \textit{When $L^{\textrm{max}}_i<\frac{\lambda_i \theta_i}{\mu_i-\lambda_i}$}, there is no feasible solution. That is, in this case, there is no $(L_i,\bar{L}_i)$ to induce the equilibrium variable $T_i(L_i,\bar{L}_i)=L_i-\frac{\lambda_i\theta_i}{\mu_i-\lambda_i}$.
    \item \textit{When $L^{\textrm{max}}_i\geq \frac{\lambda_i \theta_i}{\mu_i-\lambda_i}$ and $\frac{1}{\beta_i}>\frac{L^{\textrm{max}}_i}{\theta_i}$}, there is still no feasible solution.
    \item \textit{When $L^{\textrm{max}}_i\geq \frac{\lambda_i \theta_i}{\mu_i-\lambda_i}$ and $\frac{1}{\beta_i}\leq \frac{L^{\textrm{max}}_i}{\theta_i}$}:
    \begin{enumerate}
        \item \textit{if $\frac{1}{\beta_i}< \frac{\lambda_i}{\mu_i-\lambda_i}$}, in this case, one can check that $L_i^\ast = \frac{\lambda_i}{\mu_i-\lambda_i}\theta_i$ and $\bar{L}_i^\ast=\beta_i L_i^\ast$. The throughput becomes $\frac{\mu_i}{1+\beta_i}$, which is the same as the one in \eqref{eq:vacational_opt_solution_case1};
        \item \textit{if $\frac{1}{\beta_i}\geq \frac{\lambda_i}{\mu_i-\lambda_i}$}, in this case, $\bar{L}_i^\ast=\theta_i$ with any feasible $L_i^\ast$ is optimal, and it has achieved the first-best throughput $\lambda_i$.
    \end{enumerate}
\end{enumerate}

\smallskip

\noindent\underline{\textit{Case 3: $T_i(L_i,\bar{L}_i)=L_i-\frac{\lambda_i}{\mu_i-\lambda_i}\bar{L}_i$.}}
In this case, we have throughput as follows
\begin{align*}
   \tp(L_i,\bar{L}_i) = \frac{\lambda_i\cdot \left(L_i - \frac{\lambda_i}{\mu_i-\lambda_i}\bar{L}_i\right) + \mu_i\cdot \frac{\lambda_i}{\mu_i-\lambda_i}\bar{L}_i}{L_i+\bar{L}_i} = \frac{\lambda_i(L_i+\bar{L}_i)}{L_i+\bar{L}_i} = \lambda_i.
\end{align*}
The equilibrium constraints are
\begin{align}
 L_i-\frac{\lambda_i}{\mu_i-\lambda_i}\bar{L}_i \geq 0 \quad &\Longleftrightarrow \quad L_i\leq \frac{\lambda_i}{\mu_i-\lambda_i}\bar{L}_i,  \label{contraint:vacation_case3_1}\\
 L_i-\frac{\lambda_i}{\mu_i-\lambda_i}\bar{L}_i\geq L_i - \frac{\lambda_i}{\mu_i-\lambda_i}\theta_i \quad &\Longleftrightarrow \quad \bar{L}_i\leq \theta_i \label{contraint:vacation_case3_2},
\end{align}
and \eqref{constraint:work_limit}--\eqref{constraint:forced_vacation} . 

One can easily verify that when $\beta_i>\frac{\mu_i-\lambda_i}{\lambda_i}$, no $(L_i,\bar{L}_i)$ satisfies \eqref{contraint:vacation_case3_1}, \eqref{contraint:vacation_case3_2}, \eqref{constraint:work_limit} and \eqref{constraint:forced_vacation}. When $\beta_i\leq \frac{\mu_i-\lambda_i}{\lambda_i}$, any $(L_i,\bar{L}_i)$ satisfying \eqref{contraint:vacation_case3_1}, \eqref{contraint:vacation_case3_2}, \eqref{constraint:work_limit} and \eqref{constraint:forced_vacation} is an optimal solution, which achieves the first-best throughput $\lambda_i$. 
One possible choice is $L_i^\ast=\min\{\frac{\lambda_i}{\mu_i-\lambda_i}\theta_i,~L^{\textrm{max}}_i\}$ and $\bar{L}_i^\ast=\beta_i L_i^\ast$. The other choice is $L_i^\ast=\min\{\frac{\lambda_i}{\mu_i-\lambda_i}\theta_i,~L^{\textrm{max}}_i\}$ and $\bar{L}_i^\ast=\frac{\mu_i-\lambda_i}{\lambda_i} L_i^\ast$.
\medskip

Combining the above discussions and \Cref{coro:fixed_duration_post_clearance_duration}, we have the optimal solutions, the equilibrium-induced post-clearance duration, and the throughput as stated in \Cref{prop:vacation_opt_duration}.

We now investigate whether the equilibrium outcomes are exhaustive or non-exhaustive.

\begin{enumerate}
    \item When \(\beta_i > \frac{\mu_i - \lambda_i}{\lambda_i}\), based on the above analysis, the optimal \on-\off durations satisfy
    \begin{align*}
        L_i^\ast = \frac{\bar{L}_i^\ast}{\beta_i} < \frac{\lambda_i}{\mu_i - \lambda_i} \bar{L}_i^\ast.
    \end{align*}
    By \Cref{thm:equilibrium_fixedduration_exhaustive} and \Cref{thm:fixed_duration_nonexhaustive_equilibrium}, we have
    \begin{itemize}[--]
        \item If \(\bar{L}_i^\ast \geq \theta_i\), the equilibrium outcome is exhaustive, as depicted in Figures \ref{subfig:exh_case1}--\ref{subfig:exh_case3}. 
        \item If \(\bar{L}_i^\ast < \theta_i\), the equilibrium outcome is non-exhaustive, as illustrated in \Cref{fig:fixed_duration_nonexhaustive_equilibrium}.
    \end{itemize}
Whether the optimal \off duration $\bar{L}_i^\ast$ is greater or smaller than the waiting patience $\theta_i$ further depends on $L^{\textrm{max}}_i$.
    We further distinguish the following subcases:
    \begin{enumerate}
        \item If \(L^{\textrm{max}}_i \geq \frac{\lambda_i \theta_i}{\mu_i - \lambda_i}\), then
        \[
            L_i^\ast = \frac{\lambda_i \theta_i}{\mu_i - \lambda_i}.
        \]
        This implies
        \[
            \bar{L}_i^\ast = \beta_i L_i^\ast = \beta_i \cdot \frac{\lambda_i \theta_i}{\mu_i} > \theta_i,
        \]
        resulting in an exhaustive equilibrium outcome.

        \item If \(L^{\textrm{max}}_i < \frac{\lambda_i \theta_i}{\mu_i - \lambda_i}\) and \(\beta_i \in \bigl(\tfrac{\mu_i - \lambda_i}{\lambda_i}, \tfrac{\theta_i}{L^{\textrm{max}}_i}\bigr)\), then
        \[
            \bar{L}_i^\ast = \beta_i L_i^\ast = \beta_i L^{\textrm{max}}_i < \frac{\theta_i}{L^{\textrm{max}}_i}L^{\textrm{max}}_i = \theta_i,
        \]
        yielding a non-exhaustive equilibrium outcome.

        \item If \(L^{\textrm{max}}_i < \frac{\lambda_i \theta_i}{\mu_i - \lambda_i}\) and \(\beta_i \geq \tfrac{\theta_i}{L^{\textrm{max}}_i}\), then
        \[
            \bar{L}_i^\ast = \beta_i L_i^\ast = \beta_i L^{\textrm{max}}_i \geq \frac{\theta_i}{L^{\textrm{max}}_i} L^{\textrm{max}}_i = \theta_i,
        \]
        leading to an exhaustive equilibrium outcome.
    \end{enumerate}

    \item When \(\beta_i \leq \frac{\mu_i - \lambda_i}{\lambda_i}\), the system achieves the first-best throughput. This corresponds to the scenario in \Cref{subfig:exh_case4}, which is an exhaustive equilibrium outcome. \Halmos
\end{enumerate}
\end{proof}

\medskip

\begin{proof}{\textbf{Proof of \Cref{prop:twoqueues}.}}
We will prove each case separately.  

\smallskip

\noindent\underline{\textit{Case (i): $\I(0)=\emptyset$}}
\smallskip

Given $\I(0)=\emptyset$, we have $\alpha(0)=1$. To make it an equilibrium, we need to guarantee $A\alpha(0)\leq b(0)$, i.e., $A1\leq b(0)$.
This leads to the equilibrium condition stated in (i). By \Cref{thm:structure_opt_ex}, either $\pe$ is an optimal exhaustive service policy, or always serving the queue with the highest arrival rate is optimal. Thus, we only need to compare the throughput $\tp(0)$ with $\max\{\lambda_1,\lambda_2\}$. 
We provide the proof of $\lambda_1>\lambda_2$; the other case is similar.

The throughput $\tp(0)$ in this case is
\begin{align*}
\tp(0) = \frac{\frac{\mu_1\lambda_1}{\mu_1-\lambda_1}\theta_1+\frac{\mu_2\lambda_2}{\mu_2-\lambda_2}\theta_2}{1^\top \tau+\frac{\lambda_1\theta_1}{\mu_1-\lambda_1}+\frac{\lambda_2\theta_2}{\mu_2-\lambda_2}}.
\end{align*}
Suppose that $\lambda_1>\tp(0)$, i.e., always serving queue $1$ is better than $\pe$, we have
\begin{align*}
    \lambda_1>\tp(0) \quad &\Longleftrightarrow \quad \lambda_1>\frac{\frac{\mu_1\lambda_1}{\mu_1-\lambda_1}\theta_1+\frac{\mu_2\lambda_2}{\mu_2-\lambda_2}\theta_2}{1^\top \tau+\frac{\lambda_1\theta_1}{\mu_1-\lambda_1}+\frac{\lambda_2\theta_2}{\mu_2-\lambda_2}}\\
    & \Longleftrightarrow \quad 1^\top \tau + \frac{\lambda_2\theta_2}{\mu_2-\lambda_2}\cdot \left(1-\frac{\mu_2}{\lambda_1}\right) > \theta_1.
\end{align*}
One can verify that the intersection of the above inequality with the equilibrium conditions of $\I(0)=\emptyset$ is non-empty. This leads to the result that when $\lambda_1>\lambda_2$, if the above inequality holds, always serving queue $1$ is optimal; otherwise, $\pe$ is an optimal exhaustive service policy.

\smallskip

\noindent\underline{\textit{Case (ii): $\I(0)=\{1\}$}}

\smallskip

Assuming $\I(0)=\{1\}$, by \eqref{eq:alpha_m}--\eqref{eq:alpha_m_n}, we have the $\alpha(0)$ stated in \Cref{prop:twoqueues} (ii). To make it an equilibrium, we need to guarantee $\alpha_1\leq 1$ and $(A\alpha)_2\leq b_2(0)$. This leads to the condition in (ii). 
Using the throughput formula \eqref{eq:tp_in}, we have  $\tp(0)$ with the equilibrium set $\I(0)=\{1\}$ as follows
\begin{align}
\label{eq:pf_twoqueue_tp_pure}
    \tp(0) = \lambda_1+\frac{\lambda_2\cdot(c_2+\theta_2)}{1^\top \tau+c_1\alpha_1+c_2} = \lambda_1 + \frac{\lambda_2\cdot\frac{\mu_2\theta_2}{\mu_2-\lambda_2}}{\frac{\mu_1}{\mu_1-\lambda_1}\cdot \left(1^\top \tau+\frac{\lambda_2\theta_2}{\mu_2-\lambda_2}\right)}.
\end{align}
By \Cref{thm:structure_opt_ex} and \Cref{thm:opt_ex_algo}, since $\barI(0)=\{2\}$, it is sufficient to compare the above throughput with $\lambda_2$, i.e., always serving queue $2$, and $\tp(0,T_2)$ with $T_2$ given by line $6$ in \Cref{alg:opt}, i.e.,
\begin{align*}
    T_2 = \frac{\mu_1\theta_1}{\mu_1-\lambda_2}\cdot (1-\rho_1) - 1^\top \tau-c_2 = \theta_1-1^\top \tau-c_2.
\end{align*}
We have $\tp(0,T_2)$ as follows
\begin{align}
    \tp(0,T_2) &= \lambda_1+\frac{\lambda_2\cdot(c_2+\theta_2)+\lambda_2T_2}{1^\top \tau+c_1+c_2+T_2} = \lambda_1+\frac{\lambda_2\cdot(c_2+\theta_2)+\lambda_2\cdot (\theta_1-1^\top \tau-c_2)}{1^\top \tau+c_1+c_2+\theta_1-1^\top \tau-c_2}\nonumber\\
    & = \lambda_1+\frac{\lambda_2\cdot(\theta_1+\theta_2-1^\top \tau)}{c_1+\theta_1}.\label{eq:pf_twoqueue_tp_T2}
\end{align}

To summarize, we have three candidate policies that may be optimal: the pure exhaustive policy $\pe$, $T = (0, T_2=\theta_1 - 1^\top\tau - c_2)$, and the policy always serving queue~2. 
We first present a condition regarding the comparison between $\tp(0,T_2)$ and $\tp(0)$.
We have
\begin{align*}
   \tp(0,T_2) \geq  \tp(0) \quad &\Longleftrightarrow \quad \frac{\lambda_2\cdot(\theta_1+\theta_2-1^\top \tau)}{c_1+\theta_1} \geq \frac{\lambda_2\cdot\frac{\mu_2\theta_2}{\mu_2-\lambda_2}}{\frac{\mu_1}{\mu_1-\lambda_1}\cdot \left(1^\top \tau+\frac{\lambda_2\theta_2}{\mu_2-\lambda_2}\right)} \\
   &\Longleftrightarrow \quad \frac{\theta_1+\theta_2-1^\top \tau}{\frac{\mu_1\theta_1}{\mu_1-\lambda_1}} \geq \frac{\frac{\mu_2\theta_2}{\mu_2-\lambda_2}}{\frac{\mu_1}{\mu_1-\lambda_1}\cdot \left(1^\top \tau+\frac{\lambda_2\theta_2}{\mu_2-\lambda_2}\right)}\\
&\Longleftrightarrow \quad 1 + \frac{\theta_2-1^\top \tau}{\theta_1} \geq \frac{\frac{\mu_2\theta_2}{\mu_2-\lambda_2}}{ 1^\top \tau+\frac{\lambda_2\theta_2}{\mu_2-\lambda_2}}\\
&\Longleftrightarrow \quad \frac{\theta_2-1^\top \tau}{\theta_1} \geq \frac{\theta_2-1^\top \tau}{ 1^\top \tau+\frac{\lambda_2\theta_2}{\mu_2-\lambda_2}}.
\end{align*}
Since $\I(0)=\{1\}$ already requires $\theta_1\geq1^\top \tau+\frac{\lambda_2\theta_2}{\mu_2-\lambda_2}$, the above inequality holds if and only if $\theta_2\leq 1^\top \tau$. That is to say, given $\I(0)=\{1\}$, $\tp(0,T_2)\geq \tp(0)$ if and only if $\theta_2\leq 1^\top \tau$. 

\smallskip

In what follows, we discuss the case of $\theta_2\leq 1^\top \tau$ and $\theta_2\geq 1^\top \tau$ separately. 
Note that $\theta_2=1^\top \tau$ can belong to either case.

\begin{enumerate}[(a)]

    \item \textit{The case of $\theta_2\leq 1^\top \tau$.}
    In this situation, based on the above argument, we have $\tp(0,T_2) \geq \tp(0)$, so we can rule out the pure exhaustive policy from the potential optimal policies.
Thus, we only need to compare $\tp(0,T_2)$ with $\lambda_2$.

Suppose now that always serving queue~2 is optimal, i.e., $\lambda_2 > \tp(0,T_2)$.
Because $\tp(0,T_2)$ is strictly larger than $\lambda_1$ (due to \eqref{eq:pf_twoqueue_tp_T2} and the condition $\theta_2 \leq 1^\top\tau$), we must have $\lambda_2 > \lambda_1$.
Under this condition, $\lambda_2 > \tp(0,T_2)$ is equivalent to
\begin{align*}
    \lambda_2 >  \tp(0,T_2) \quad &\Longleftrightarrow \quad \lambda_2 > \lambda_1+\frac{\lambda_2\cdot(\theta_1+\theta_2-1^\top \tau)}{\frac{\mu_1\theta_1}{\mu_1-\lambda_1}} \\
    & \Longleftrightarrow \quad \frac{\mu_1\theta_1}{\mu_1-\lambda_1} - \frac{\lambda_1}{\lambda_2}\cdot \frac{\mu_1\theta_1}{\mu_1-\lambda_1} > \theta_1+\theta_2-1^\top \tau\\
    & \Longleftrightarrow \quad \frac{\lambda_2-\mu_1}{\lambda_2}\cdot \frac{\lambda_1\theta_1}{\mu_1-\lambda_1} + 1^\top \tau> \theta_2.
\end{align*}
This concludes that when $\theta_2\leq 1^\top \tau$: if $\lambda_2>\lambda_1$ and the above inequality holds, always serving queue $2$ is optimal; otherwise, $T_1=0$ and $T_2=\theta_1-1^\top \tau-c_2$ is optimal.

\item \textit{The case of $\theta_2\geq 1^\top \tau$.}
 Here, from the previous argument, we know $\tp(0) \geq \tp(0,T_2)$. Hence, we only need to compare the pure exhaustive service policy $\pe$ and the policy that always serves queue~2.
Suppose now that $\lambda_2 > \tp(0)$. 
Because $\tp(0)$ is strictly larger than $\lambda_1$ (see \eqref{eq:pf_twoqueue_tp_pure}), we must then have $\lambda_2 > \lambda_1$. Under this condition,
\begin{align}
    \lambda_2 >  \tp(0) \quad &\Longleftrightarrow \quad \lambda_2 > \lambda_1+\frac{\lambda_2\cdot\frac{\mu_2\theta_2}{\mu_2-\lambda_2}}{\frac{\mu_1}{\mu_1-\lambda_1}\cdot \left(1^\top \tau+\frac{\lambda_2\theta_2}{\mu_2-\lambda_2}\right)} \nonumber \\
    & \Longleftrightarrow \quad \frac{(\lambda_2-\lambda_1)\cdot \mu_1}{\mu_1-\lambda_1} > \frac{\lambda_2\cdot\frac{\mu_2\theta_2}{\mu_2-\lambda_2}}{1^\top \tau+\frac{\lambda_2\theta_2}{\mu_2-\lambda_2}} \nonumber\\
     & \Longleftrightarrow \quad  \theta_2 < \frac{\mu_1(\lambda_2-\lambda_1) (\mu_2-\lambda_2)}{\lambda_2\cdot \left[\mu_2(\mu_1-\lambda_1)-\mu_1(\lambda_2-\lambda_1)\right]}\cdot 1^\top \tau. \label{eq:pf_twoqueue_inequality_case2}
\end{align}
Since $\theta_2\geq 1^\top \tau$, to make \eqref{eq:pf_twoqueue_inequality_case2} valid, we must have 
\begin{align*}
 & \quad \frac{\mu_1(\lambda_2-\lambda_1) (\mu_2-\lambda_2)}{\lambda_2\cdot \left[\mu_2(\mu_1-\lambda_1)-\mu_1(\lambda_2-\lambda_1)\right]} >1\\
 \overset{(a)}{\Longleftrightarrow}& \quad
 \mu_1(\lambda_2-\lambda_1)(\mu_2-\lambda_2) > \lambda_2\Bigl[\mu_2(\mu_1-\lambda_1)-\mu_1(\lambda_2-\lambda_1)\Bigr] \\
\Longleftrightarrow &\quad
\mu_1\mu_2(\lambda_2-\lambda_1) - \mu_1\lambda_2(\lambda_2-\lambda_1) >\lambda_2\mu_2(\mu_1-\lambda_1) - \lambda_2\mu_1(\lambda_2-\lambda_1)\\
 \Longleftrightarrow & \quad
 \mu_1\mu_2(\lambda_2-\lambda_1) > \lambda_2\mu_2(\mu_1-\lambda_1) \\
\Longleftrightarrow &\quad
\mu_1<\lambda_2,
\end{align*}
where $(a)$ holds since the denominator $\lambda_2\cdot \left[\mu_2(\mu_1-\lambda_1)-\mu_1(\lambda_2-\lambda_1)\right]>0$; otherwise, \eqref{eq:pf_twoqueue_inequality_case2} cannot hold since the numerator $\mu_1(\lambda_2-\lambda_1)(\mu_2-\lambda_2)$ is positive because $\lambda_2>\lambda_1$ and $\mu_2>\lambda_2$.
Based on the above arguments and considering $\mu_1>\lambda_1$, we
conclude that when $\theta_2\geq 1^\top \tau$, if $\lambda_2>\mu_1$ and \eqref{eq:pf_twoqueue_inequality_case2} holds, always serving queue $2$ is optimal; otherwise, $\pe$ is an optimal exhaustive service policy.

\end{enumerate}

\smallskip

\noindent\underline{\textit{Case (iii): $\I(0)=\{2\}$}}
\smallskip

This is symmetric to Case (ii).

\smallskip

\noindent\underline{\textit{Case (iv): $\I(0)=\{1,2\}$}}
\smallskip

When $\I(0)=\{1,2\}=\N$, we have $\alpha(0)$ given by \eqref{eq:alpha_m}.
To make it an equilibrium, we need to guarantee $\alpha(0)\leq 1$. This leads to the condition in (iv). Since in this case, $\pe$ has achieved the first-best throughput, $\pe$ is one of the optimal exhaustive service policies. This completes the proof. \Halmos
\end{proof}

\medskip

\begin{proof}{\textbf{Proof of \Cref{coro:twoqueue_pureex_utilizationrates}.}}
    This is implied by \Cref{prop:twoqueues}. \Halmos
\end{proof}

\medskip

\begin{proof}{\textbf{Proof of \Cref{lem:inverse_mat}.}}
We show the case of $\I=\N$. Other cases are similar.
Let \(B\in \mathbb{R}^{n\times n} \) with $B_{ij}$ given by \eqref{eq:invermatrix_1}--\eqref{eq:invermatrix_2} under $\I=\N$. 
Observe that \eqref{eq:invermatrix_1}--\eqref{eq:invermatrix_2} are valid unless the denominator $\bar{\rho}_{\I}=1-1^\top \rho_{\I}=0$.
To conclude the proof, it is sufficient to show that \(A B = B A = \mathbb{I}\), where $\mathbb{I}$ is the identity matrix.
We provide the proof of $AB=\mathbb{I}$. The other side $BA=\mathbb{I}$ is similar.

Consider \((A B)_{ij} = \sum_{k} A_{ik}\,B_{kj}\). We split into two cases.

\begin{enumerate}
\item \textit{Case \(i = j\).}
We have
\begin{align*}
(A B)_{ii}
&=\;
\sum_{k} A_{ik}\,B_{ki}
\\
&=\;
A_{ii}\,B_{ii}
\;+\;
\sum_{k \neq i} A_{ik}\,B_{ki}.
\end{align*}
By definition, \(A_{ii} = \theta_i\) and \(A_{ik} = -\frac{\lambda_k\theta_k}{\mu_k-\lambda_k}\) for \(k \neq i\); and 
\[
B_{ii}
=
\frac{\mu_i - \lambda_i}{\mu_i\,\theta_i}\,
\frac{1}{\bar{\rho}}
\Bigl(1 - \sum_{\ell \neq i} \rho_{\ell}\Bigr),
\quad
B_{ki}
=
\frac{\mu_k - \lambda_k}{\mu_k\,\theta_k}\,
\frac{\rho_i}{\bar{\rho}}, \quad \text{for } k \neq i,
\]
where \(\bar{\rho} = 1 - 1^\top \rho\).  Hence
\begin{align*}
(A B)_{ii}
&=
\theta_i \cdot
\biggl[
  \frac{\mu_i - \lambda_i}{\mu_i\,\theta_i}\,\frac{1}{\bar{\rho}}
  \Bigl(1 - \sum_{\ell \neq i} \rho_\ell\Bigr)
\biggr]
\;+\;
\sum_{k \neq i} \bigl(-\frac{\lambda_k\theta_k}{\mu_k-\lambda_k}\bigr)\,\cdot
\biggl[
  \frac{\mu_k - \lambda_k}{\mu_k\,\theta_k}\,\frac{\rho_i}{\bar{\rho}}
\biggr]\\
& = \frac{1-\rho_i}{\bar{\rho}}
  \Bigl(1 - \sum_{\ell \neq i} \rho_\ell\Bigr)
 - \frac{\rho_i}{\bar{\rho}}\sum_{k\neq i}\rho_k \\
&\overset{(a)}{=}
\frac{1}{\bar{\rho}}\left[(1-\rho_i)(\bar{\rho}+\rho_i) - \rho_i (1-\bar{\rho}-\rho_i)\right]\\
& = \frac{1}{\bar{\rho}}\left[(1-\rho_i)\bar{\rho} + \rho_i\bar{\rho}\right]\\
& = 1,
\end{align*}
where $(a)$ holds since \(
1 - \sum_{\ell \neq i} \rho_\ell
= \bar{\rho} + \rho_i
\)
and
\(\sum_{k \neq i} \rho_k = (1 - \bar{\rho}) - \rho_i\).

\item \textit{Case \(i \neq j\).}
We have
\begin{align*}
(A B)_{ij}
&=
\sum_{k} A_{ik}\,B_{kj}
\\
&=
A_{ii}\,B_{ij}
\;+\;
A_{ij}\,B_{jj}
\;+\;
\sum_{k \neq i, k \neq j} A_{ik}\,B_{kj}\\
& = \theta_i \cdot \frac{\mu_i - \lambda_i}{\mu_i\,\theta_i}\,\frac{\rho_j}{\bar{\rho}}
+ \left(-\frac{\lambda_j\theta_j}{\mu_j-\lambda_j}\right)\cdot \frac{\mu_j - \lambda_j}{\mu_j\,\theta_j}\,
\frac{1}{\bar{\rho}}
\Bigl(1 - \sum_{\ell \neq j} \rho_\ell\Bigr)
+ \sum_{k\neq i,k\neq j} \left(-\frac{\lambda_k\theta_k}{\mu_k-\lambda_k}\right)\cdot \frac{\mu_k - \lambda_k}{\mu_k\,\theta_k}\,\frac{\rho_j}{\bar{\rho}} \\
& = \frac{\rho_j}{\bar{\rho}}(1-\rho_i) 
- \frac{\rho_j}{\bar{\rho}}\Bigl(1 - \sum_{\ell \neq j} \rho_\ell\Bigr) 
- \frac{\rho_j}{\bar{\rho}}\sum_{k\neq i,k\neq j}\rho_k\\
& = \frac{\rho_j}{\bar{\rho}} \left[
1-\rho_i - \Bigl(1 - \sum_{\ell \neq j} \rho_\ell\Bigr)  - \sum_{k\neq i,k\neq j}\rho_k
\right]\\
& = \frac{\rho_j}{\bar{\rho}}\left[
\sum_{\ell \neq j,\ell \neq i}\rho_\ell - \sum_{k\neq i,k\neq j}\rho_k
\right]\\
& = 0.
\end{align*}
\end{enumerate}

This confirms that the matrix composed by \eqref{eq:invermatrix_1}--\eqref{eq:invermatrix_2} with $\I=\N$ is indeed the inverse matrix $A^{-1}$.
Then, it can be easily shown that the sum of elements in the $i^{\textrm{th}}$ row of $A_{\I}^{-1}$ is $ \frac{\mu_i-\lambda_i}{\mu_i\theta_i}\cdot \frac{1}{\bar{\rho}_{\I}}.$
This completes the proof.
\Halmos
\end{proof}

\medskip

\section{Reward Maximization}
\label{app_sec:reward_maximization}

We now discuss how our results can be extended to the situation where serving customers from different queues generates different rewards, and the planner wants to maximize \textit{long-run average reward}, shorthanded as \textit{reward.}
Let $r_i>0$ be the reward to the planner per served customer of queue $i$.
Let $r=(r_i)_{i\in \N}$.
If the throughput of queue $i$ is $\lambda_i^\ast$, then the long-run average reward from queue $i$ is $r_i\cdot \lambda_i^\ast$.

First, observe that all equilibrium analysis results for both exogenous and endogenous regimes remain valid under reward maximization.
We now adapt our throughput maximization results to reward maximization.

Consider first \Cref{prop:LP_fixed_duration}, which provides an LP formulation for finding the throughput maximizing $\on$-$\off$ durations in the exogenous regime. For reward maximization, this proposition still holds with the following modification: replace the original objective function in problem \eqref{prob:LP_fixed_duration} with
\begin{align}
\nonumber
    \sum_{i \in \N} \left( \mu_i \cdot r_i x_i - (\mu_i - \lambda_i) \cdot r_i y_i \right).
\end{align}
This adjusted LP formulation yields the reward-maximizing $\on$-$\off$ durations.

\medskip

We now turn to the endogenous regime. Our first revised result characterizes the optimal exhaustive service policy, analogous to \Cref{thm:structure_opt_ex}.
\begin{theorem}[Structure of the Reward Maximizing Exhaustive Service Policy]
        \label{thm:structure_opt_ex_RM}
When maximizing the long-run average reward:
  \begin{enumerate}
  \item If $\I(0)=\N$, the pure exhaustive service policy $\pe$ achieves the first-best reward $r^\top \lambda$.
      \item If $\I(0)\neq \N$, for any $j\in\argmax_{i\in \bar{\I}(0)}r_i\lambda_i$, we can set $T^\ast_{-j}=0$ without loss of optimality. 
      Furthermore, if $\I(0)=\emptyset$, either the pure exhaustive service policy $\pe$
      or always serving the queue with the highest $r_i\lambda_i$ is optimal.
  \end{enumerate}
\end{theorem}

\Cref{thm:opt_ex_algo}, which presents a $2\bar{n}$-step algorithm for identifying throughput-maximizing exhaustive service policies, remains applicable in the context of reward maximization, subject to the following modifications to \Cref{alg:opt}:
\begin{enumerate}
    \item In Line 1: Store initial reward instead of throughput under the pure exhaustive service policy.
    \item In Line 3: Select $j \in \argmax_{i \in \bar{\I}(0)} r_i\lambda_i$.
    \item In Line 9: Compute reward instead of throughput.
    \item In Lines~10--15: Compare reward values rather than throughput values.
\end{enumerate}

\medskip

\begin{proof}{\textbf{Proof of \Cref{thm:structure_opt_ex_RM}.}}
We prove each case separately. For simplicity, let $\I(0)$ be written as $\I$. The proof is similar to the proof of \Cref{thm:structure_opt_ex} regarding the structure of the throughput-maximization exhaustive service policy.

\medskip

\noindent\underline{\textit{Proof of \Cref{thm:structure_opt_ex_RM} (i)}}

\smallskip

If $\I=\N$, then all customers join all queues. In this case, the policy $\pe$ achieves the first-best reward $r^\top \lambda$.

\medskip

\noindent\underline{\textit{Proof of the First Part of \Cref{thm:structure_opt_ex_RM} (ii)}}

\smallskip

We prove this in two steps. First, if $\I$ is non-empty, then $T_\I^\ast=0$. Second, $T^\ast_{\barI\setminus\{j\}}=0$ where $j\in\argmax_{i\in\barI}r_i\lambda_i$.

\smallskip

\noindent\underline{\textit{Step 1: $T_{\I}^\ast=0$ if $\I$ is non-empty}}

\smallskip

Let $\R(T_\J,T_{\barI})$ be the reward under the exhaustive service policy $(T_\J, T_{\I\setminus\J}=0, T_{\barI})$ for some $\J\subseteq \I$. Let $\alpha_{\HH}(T_\J,T_{\barI})\in [0,1]^n$ be the equilibrium variables under this policy. Similar to \eqref{eq:tp_ex_k}, we have
\begin{align*}
\R(T_\J,T_{\barI}) 
= \sum_{i\in\N}r_i\lambda_i 
- \sum_{i\in \barI\cup\HH}r_i\lambda_i 
\left(1 - \frac{c_i+\theta_i+T_i}{1^\top \tau 
+ c_{\I}^\top \alpha_{\I}(T_\J,T_{\barI}) 
+ 1^\top_\J T_\J + 1_{\barI}^\top T_{\barI} 
+ c_{\barI}^\top 1_{\barI}} \right),
\end{align*}
where $\HH\subseteq \I$ is the set with $\alpha_{\HH}(T_\J,T_{\barI})=1$. From the argument in the proof of \Cref{thm:structure_opt_ex}, both $\alpha_{\I}(T_\J,T_{\barI})$ and the cardinality of $\HH$ are non-decreasing in $T_\J$ (given any $T_{\barI}$). Hence, $\R(T_\J,T_{\barI})$ is non-decreasing in $T_\J$. Therefore, we can set $T_\J=0$ without loss of optimality.

\smallskip

\noindent\underline{\textit{Step 2: $T^\ast_{\barI\setminus\{j\}}=0$}}

\smallskip

We provide the proof for the case when $\I\subset \N$ is non-empty. The case when $\I$ is empty is analogous.
From step 1, we can set $T^\ast_\I=0$ without loss of optimality. Similar to \eqref{eq:tp_e_I}, the reward under $(0_{\I},T_{\barI})$ is
\begin{align}
\R(0_\I,T_{\barI}) 
= \left(\sum_{i\in \I\setminus\HH} r_i\lambda_i \right)
+ \left(\frac{\sum_{i\in\barI\setminus\HH}r_i\lambda_i(c_i+\theta_i)}{1^\top \tau
+ c^\top_{\I}\alpha_{\I}(0_\I,T_{\barI}) 
+ 1^\top_{\barI}T_{\barI} + c^\top_{\barI}1_{\barI}} \right)
+ \frac{\sum_{i\in\barI}r_i\lambda_iT_i}{1^\top \tau 
+ c^\top_{\I}\alpha_{\I}(0_\I,T_{\barI}) 
+ 1^\top_{\barI}T_{\barI}+c^\top_{\barI}1_{\barI}}.  \nonumber
\end{align}
By the proof of \Cref{thm:structure_opt_ex}, given $T_\I=0$, both $\alpha_\I(0_\I,T_{\barI})$ and $\HH$ depend only on $1^\top_{\barI}T_{\barI}$. Hence, the first two terms and the denominator of the third term depend only on $1^\top_{\barI}T_{\barI}$. Because
\[
\sum_{i\in\barI}r_i\lambda_iT_i 
\le \max_{i\in\barI}(r_i\lambda_i)\cdot \sum_{i\in\barI}T_i 
= r_j\lambda_j \cdot 1^\top_{\barI}T_{\barI}
\quad \text{for } j\in\argmax_{i\in\barI}r_i\lambda_i,
\]
the reward under $\tilde{T}=(\tilde{T}_j=1^\top_{\barI}T_{\barI},0_{\barI\setminus \{j\}})$ is (weakly) greater than $\R(0_\I,T_{\barI})$ for any $T_{\barI}$. Thus $T^\ast_{\barI\setminus\{j\}}=0$.

\medskip

\noindent\underline{\textit{Proof of the Second Part of \Cref{thm:structure_opt_ex_RM} (ii)}}

\smallskip

When $\I=\emptyset$, from the first part of \Cref{thm:structure_opt_ex_RM} (ii), it is without loss of optimality to set $T_{-j}=0$, where $j\in\argmax_{i\in\barI}r_i\lambda_i$. Setting $T_{-j}=0$, similar to \eqref{eq:derivative_tp_exhaustiveservice}, the partial derivative of the reward with respect to $T_j$ is
\[
\frac{\partial \R(T_j,T_{-j}=0)}{\partial T_j}
= \frac{r_j\lambda_j\cdot 1^\top \tau 
- \sum_{i\in\N}r_i(\mu_i-\lambda_i)c_i}
{(1^\top \tau + 1^\top c + T_j)^2},
\]
whose sign does not depend on $T_j$.
If the numerator is positive, the reward is always increasing in $T_j$, so it is optimal to serve queue $j$ indefinitely. Otherwise, the reward decreases with $T_j$, giving $T_j^\ast=0$, i.e., the pure exhaustive service policy $\pe$ is optimal. This completes the proof.\ \Halmos
\end{proof}

\end{APPENDICES}

\end{document}